\documentclass[sigconf, 10pt]{acmart}
\usepackage{graphicx}
\usepackage[english]{babel}
\usepackage[utf8]{inputenc}
\usepackage{blindtext}
\usepackage{afterpage}
\usepackage{amssymb} 
\usepackage{pgfplots}
\pgfplotsset{width=7cm,compat=newest}
\usepgfplotslibrary{statistics, fillbetween}
\setcopyright{none}
\renewcommand\footnotetextcopyrightpermission[1]{} 
\usepackage[outdir=figures/]{epstopdf}
\usepackage{enumitem}
\acmDOI{}

\acmISBN{}

\acmConference[SIGCOMM'19]{ACM Conference}{August 19-24, 2019}{Beijing, China}

\usepackage{url}

\usepackage{breakurl}

\acmPrice{}
\usepackage{appendix}
\usepackage[linesnumbered, boxed]{algorithm2e}
\usepackage{multirow}
\usepackage{subcaption}
\usepackage{tikz}
\usepackage{mathtools}
\usetikzlibrary{arrows, calc, pgfplots.groupplots, patterns,backgrounds, spy, shadings}

\theoremstyle{definition}
\newtheorem{definition}{Definition}[section]
\newtheorem{remark}{Remark}

\pgfplotsset{every axis/.append style={
	        ylabel near ticks,
	        xlabel near ticks,
	        ylabel style={font=\scriptsize},
	        xlabel style={font=\scriptsize},
                    xticklabel style={font=\tiny},
                    yticklabel style={font=\tiny},
            }}

\begin{document}
\newcommand{\muttoe}{{M{\scalebox{1.1}{u}}TT{\scalebox{1.1}{o}}E}}
\newcommand{\dittoe}{{D{\scalebox{1.1}{i}}TT{\scalebox{1.1}{o}}E}}
\newcommand{\metteor}{{METTEOR}}

\title[METTEOR]{\LARGE METTEOR: Robust Multi-Traffic Topology Engineering for Commercial Data Center Networks}


\author{Min Yee Teh}
\affiliation{%
  \institution{Columbia University}}

\author{Shizhen Zhao}
\affiliation{%
  \institution{Shanghai Jiao Tong University}}

\author{Keren Bergman}
\affiliation{%
  \institution{Columbia University}}


\begin{abstract}

Numerous optical circuit switched data center networks have been proposed over the past decade for higher capacity, though commercial adoption of these architectures have been minimal so far. One major challenge commonly facing these architectures is the difficulty of handling bursty traffic with optical circuit switches (OCS) with high switching latency. Prior works generally rely on fast-switching OCS prototypes to better react to traffic changes via frequent reconfigurations. This approach, unfortunately, adds further complexity to the control plane. 



We propose {\metteor}, an easily deployable solution for optical circuit switched data centers, that is designed for the current capabilities of commercial OCSs. Using multiple predicted traffic matrices, {\metteor} designs data center topologies that are less sensitive to traffic changes, thus eliminating the need of frequently reconfiguring OCSs upon traffic changes. Results based on extensive evaluations using production traces show that {\metteor } increases the percentage of direct-hop traffic by about 80\% over a fat tree at comparable cost, and by about 35\% over a uniform mesh, at comparable maximum link utilizations. Compared to ideal solutions that reconfigure OCSs on every traffic matrix, {\metteor } achieves close-to-optimal bandwidth utilization even with biweekly reconfiguration. This drastically lowers the controller and management complexity needed to perform {\metteor } in commercial settings.

\end{abstract}
\maketitle
\newcommand{\floor}[1]{\left\lfloor #1 \right\rfloor}
\newcommand{\ceil}[1]{\left\lceil #1 \right\rceil}
\vspace{-10pt}
\section{Introduction}\label{section_introduction}
Given the exponential growth in data center traffic, building networks that meet the requisite bandwidth has also become more challenging. Modern data center networks (DCN) typically employ multi-rooted tree topologies~\cite{leiserson1985fat}, which have a regular structure and redundant paths to support high availability. However, uniform multi-rooted trees are inherently suboptimal structures to carry highly skewed traffic common to DCNs~\cite{kandula2009flyways, roy2015facebook}. This has motivated several works on using optical circuit switches (OCS) to design more performant data center architectures~\cite{farrington2011helios, cthrough_wang_2011}. Compared to conventional electrical packet switches, OCSs offer much higher bandwidth and consumes less power. More importantly, OCSs introduces the possibility of Topology Engineering (ToE), which allows DCNs to dynamically allocate more capacity between ``hot spots’’ to alleviate congestion.

Despite showing immense promise, optical circuit-switched data centers have not been widely deployed even after a decade’s worth of research efforts. One of the most daunting challenges is to perform ToE under bursty traffic. Early works on ToE proposed reconfiguring topology preemptively using a single estimated traffic matrix (TM)~\cite{farrington2011helios, cthrough_wang_2011}. However, the bursty nature of DCN traffic makes forecasting TMs accurately very difficult~\cite{benson2010network, kandula2009nature}. An inaccurate prediction may lead to further congestion. Even if predictions were accurate, the forecast could still turn stale if topology reconfiguration takes tens of milliseconds. Subsequent works have thus focused on designing OCSs with microsecond-level switching latency~\cite{ProjectSirius, ghobadi2016projector, porter2013integrating, rotornet_mellette2017}, to enable faster reaction to traffic burst. However, these proposals require changing topology and routing frequently, an act that introduces significant complexity to the control plane, thus hindering the adoption by large vendors.



We tackle bursty DCN traffic from a different perspective, using a robust optimization-based ToE framework called { \metteor } (\textbf{M}ultiple \textbf{E}stimated \textbf{T}raffic \textbf{T}opology \textbf{E}ngineering for \textbf{O}ptimized \textbf{R}obustness). While prior works optimize topology for a \emph{single} estimated traffic matrix~\cite{cthrough_wang_2011, halperin2011augmenting}, our approach optimizes topology based on \emph{multiple} traffic matrices (TM). Traffic uncertainty is captured by a set of multiple TMs. Optimizing topology using this set helps desensitize the topology to traffic uncertainties. To our knowledge, {\metteor } is the first framework that tackles ToE from a robust optimization approach. 
The most compelling advantage of {\metteor } is that it does not rely on frequent OCS reconfiguration to handle traffic changes, as long as the new traffic is captured by a traffic set, thus reducing the management complexity in commercial data centers. In fact, {\metteor } shifts the major complexity of ToE from the system control aspect to the algorithm design aspect. Designing an optimal topology for multiple TMs is an immensely challenging problem~\cite{foerster2018characterizing, zhao2018minimal}. We first formalize the overall problem in \S\ref{section_model}, and discuss various techniques used for relaxing the algorithmic complexity in \S\ref{section_overall_methodology}.

We apply {\metteor } to the core layer of data centers. Based on traffic analysis of production data center traces, we found that while pod-level traffic do not exhibit strong temporal stability, they do exhibit a weaker form of temporal stability, which we refer to as traffic recurrence. This recurrent behavior in traffic leads to a slow-varying clustering effect, which is a novel observation in DCN traffic characteristics. By optimizing topology based on these slow-varying clusters can achieve great performance without frequent reconfiguration. Because of the low reconfiguration frequency, {\metteor } requires minimal changes to the data center control plane, and thus can be viewed as a first step towards fully optical circuit switched data centers. 

We evaluate {\metteor}’s performance under different routing algorithms that minimize maximum link utilization (MLU). Based on production data centers traces, {\metteor}  increases the percentage of direct-hop traffic by about 80\% over a fat tree at comparable cost, and by about 35\% over a uniform mesh, at comparable maximum link utilizations (MLU). (However, the tail MLU of {\metteor} may suffer if routing uncertainty exists.) Further, {\metteor } with ideal routing performs close to an idealized ToE that requires instantaneous switching and frequent reconfigurations. Note that using {\metteor}, we can obtain this level of performance with fortnightly OCS reconfiguration, making it deployable with the current off-the-shelf OCSs\footnote{To approximate ideal routing does require frequent routing update. Fortunately, routing update can be much easier than OCS reconfiguration.}. Moreover, {\metteor} is less dependent on the frequency of topology reconfigurations for good performance, when compared with the ToE solutions that optimize topology based on a single traffic matrix.

\begin{figure}[!tp]
\centering
\includegraphics[draft=false,width=2.5in, height=1.2in, trim={0.8cm 0.8cm 0.8cm 0.4cm}]{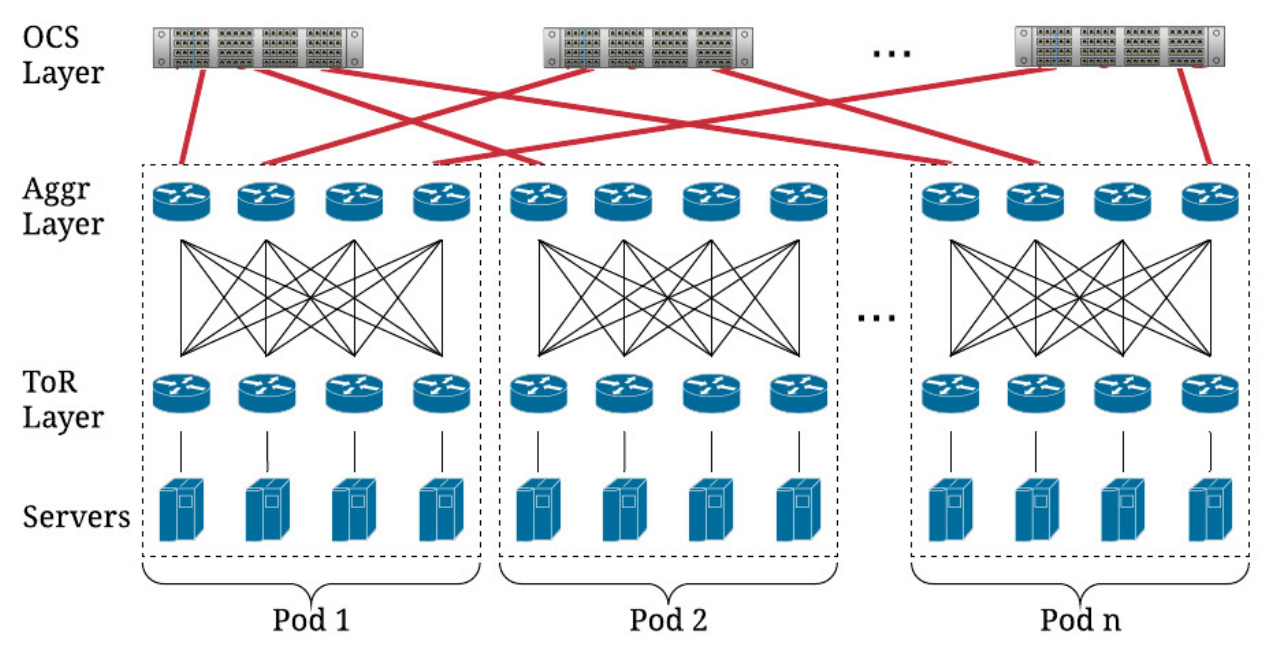}
\caption{\small Physical topology model, with pods fully-interconnected via OCSs at the core layer.}
\label{helios_topology}
\vspace{-18pt}
\end{figure}

\vspace{-10pt}
\section{Related Work}\label{section_related_works}
\subsection{Traffic-Agnostic DCN Topology}
DCN topologies have been traditionally designed to be static and traffic-agnostic, focusing on bisection bandwidth, scalability, failure resiliency, etc. They can be divided into either Clos-like and mesh-like topologies. Clos topology (e.g., Fat-Tree~\cite{al2008scalable, liu2013f10}) is more widely-adopted in large-scale data centers (e.g.,  Google~\cite{singh2015jupiter}, Facebook~\cite{farrington2013facebook}, Cisco~\cite{cisco2016}, and Microsoft~\cite{greenberg2009vl2}), as its regular hierarchical structure simplifies routing and congestion control. Mesh-like expander topology~\cite{singla2012jellyfish, valadarsky2015xpander, yu2016space} also shows great promise, as its flatter hierarchy saves cost by eliminating the spine layer in Clos, while still offering rich capacity and path diversity. 



However, DCN traffic is inherently skewed. A study from Microsoft~\cite{kandula2009flyways} showed that only a few top-of-rack (ToR) switches are ``hot'' in a small (1500-server) production data center. Facebook~\cite{roy2015facebook} reported that the inter-pod traffic in one of their data centers varies over more than seven orders of magnitude. As a result, traffic-agnostic networks can be inherently suboptimal under skewed DCN traffic.

\vspace{-9pt}
\subsection{Traffic-Aware DCN Topology}\label{section_reconfig_topo}
To handle fast-changing, high-skewed traffic patterns, some researchers have argued for reconfigurable DCN topologies based on optical circuit switches (OCS)~\cite{vahdat2011emerging, liu2010scaling, fields2010transceivers, zhou2017datacenter}. The pioneering work,  Helios~\cite{farrington2011helios}, proposed reconfiguring pod-to-pod topology using OCSs based on a \emph{single} estimated traffic matrix. However, reconfiguring Helios incurs a significant delay (about 30ms), a problem that most commercial OCSs today still face~\cite{calient}. Given that 50\% of DCN flows lasting below 10ms~\cite{kandula2009nature}, a 30ms reconfiguration latency could mean that the topology optimized for pre-switching traffic may no longer be a good fit for post-switching demands.

The need to cope with rapid traffic changes motivated subsequent works aimed at decreasing reconfiguration latency for OCSs. Some of these have focused on providing ToR-level reconfigurability~\cite{ reactor_liu2014, cthrough_wang_2011, singla2010proteus}, potentially reducing latency to microseconds level using sophisticated hardware.  However, these approaches might not scale to data centers with thousands of ToRs, due to the low radix of ToRs and the finite size of OCSs. Others have proposed scaling up reconfigurable networks with steerable wireless transceivers~\cite{ghobadi2016projector, hamedazimi2014firefly, zhou2012mirror}, but these architectures face serious deployment challenges related to environmental conditions in real DCNs, and to the need for sophisticated steering mechanisms.
The Opera architecture~\cite{mellette2019expanding}, built using rotor switches from~\cite{rotornet_mellette2017}, forms a mesh-like expander topology by multiplexing a set of preconfigured matchings in the time-domain. Unfortunately, frequently changing OCS connections may overload the SDN controller, and thus undermine data center availability.

Another line of work have looked into better algorithms that schedule circuits more optimally in the presence of reconfiguration delays~\cite{bojja2016costly, liu2015scheduling, wang2018neural}. However, the assumed problem setups of these works fundamentally differs from ours, as we are interested in designing a \emph{single} topology optimized for many possible traffic demands.


\vspace{-9pt}
\subsection{Traffic Engineering} \label{section_traffic_engineering_related_works}
To fully realize the potential of reconfigurable topologies, traffic engineering (TE), is required. TE typically consists of two phases: 1) the path-selection phase, and 2) the load-balancing phase. The path-selection phase selects a set of candidate paths for carrying traffic. Given a selection of paths, the load-balancing phase then computes the relative weights for sending traffic along the candidate paths. 

Path-selection in data centers typically employs the K-shortest-path algorithm~\cite{yen_ksp, singla2012jellyfish, valadarsky2015xpander}. As for load balancing, nearly all the related works~\cite{hamedazimi2014firefly, farrington2011helios, cthrough_wang_2011} on optical circuit-switched data centers compute the relative weights by solving a multi-commodity flow (MCF) problem using a single predicted traffic matrix. However, predicting a traffic matrix accurately can be difficult, and an inaccurate traffic prediction may incur unexpected congestion. Rotornet~\cite{rotornet_mellette2017} load-balances traffic using Valiant load-balancing (VLB)~\cite{zhang2008designing}. VLB has several desirable properties, such as being traffic-agnostic and robust under demand uncertainties by routing traffic via indirect paths, and having a worst-case throughput-reduction of 2$\times$. However, DCN operators tend to have a strong sense of what traffic patterns may likely occur, based on a wealth of historical traffic data. This makes VLB overly conservative. Some TE literatures use robust optimization to strike a balance between network performance and robustness to traffic uncertainty~\cite{wang2006cope, zhang2005optimal, chang2017robust}. Although these solutions are mainly designed for wide area networks (WAN), the core ideas are equally applicable to DCNs.

\section{Motivating \metteor}\label{section_proof_of_concept}


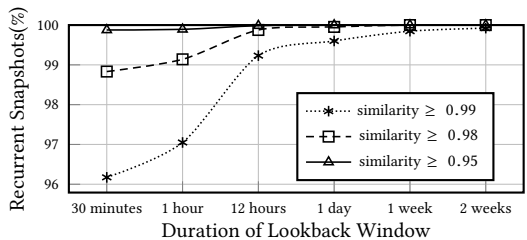
\begin{figure}[!tp]
\begin{tikzpicture}
\pgfplotstableread{plot_data/historical_similarity_ju2_dls03.txt} \datatable;
\begin{axis}[height=1.5in, width=0.90\columnwidth, ymax=1, grid=minor, ylabel near ticks, xlabel near ticks, ylabel={\footnotesize Recurrent Snapshots(\%)}, ylabel shift = -4pt, xlabel=\footnotesize Duration of Lookback Window, xlabel shift = -5pt, legend columns=1, legend style={at={(0.50,0.60)},anchor=north west, cells={align=left}}, xtick pos=left, ytick pos=left, xtick={1,2,3,4,5,6}, xticklabels={30 minutes, 1 hour, 12 hours, 1 day, 1 week, 2 weeks}, ylabel style={align=center}, ytick={0.96, 0.97, 0.98, 0.99, 1.00}, yticklabels={96, 97, 98, 99, 100}, grid, thick]
\addplot+[densely dotted, black, mark=asterisk, smooth, semithick, mark options=solid] table[x=x_val, y=threshold_0p99] from \datatable;
\addlegendentry{\tiny similarity $\geq \; 0.99$ }
\addplot+[densely dashed, black, mark=square, smooth, semithick, mark options=solid] table[x=x_val, y=threshold_0p98] from \datatable;
\addlegendentry{\tiny similarity $\geq \; 0.98$}
\addplot+[black, mark=triangle, smooth, semithick, mark options=solid] table[x=x_val, y=threshold_0p95] from \datatable;
\addlegendentry{\tiny similarity $\geq \; 0.95$}
\end{axis}
\end{tikzpicture}
\vspace{-13pt}
\caption{\small Percentage of TM snapshots with at least one historical ``lookalike’’. Two snapshots are considered ``lookalikes’’ if their cosine similarity exceeds a given threshold.}
\label{fig:historical_similarity}
\vspace{-18pt}
\end{figure}

\vspace{-6pt}
\subsection{Recurrence-A Weaker Form of Stability}
The conventional wisdom in ToE is to switch topology as frequently as possible to handle demand changes. The belief that DCN traffic lacks stability has driven much work on designing faster OCSs and control planes. However, DCN traffic is not entirely random, especially at the pod level. In fact, while pod-level traffic matrices (TM) do not generally exhibit strong stability over time, they do exhibit a weaker form of temporal stability, which we refer to as \emph{traffic recurrence}. This means that while most traffic snapshots may not be close to the snapshot preceding them, it is very likely that a similar TM has occurred in the recent past.

To quantify this phenomenon, we performed a simple case study on recurrence using 6 months’ worth of TM snapshots obtained from a data center. Each TM snapshot is a 5-minute-average of inter-pod traffic. We present results from 1 data center out of the 12 studied, though all other DCNs exhibit similar results. A TM snapshot is considered recurrent if it is close to at least one past TM within an observation period. The ``closeness’’ between two TMs is measured with cosine similarity~\cite{CosineSimilarity}. Fig. \ref{fig:historical_similarity} plots the percentages of recurrent TM snapshots as a function of the lookback window (i.e., how far back in the past we look). When closeness is loosely-defined (i.e. similarity $\geq$ 0.95), almost all snapshots are recurrent even with a 30-minute lookback window. When considering closeness as similarity $\geq$ 0.99, over 96\% of snapshots are recurrent within a 30-minute lookback window. Regardless of how closeness is defined, nearly all TMs are recurrent with a 2-week lookback window. This property of weak temporal stability may partially explain the slow-varying clustering effects in traffic patterns, which we explore in \S\ref{section_traffic_clustering}.

\subsection{Toy Example - {\metteor}}




Fig. \ref{fig:mttoe_concept} shows a proof-of-concept for {\metteor}. Clearly, no single TM can adequately represent all TMs properly in this case, so single-traffic-based ToE approaches, as in~\cite{hamedazimi2014firefly, farrington2011helios, cthrough_wang_2011}, may not work well. Our approach accounts for traffic uncertainty by optimizing topologies based on multiple TMs. When traffic is recurrent, many observed TMs will likely reappear in the future. With topologies optimized for a few representative TMs derived from historical snapshots, {\metteor} could perform well for future recurring traffic.


Using a simple experiment, we motivate the use of {\metteor}. In this example, we consider a network with 8 pods interconnected via an OCS layer in a manner similar to that in Fig. \ref{helios_topology}; each pod has 100 directed links of unit capacity. We generate 30 traffic matrices at random. For {\metteor}, we find 3 traffic centroids using $\kappa$-means clustering algorithm, and optimize topology based on the 3 TMs. For comparison, we also optimize topology based on the average of all TMs. Then, for each TM, we compute the maximum link utilization (MLU) of routing each TM over the two topologies. 


Fig.~\ref{fig:toy_example_timeseries} shows the MLU performance. Cearly, {\metteor} performs better, as it is able to design topology that is well-suited for most of the traffic snapshots. Under traffic uncertainties, a multi-traffic optimization approach may improve solution robustness by minimizing topology-overfitting to a single predicted demand.

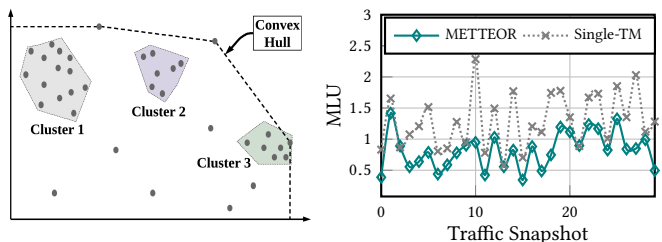
\begin{figure}[!tp]
\pgfplotstableread{plot_data/multi_traffic_toe_motivation.txt} \datatable
\hspace{-8pt}
\begin{subfigure}[c]{0.47\columnwidth}
\hspace{-8pt}
\includegraphics[draft=false,width=1.05\linewidth,  height=3.0cm, trim={0cm 0cm 0cm 0.cm}]{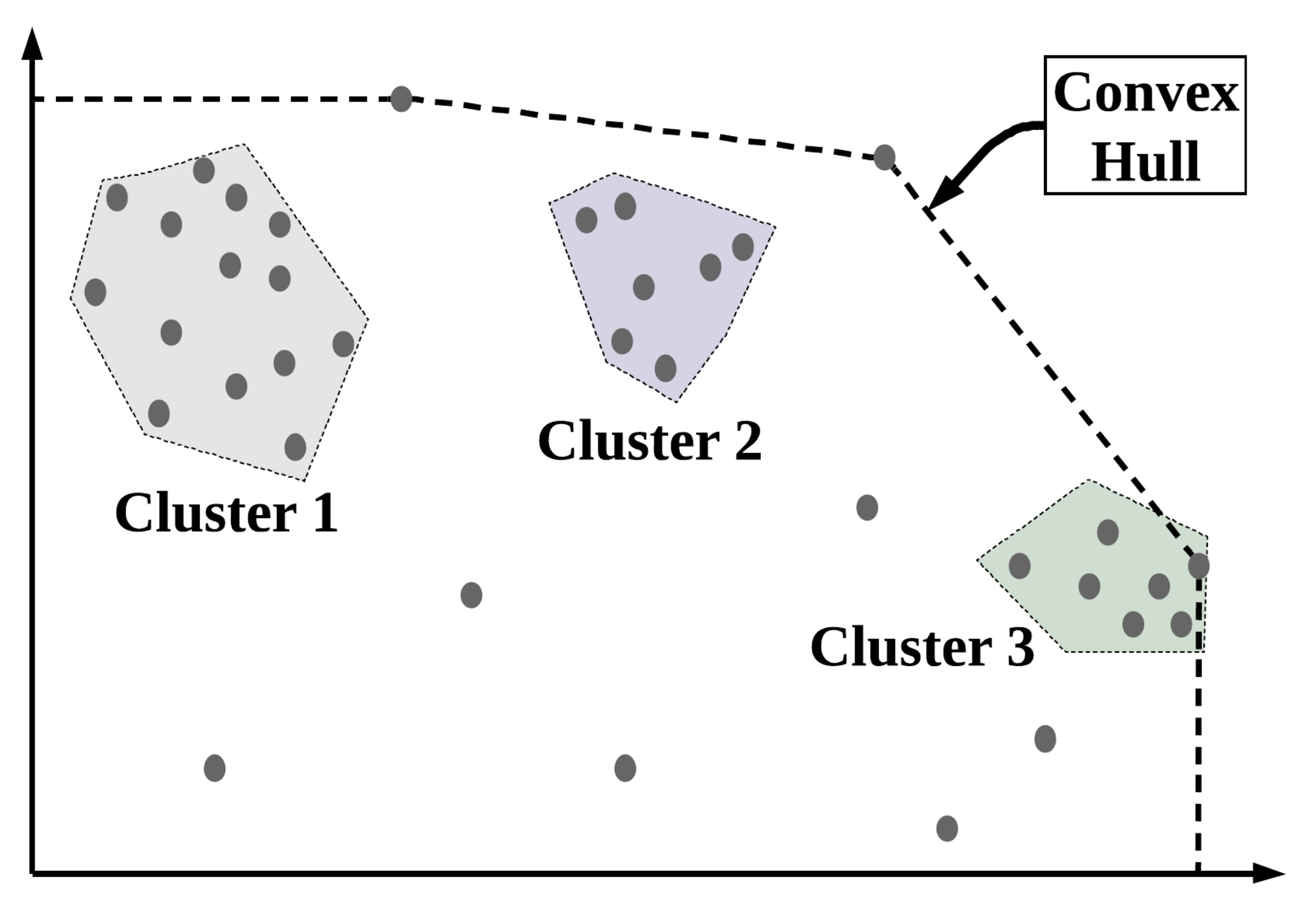}
\caption{\small Multi-traffic ToE concept}
\label{fig:mttoe_concept}
\end{subfigure}
~
\hspace{-10pt}
\begin{subfigure}[c]{0.47\columnwidth}
\hspace{-10pt}
\begin{tikzpicture}
\begin{axis}[xlabel = \footnotesize Traffic Snapshot, ylabel = \footnotesize MLU, ylabel near ticks, width=1.31\linewidth, height=4.0cm, ylabel shift = -4pt, xlabel shift = -4pt, xtick pos=left, ytick pos=left, grid, thick, smooth, xmin=0,xmax=29, legend columns=2, legend style={at={(0.01, 0.99)},anchor=north west,}, ymax=3, ytick={0, 0.5, 1, 1.5, 2, 2.5, 3}]
\addplot+[color=teal, mark=diamond, mark options={solid}, mark repeat={1}] table[x=x, y=robust_toe] from \datatable ;
\addlegendentry{\tiny \metteor}
\addplot+[color=gray, mark=x, mark options={solid}, mark repeat={1}, densely dotted] table[x=x, y=ave_toe] from \datatable ;
\addlegendentry{\tiny Single-TM}
\end{axis}
\end{tikzpicture}
\vspace{-20pt}
\caption{\small MLU timeseries}
\label{fig:toy_example_timeseries}
\end{subfigure}
\vspace{-14pt}
\caption{\small METTEOR concept illustration.} 
\label{fig:proof_of_concept}
\vspace{-21pt}
\end{figure}

\vspace{-4pt}
\section{METTEOR - System Level Overview}\label{section_challenges_overall_approach}

\begin{figure*}[htb]
\centering
\includegraphics[trim={0.2cm 0.7cm 0.2cm 0.7cm}, width=0.82\textwidth, height=2.1cm] {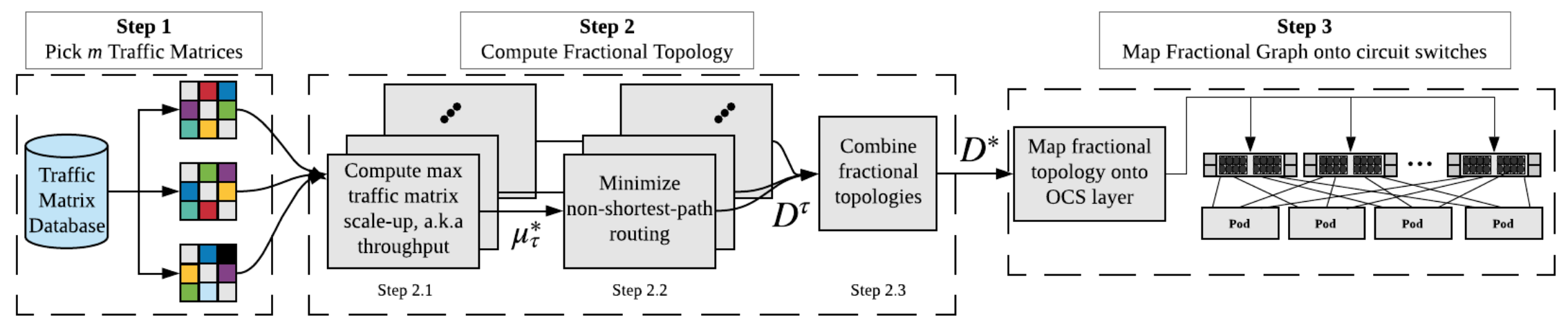}
\caption{\small Illustrating the complete software workflow of {\metteor}.}
\vspace{-11pt} 
\label{topology_engineering_overall_workflow}
\end{figure*}

\subsection{Network Architecture}\label{section_network_architecture}
The assumed DCN topology is shown in Fig. \ref{helios_topology}, with a layer of OCSs interconnecting all pods, each constructed from packet switches. This topology resembles a Clos typically seen in large scale data centers, although we replace the core switches with OCSs. Like Helios~\cite{farrington2011helios}, our work employs inter-pod reconfigurability, which deviates from some architectures that argue for inter-ToR reconfigurability~\cite{ghobadi2016projector, hamedazimi2014firefly, porter2013integrating, cthrough_wang_2011, zhou2012mirror}. We chose inter-pod reconfigurability over inter-ToR reconfigurability for the following reasons: 
\begin{itemize}[leftmargin=3pt]
\item \emph{Scalability: } - Using pods with hundreds of uplinks to the OCSs, and downlinks to ToRs for $\Theta(1000)$ servers, our architecture can scale up to over 100k servers.  

\item \emph{Traffic stability: } Inter-pod traffic shows more noticeable locality~\cite{roy2015facebook}, and is more stable than inter-ToR traffic~\cite{delimitrou2012echo, kandula2009nature} due to averaging effects from the aggregation switches. 

\item \emph{High fan-out\footnote{Ability to form direct links with many destinations,.}: } Pods have much higher fan-out than ToRs. Combined with multi-hop routing, every pod is reachable within one or two hops, making it possible for \emph{one} logical topology to serve several, possibly dense\footnote{While inter-ToR traffic matrices is quite sparse~\cite{ghobadi2016projector}, inter-pod traffic matrices tend to be dense, with mostly non-zero entries.}, TMs.

\end{itemize}

\vspace{-3pt}
In this paper, we refer to the (fixed) physical connections between the pod and OCSs as the \emph{physical topology}.  Topology engineering reconfigures the OCSs to realize a specific \emph{logical topology} as an overlay on the physical topology. 

\vspace{-4pt}
\subsection{Computing Logical Topology}\label{section_computing_logical_topology}
Prior works have designed reconfigurable topology based on a \emph{single} estimated traffic matrix, obtained either from switch measurements (e.g. Hedera~\cite{al2010hedera}) or from end-host buffer occupancy~\cite{cthrough_wang_2011}. However, due to the bursty nature of DCN traffic \cite{kandula2009nature}, even inter-pod traffic can be difficult to predict accurately, which fundamentally limits the robustness of such an approach.


Therefore, we compute logical topologies based on multiple TMs instead. The first step is to obtain multiple TMs that will be representative of future traffic (see Step 1 of Fig. \ref{topology_engineering_overall_workflow}), based on historical traffic snapshots. Traffic snapshots can be easily obtained from flow-monitoring tools like sFlow~\cite{phaal2001inmon}. While we could get an accurate traffic estimation directly from applications, this would require application level modifications. Instead, we employ a simpler approach that exploits the spatial-temporal traffic behavior of production traffic to extract multiple representative TMs (see \S\ref{section_traffic_clustering}). The next step is to optimize topology for the extracted TMs (see  \S\ref{section_model} and \S\ref{section_overall_methodology}), which is the biggest challenge of this paper. In fact, the topology optimization problem for even a single TM is already NP-complete; having multiple TMs further complicates this problem. Our goal is to design a polynomial-time heuristic to this problem. The algorithm design must be done carefully. Otherwise, a poorly-design topology can easily nullify the potential benefits of ToE.

\vspace{-10pt}
\subsection{Reconfiguring Logical Topology Safely}
Despite having shown great promise on paper, ToE has not seen widespread commercial adoption. One key reason is that existing reconfigurable architectures do not consider high network availability. Network availability is generally defined as a high-level service level objective (SLO), measured as a number of ``nines’’ in service uptime~\cite{govindan2016evolve, hong2018b4}. Under the hood, however, availability is inextricably linked to factors like traffic volume, controller workload, hard/software failure rates, and packet loss~\cite{mogul2017thinking}. 

Performing ToE frequently, if not done properly, could be detrimental to availability. For instance, high-frequency switching places a tremendous workload on the SDN controller. A poor-choice of switching configuration, or even a bug, risks failing entire DCN blocks; an admittedly rare risk, but one that increases with the rate of reconfiguration.

There are two major considerations when reconfiguring topology. First, reconfiguration must be carefully sequenced to avoid routing packets into ``black holes.''  For each reconfiguration event, the SDN controller must first ``drain’’ links by informing packet switches not to route traffic through the optical links that are about to be switched. Only upon verifying that no traffic flows through these links can physical switching take place. After switching completes, the SDN controller can then ``undrain’’ links and start sending traffic through them again. 

Second, topology reconfiguration needs to be staged to maintain sufficient network capacity, especially when traffic demands are high. For instance, if 40\% of links need to be reconfigured when network utilization is at 80\%, the reconfiguration process must take at least 2 stages (switching 20\% of links in each stage) to avoid congestion and packet loss due to over-utilization. 



\subsection{Bootstrapping Greenfield DCNs}

{ \metteor } requires a sufficient history of TMs to find the right clusters. However, when a greenfield DCN is initially deployed, or when new pods are added during DCN expansion, there are not sufficient traffic data to locate the correct traffic clusters. So, the initial configuration should aim for a uniform logical topology, and route traffic evenly along both direct and indirect  paths. This reduces the risks of maximum congestion due to traffic bursts, at the cost of poor bandwidth tax performance as most traffic will traverse indirect paths. Once sufficient historical traffic measurement is available, then {\metteor } can be triggered. Based on our experience, one week’s worth of traffic snapshots should suffice.

\vspace{-6pt}
\section{Formalizing METTEOR}\label{section_model}
We now formalize the mathematics of {\metteor}. All notations are tabulated in Table \ref{table:notations}.

\vspace{-5pt}
\subsection{Logical Topology}
Let $\mathcal{S}=\{s_1,..,s_n\}$ be the set of pods, $\mathcal{O}=\{o_1,..,o_y\}$ be the set of OCSs, and $x_{ij}^k$ be the number of links from pod $s_i$ to pod $s_j$ through OCS $o_k$. We represent a logical topology using $X=[x_{ij}], i,j=1,...,n$, where $x_{ij} = \sum_{k=1}^y  x_{ij}^k $ is the number of links between pods $s_i$ and $s_j$. The logical topology $X$ \textbf{must} be feasible under a given physical topology, so it must satisfy the following group of constraints.\\
\textbf{OCS-level (Hard) Physical Constraints}: 
\begin{equation}\label{constraint:ocslevel}
\left\{
\begin{split}
&\sum_{j=1}^n x_{ji}^k \leq h_{\text{ig}}^k(i), \sum_{j=1}^n x_{ij}^k \leq h_{\text{eg}}^k(i),\; \forall  i=1,..,n, k = 1,..,y;&\\
&x_{ij} = \sum_{k=1}^y  x_{ij}^k, \hspace{2mm}x_{ij}\text{ and }x_{ij}^k\text{ are all integers;}&
\end{split}
\right.    
\end{equation}
where $h_{\text{ig}}^k(i), h_{\text{eg}}^k(i)$ are the number of ingress/egress links of pod $s_i$ through OCS $o_k$. 
\vspace{-5pt}
\subsection{Network Throughput}
Let $T= [t_{ij}], i,j=1,...,n$ be a TM, where $t_{ij}$ is the traffic rate (in Gbps) from pod $s_i$ to pod $s_j$. Given a logical topology $X$, we measure its throughput $\mu$ \emph{w.r.t.} $T$, such that $\mu T$ is the maximum scaled TM that can be feasibly routed over the topology $X$. Routing feasibility is defined as follows. 

Let  $\mathcal{P}= \cup_{(i,j)} \mathcal{P}_{ij}$ be the candidate path set for all pod pairs $(s_i, s_j)$, where $\mathcal{P}_{ij}$ is the set of candidate paths from $s_i$ to $s_j$. We allow no more than two hops between pods, so $\mathcal{P}_{ij}=\{[s_i, s_j], [s_i, s_1, s_j], …, [s_i, s_n, s_j]\}$. The feasibility of routing $T$ over $X$ can be verified using:
\begin{eqnarray}\label{feasible_constraints}
&& \text{Find } \Omega=\{\omega_p\}, p\in \mathcal{P} \text{ such that }\\
&& 1) \sum_{p \in \mathcal{P}_{ij}} \omega_p = t_{ij}, \;  \forall \; i, j=1,...,n \nonumber\\
&& 2) \sum_{p \in \mathcal{P}, (s_i, s_j) \text{ is a link in }p} \omega_p \leq x_{ij} b_{ij}, \; \forall \; i, j=1,...,n\nonumber
\end{eqnarray}
where $b_{ij}$ is the link capacity between $s_i$ and $s_j$, and $\omega_p$ is the traffic routed (in Gbps) via path $p$. 

When computing throughput, we scale $T$ until max link utilization (MLU) hits 1, where link utilization is the ratio of a link’s traffic flow rate to its capacity. As it turns out, this problem (\ref{feasible_constraints}) is related to that of minimizing max link utilization (MLU) when routing an unscaled $T$ over $X$. Thus, a lower MLU implies that there is more room for $T$ to grow before MLU hits 1, which leads to higher throughput.

\vspace{-8pt}
\subsection{Design Objective}
Given $m$ traffic matrices, $\{T_1, .. , T_m\}$, let $\{\mu_1, .. , \mu_m\}$ be the throughputs of routing $\{T_1, .. , T_m\}$ over $X$. We aim to design $X$ such that $\min (\mu_1, .. , \mu_m)$ is maximized: 
\vspace{-5pt}
\begin{eqnarray}\label{overall_formulation}
&& \max_{X}\mu=\min\{\mu_1, .. ,\mu_m\} \text{, s. t}\\
&& 1) \text{ } X\text{ is an integer matrix that satisfies (\ref{constraint:ocslevel})} \nonumber\\
&& 2) \text{ } (X, \mu_\tau T_\tau)\text{ satisfies (\ref{feasible_constraints})}, \; \forall \; \tau \in \{1, .., m\} \nonumber\\
&& 3) \text{ The majority of traffic in }T_\tau\text{ is routed through}\nonumber\\
&&\hspace{4mm}\text{direct paths in }X, \; \forall \; \tau \in \{1, .., m\} \nonumber
\end{eqnarray}
Note that (\ref{overall_formulation})’s formulation ensures that the logical topology maximizes all TM throughputs as evenly as possible. Although we could maximize the total throughput of all TMs, we avoid this as it gives the logical topology freedom to selectively-optimize the throughputs of the “easier” TMs.


Solving (\ref{overall_formulation}) gives us the optimal logical topology. However, the runtime complexity scales exponentially with the number of pods and OCSs, which is too challenging for commercial solvers like Gurobi~\cite{gurobi}. Some prior work has studied traffic engineering (TE) techniques based on multiple TMs~\cite{zhang2005optimal, zhang2014load}. Unfortunately, those techniques cannot be applied here, as ToE, unlike TE, requires integer solutions. 

The complexity of (\ref{overall_formulation}) is imposed by the structure of the physical topology.  Since OCSs have limited radix, and the OCS layer may involve $\sim 10k$ links, this layer must use multiple OCSs. Using multiple OCSs, rather than one giant OCS, makes this optimization a strongly NP-complete combinatorics problem~\cite{foerster2018characterizing, zhao2018minimal}. Since tackling (\ref{overall_formulation}) head on is infeasible, we split the overall problem into smaller subproblems. 



\label{section_mathematical_model}
\begin{table}[!t]
\footnotesize
\begin{tabular}{|p{2.0cm}|p{6.06cm}|}
\hline
$\mathcal{S} = \{s_1, .., s_n\}$ & Set of all $n$ pods \\
\hline
$\mathcal{O} = \{o_1, .., o_y\}$& Set of all $y$ circuit switches\\
\hline
$x_{ij}^k$ & Integer number of pod $i$'s egress links connected to pod $j$'s ingress links through $o_k$\\ 
\hline
$X = [x_{ij}] \in \mathbb{N}^{n\times n}$ & Inter-pod topology; $x_{ij}$ denotes the number (integer) of $s_i$ egress links connected to ingress links of $s_j$\\ 
\hline
$T = [t_{ij}] \in \mathbb{R}^{n\times n}$ & Traffic matrix, where $t_{ij}$ denotes the traffic rate (Gbps) sent from $s_i$ to $s_j$\\ 
\hline
$D = [d_{ij}] \in \mathbb{R}^{n\times n}$ & Fractional topology; $d_{ij}$ denotes the number (fractional) of $s_i$ egress links connected to ingress links of $s_j$\\ 
\hline
$h_{eg}^k(s_i), h_{ig}^k(s_i)$ & Number of physical egress and ingress links, respectively, connecting $s_i$ to $o_k$\\
\hline
$r_{\text{eg}}^i, r_{\text{ig}}^i$ & Number of egress and ingress links, respectively, of $s_i$\\
\hline
$b_{ij}$ & Link capacity (Gbps) between $s_i$ and $s_j$\\
\hline
$\mathcal{P}_{ij}$ & Set of all routing paths from $s_i$ to $s_j$ \\
\hline
$\omega_{p}$ & Traffic (Gbps) on path $p$ \\
\hline
$\mu$ & Traffic scale-up factor\\
\hline
\end{tabular}
\caption{\small Notations used in this paper} 
\label{table:notations}
\vspace{-31pt}
\end{table}

\section{Overall Methodology}\label{section_overall_methodology}
Next, we discuss the techniques employed to sidestep the complexity of (\ref{overall_formulation}). Specifically, we split the overall problem into steps 2 and 3 of Fig. \ref{topology_engineering_overall_workflow}.

First, we design a \emph{fractional} logical topology (Step 2 in Fig.  \ref{topology_engineering_overall_workflow}) that optimizes throughput for all TMs, instead of computing an integer solution directly. Without the integer requirement, this step can be solved using linear programming (LP). Next, we configure the OCSs such that the \emph{integer} logical topology best approximates the fractional topology (Step 3 of Fig. \ref{topology_engineering_overall_workflow}). These steps are detailed in \S\ref{section_fractional_topology} and \S\ref{section_circuit_switch_mapping}.

\vspace{-6pt}
\subsection{Computing Fractional Topology}\label{section_fractional_topology}
Before proceeding, we need to define \emph{fractional topology}.
\begin{definition}
Given a set of pods $\mathcal{S}=\{s_1, …, s_n\}$ and the number of ingress \& egress links $r_{\text{ig}}^i, r_{\text{eg}}^i$, $D = [d_{ij}] \in \mathbb{R}^{n \times n}$ is a fractional topology \emph{iff} it satisfies:
\vspace{-3pt}
\begin{equation}
\sum\limits_{j=1}^n d_{ij} \leq r_{\text{eg}}^i, \sum\limits_{i=1}^n d_{ij} \leq r_{\text{ig}}^{j} \; \quad\forall \; i,j = 1, ..., n
\label{constraint:egressandingress}
\end{equation}
\label{defn:fractional_graph}
\end{definition}
\vspace{-7pt}
A fractional topology, $D$, simply describes the inter-pod (fractional) link count, where each pod’s in/out-degree constraints are satisfied. This definition noticeably ignores the OCSs; since the OCS layer will be considered when rounding the fractional topology into an integer logical topology, accounting for them here unnecessarily increases the number of variables needed for representation\footnote{As the number of ports of an OCS and a pod is comparable, the total number of OCSs, $y$, must be in the same order as the total number of pods $n$. Factoring in the OCS layer increases the variable space from $O(n^2)$ to $O(n^2y)$, causing our solver to run out of memory for large fabrics.}. 

Our goal is to design a fractional topology, $D$, that leads to good throughput for all the input TMs. Initially, we formulated an LP that computes the optimal $D$ for all TMs:
\begin{equation}\label{eqn:single_large_lp}
\max_{D\text{ satisfies (\ref{constraint:egressandingress})}}\mu,\hspace{4mm}\text{s.t.} (D, \mu T_\tau)\text{ satisfies (\ref{feasible_constraints})},  \forall \tau \in \{1, .., m\}
\end{equation}
However, the above formulation scales badly due to the large number of constraints when considering multiple TMs in one LP. To achieve scalability, we use a two-step approach: 1) compute the optimal fractional topology for every TM, and 2) combine the fractional topologies into one.

\subsubsection{Fractional Topology for One Traffic Matrix} \hfill\\
We first compute a fractional topology for a single TM based on two routing metrics: throughput and average hop count. However, there is a tradeoff between these two metrics under a given topology. For instance, throughput may be increased if we allow non-shortest-path routing, but this can increase hop count. We want to find a fractional topology that gives a Pareto-optimal tradeoff between these metrics.

Given a TM, $T$, and a set of candidate paths, $\mathcal{P}$, we compute a fractional topology $D$ in two steps. First, we compute $D$ that maximizes throughput $\mu$ for $T$ as follows:
\begin{equation}\label{fractional_graph_lp}
\max\limits_{\mu, D} \;  \mu \;\quad\text{s. t: } (D, \mu T)\text{ satisfies }(\ref{feasible_constraints}).
\end{equation}
Let $\mu^*$ be the optimal value of (\ref{fractional_graph_lp}). There could be many fractional topologies that maximize throughput $\mu^*$ for $T$. We select the one that minimizes the usage of the non-shortest paths. Let $\mathcal{P}_{ij}^{\prime}\subset \mathcal{P}_{ij}$ be the set of non-shortest paths in $\mathcal{P}_{ij}$. The formulation is as follows:
\begin{equation}\label{fractional_graph_qp}
\min\limits_{\Omega, D} \; \sum_{p \in \mathcal{P}^{\prime}}\big(\omega_{p}\big)^2  \;\quad\text{s. t: } (D, \mu^* T)\text{ satisfies }(\ref{feasible_constraints}).
\end{equation}

Note that average hop count can be reduced implicitly by minimizing the routing weights of non-shortest paths, thus solving (\ref{fractional_graph_qp}) helps $D$ meet the third requirement in the formulation (\ref{overall_formulation}). We opted for a quadratic objective function in (\ref{fractional_graph_qp}) over a linear one due to its ``sharper’’ landscape, which helps desensitize solution to slight TM input variations. {\bf In our evaluation, we found this step instrumental in increasing the amount of direct-hop traffic overall.}

\subsubsection{Combining Fractional Topologies}  \hfill\\
Having computed $D^{\tau}$ for every $T^{\tau},\tau=1,..., m$, we then linearly combine them into one, $D^*$, which is then used to map onto the OCS layer. This can be formulated as:
\begin{equation}\label{combine_fractional_graphs}
\max_{\alpha > 0, D^*}\alpha, \text{ s. t: } d_{ij}^* \geq \; \alpha \; d_{ij}^\tau, \; \forall \; i, j =1,..,n\text{ and }\tau=1,..,m.
\end{equation}
The constraint in (\ref{combine_fractional_graphs}) guarantees that the throughput of routing $T^{\tau}$ in the combined fractional topology $D^*$ is at least $\alpha$ times of that of routing  $T^{\tau}$ in $D^{\tau}$. By maximizing $\alpha$, $D^*$ achieves a good balance among different fractional topologies in terms of throughput.

\subsection{Mapping $D^*$ onto the OCS Layer}\label{section_circuit_switch_mapping}
We now map $D^*$ onto the OCSs such that the integer logical topology, $X$, best approximates $D^*$.

\vspace{-2pt}
\subsubsection{Problem Setup} \hfill\\
The goal here is to decide the total number of links $x_{ij}^k$ from pod $s_i$ to pod $s_j$ through OCS $o_k$, for every $i,j=1,2,...,n$ and $k=1,2,...,y$. Since there are $y$ OCSs, we split each $d_{ij}^*$ entry in $D^*$ into $y$ integers $x_{ij}^k, k=1,...,y$, such that $\sum_{k = 1}^{y} x_{ij}^k \approx d_{ij}$, which can be formulated as

\noindent \textbf{Soft / Matching Constraints }
\begin{equation}\label{con_constraint}
    \floor{d_{ij}^*} \leq \sum_{k=1}^y x_{ij}^k \leq \ceil{d_{ij}^*}, \quad \forall \; i,j=1,...,n
\end{equation}
Then, a logical topology can be found by solving
\begin{equation}\label{map_to_ocs}
\text{Find }\{x_{ij}^k\}\text{ satisfying (\ref{constraint:ocslevel}) and (\ref{con_constraint}).}
\end{equation} 


Solving (\ref{map_to_ocs}) strictly is NP-Complete, as 3-Dimensional Contingency Table problem, proven to be NP-Complete \cite{irving1994three}, can be reduced to our problem. Fortunately, unlike the physical OCS constraint (\ref{constraint:ocslevel}), constraint (\ref{con_constraint}) is ``soft’’, which can be relaxed to reduce algorithmic complexity.


Initially, we tried two natural ideas for solving (\ref{map_to_ocs}). The first is to solve it directly using ILP. This naive approach has an extremely high runtime complexity; it also cannot gracefully relax the soft constraints when satisfying (\ref{con_constraint}) is infeasible. The second idea is to employ a greedy maximum matching as in Helios~\cite{farrington2011helios}, which maps each OCS to a max-weight matching subproblem based on $D^*$, and then greedily solves these subproblems. However, this greedy approach could violate so many soft constraints, such that the resulting logical topology $X$ may no longer be a good estimate of $D^*$, causing poor network performance.

We wanted an approach that (a) has low complexity, (b) is mathematically sound, and (c) can \emph{gracefully} relax soft constraints when necessary. The ILP approach only achieves (b), while the greedy algorithm only achieves (a). Inspired by convex optimization theories,  we developed two algorithms that achieve all three criteria.

\subsubsection{Algorithm Intuition}\label{section_barrier_penalty}\hfill\\
One standard approach for relaxing hard-to-satisfy constraints is the Barrier Penalty Method (BPM)~\cite{ConvexOptimization}. The idea is to transfer all the soft constraints into an objective function $U(X)$ that penalizes soft constraint-violation:
\vspace{-5pt}
\begin{equation}\label{eqn:barrier_objective}
\min\limits_{X \text{ satisfies }(\ref{constraint:ocslevel})} U(X) = \sum_{i=1}^n\sum_{j=1}^n\Big[\big(\sum_{k =1}^y x_{ij}^k - \floor{d_{ij}^*}\big) \big(\sum_{k=1}^y x_{ij}^k - \ceil{d_{ij}^*}\big)\Big]
\end{equation}

Since $x_{ij}^k$ are all integers, it is easy to verify that $U(X)\geq 0$ and $U(X)=0$ \emph{iff} all the soft constraints in (\ref{con_constraint}) are satisfied. When (\ref{con_constraint}) is not satisfiable, minimizing $U(X)$ provides a graceful relaxation of (\ref{con_constraint}). Still, computing an integer $X$ directly requires exponential runtime. We instead compute $x_{ij}^k$ for each OCS $k$ iteratively, while keeping $x_{ij}^k$ for all other OCSs constant. Using first order approximation, computing $x_{ij}^k$ for a given OCS $k$ can be mapped to a min-cost flow problem, which can be solved in polynomial time~\cite{Edmonds1972Theoretical}. Unfortunately, since BPM weighs every $(i, j)$ soft constraint equally in its objective function, we found that BPM suffers from an increased soft constraints violation when $D^*$ is skewed. This finding highlights the BPM’s limited adaptability to a wide range of fractional topologies.

To address the shortcomings of BPM, we employ another approach based on the Lagrangian dual method (LDM)~\cite{low1999optimization}. LDM relaxes the soft constraints (\ref{con_constraint}) with dual variables that can adapt to the skewness of $D^*$ over multiple iterations, leading to fewer soft constraint violations. Our {\metteor } implementation uses LDM precisely for its adaptability. The full derivations of the LDM and the BPM are in \ref{appendix:barrier_penalty_detailed_walkthrough} and \ref{appendix:lagrangian_dual_detailed_walkthrough}, respectively, followed by their optimality evaluation in \ref{appendix_reconfiguration_algo_analysis}.

\begin{figure*}[ht]
  \centering
   \includegraphics[width=0.82\linewidth,height=1.3in,trim={0.7cm 0.10cm 0.7cm 0.8cm}, scale=1]{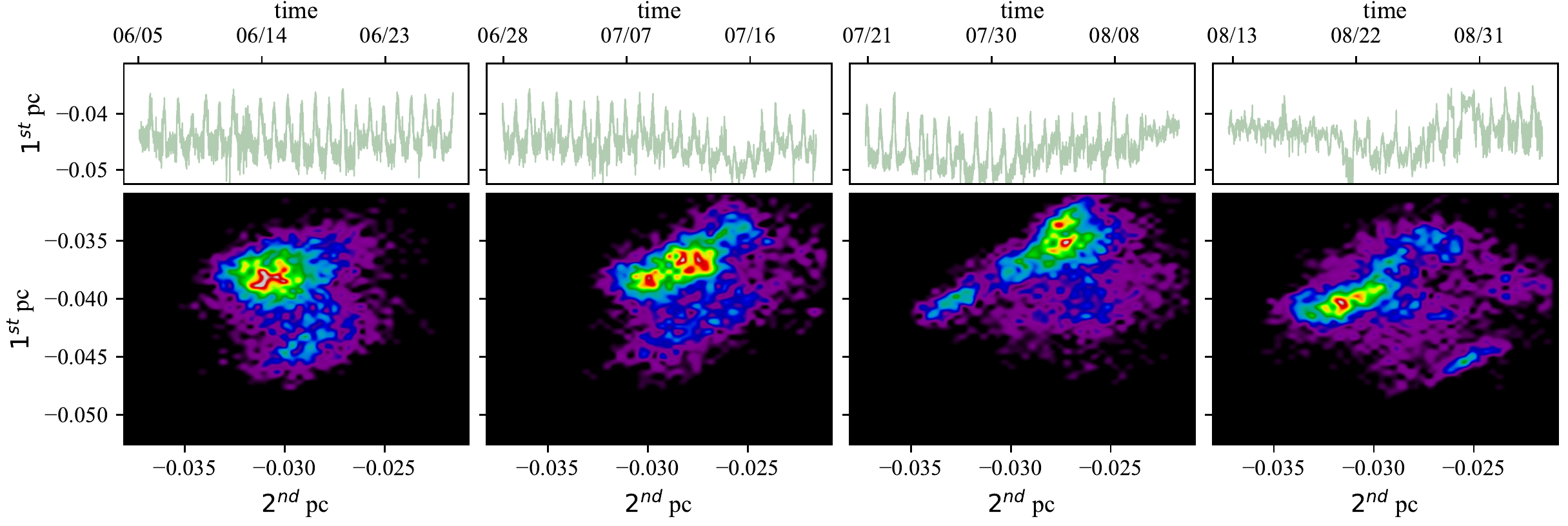}
\vspace{-8pt}
   \caption{\small PCA on 3 months of production DCN traffic; each plot shows 24 days. Top: temporal variation along the first principal component. Bottom: PCAs as 2D-heatmaps; higher occurrence is represented by brighter colors. }
\vspace{-13pt}
    \label{fig:traffic_clustering}
\end{figure*}

\vspace{-5pt}
\section{Picking Representative Traffic} \label{section_traffic_clustering}
The first step in {\metteor}’s workflow is to extract multiple representative TMs. We show how to extract these TMs purely from historical traces, while assuming no knowledge of the underlying application mix. 

Recall Fig. \ref{fig:historical_similarity} in \S\ref{section_proof_of_concept} showed that inter-pod traffic exhibits a weak form of temporal stability, which we call recurrence. This behavior causes TMs to form clusters that vary slowly over time.  But how exactly does traffic recurrence lead to clustering behavior? We offer an informal reasoning as follows. Consider a traffic matrix snapshot as a point in high-dimensional space. Over time, recurrent snapshots will begin to ``congregate’’ within the vicinity of one another to form clusters, rather than scatter around uniformly in space.

The appearance of clustering effects is predicated of traffic exhibiting both spatial and temporal locality. The spatial locality is inherent to data center job placement. Large DCNs tend to assign different groups of pods to specific production areas, so pods belonging to the same production areas are more likely to communicate with one another. Meanwhile, the temporal locality comes from traffic recurrence, and this property is partially determined by user behavior. The regularity of usage patterns from long-term customers in cloud data centers, or the routine running of batched jobs (e.g. integration tests) in private data centers may all cause traffic recurrence.  

\vspace{-5pt}
\subsection{Traffic Clustering Effect}\label{section_traffic_clustering_effect}
\subsubsection{Visualization of Traffic Clusters}\hfill \\
\noindent Traffic matrices are high dimensional data points (each point has $\Theta$(pod num${}^2)$ dimensions), so visualizing their temporal evolution is exceedingly difficult. To this end, we employ principal component analysis (PCA) to reduce the dimensionality of the TM snapshots, and project each snapshot onto the plane formed by the first and second principal components~\cite{ghodsi2006dimensionality}. Since an optimal topology for a TM remains optimal regardless of scaling, we should not distinguish TMs differ only in their total traffic volume. Thus, we normalize all TMs to 1. 


Fig. \ref{fig:traffic_clustering} shows an example of traffic from one of the production data centers, with this projection represented as a 2-D heatmap, where brighter colors indicate areas with a higher occurrence. The proportion of variance explained (PVE) by the first two components is 91\% of the total variation. Each plot covers about 24 days’ worth of traffic snapshots. There are noticeable clustering effects, which manifest as ``clouds'' of bright patches. These clusters shift slowly over time, so topology reconfiguration to handle these shifts is necessary. 

The top row in Fig. \ref{fig:traffic_clustering} shows the traffic temporal variation along the principal component. Note that while the clusters change slowly over the course of weeks, the variation along the principal component between snapshots is rather significant. The maximum variation from peak to trough accounts for about $5\%$ of the total traffic, which is $85\%$ of the largest inter-pod traffic, and 25$\times$ the average inter-pod traffic. 

\begin{figure}[!tp]
\centering
\begin{tikzpicture}
\begin{axis}[ybar, width=1.\columnwidth, height=3.8cm, symbolic x coords={fab1, fab2, fab3, fab4, fab5, fab6, fab7, fab8, fab9, fab10, fab11, fab12}, xtick=data,bar width=8pt,legend pos=north west, legend columns=2, ymax=7, ymin= 0, xtick align=inside,ylabel style={align=center}, ylabel style={align=center}, ylabel=\footnotesize Avg. number \\ of clusters, ylabel shift = -2pt, legend style={draw=none, at={(0,1.225), fill=none},anchor=north west,}, xtick pos=left, ytick pos=left, ymajorgrids]
\addplot[fill=lightgray, postaction={pattern=dots}, error bars/.cd, y dir=both, y explicit,error bar style=black] coordinates {
(fab1, 2.44) += (0, 0.49) -= (0, 0.49)
(fab2, 2.98) += (0, 0.97) -= (0, 0.97)
(fab3, 2.00) += (0, 0.00) -= (0, 0.00)
(fab4, 2.67) += (0, 0.62) -= (0, 0.62)
(fab5, 2.56) += (0, 0.53) -= (0, 0.53)
(fab6, 2.78) += (0, 0.86) -= (0, 0.86)
(fab7, 4.00) += (0, 1.02) -= (0, 1.02)
(fab8, 2.33) += (0, 0.47) -= (0, 0.47)
(fab9, 2.24) += (0, 0.42) -= (0, 0.42)
(fab10, 2.67) += (0, 0.94) -= (0, 0.94)
(fab11, 2.78) += (0, 0.82) -= (0, 0.82)
(fab12, 3.11) += (0, 1.01) -= (0, 1.01)
};
\addplot+[only marks, mark = asterisk, mark options={xshift=\pgfkeysvalueof{/pgf/bar shift}}] table[x index=0,y index=1] {
fab1   2
	fab2   2
fab3   2
fab4   2
	fab5   2
fab6   2
fab7   3
	fab8   2
fab9   2
fab10   2
	fab11   2
fab12   2
fab1 2
		fab1   3
	fab2   6
fab3   2
fab4   4
	fab5   4
fab6   5
fab7   6
	fab8   3
fab9   3
fab10  5
	fab11   5
fab12   5
};
\end{axis}
\end{tikzpicture}
\vspace{-11pt}
\caption{\small TM clustering statistics of 12 production DCNs. Six months of traffic snapshots are split into 10 segments of $\sim$18 days each. The appropriate \# of clusters of each segment is determined using silhouette method. Error bars denote the std dev, and asterisk marks denote max and min.}
\vspace{-21pt}
\label{fig:traffic_longterm_statistics}
\end{figure}
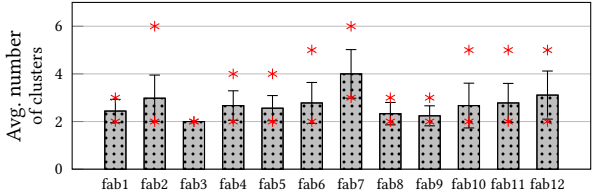

\vspace{-6pt}
\subsubsection{Clustering Statistics Across 12 Data Centers}\hfill \\
\noindent Next, we study the clustering effect across a fleet of 12 production DCNs. Using 6 months’ worth of historical traffic from each data center, we compute the optimal number of clusters for each 18-day period using the silhouette method~\cite{SilhouetteMethod}. This step repeats for 10 contiguous 18-day segments. Fig. \ref{fig:traffic_longterm_statistics} summarizes these statistics.



Across all 12 fabrics, the average ``appropriate’’ number of clusters for each 18-day period is below 4. The appropriate number of clusters in each period, and how it evolves over time, are generally artifacts of the DCN’s underlying application mix and scheduler behavior. Therefore, the optimal number of clusters has to be determined individually for each fabric through traffic analysis. 

The optimal number of clusters number changes from one 18-day segment to the next (see the error bar overlays in Fig. \ref{fig:traffic_longterm_statistics}), though the deviations are small (within a $\pm 1$ range of the average). This suggests that network operators can pick a consistent number of TMs used for {\metteor } in all reconfiguration epochs for each fabric.

\vspace{-6pt}
\subsection{Finding Representative TMs}
\label{section_finding_representative_traffic_matrices}
In theory, we could consider using the set of \emph{all} historical TMs for optimization. Doing so guarantees coverage of any future traffic that is recurrent, but the runtime and memory complexity required to compute a topology for such a large set of TMs would also increase astronomically. Hence, to avoid adding significant computational complexity to {\metteor}, we need to pick the smallest set of TMs that is still sufficiently representative of future traffic.

Though prior works have proposed effective methods for selecting traffic matrix estimators (e.g. CritMat~\cite{zhang2005finding}), we employ a simple, yet effective, $\kappa$-means clustering algorithm to find the centroids within the historical traffic snapshots. Computing these cluster centroids gives us a compact representation of historical traffic that retains much of the ``features’’ of the historical traffic. As long as a future traffic snapshot is recurrent, there is a high probability it will be well-represented by at least one of these cluster centroids. 


\vspace{-10pt}
\subsection{Accuracy of Cluster-based Prediction}\label{subsection:compare_prediction}
\noindent Next, we test how well traffic clusters predict future traffic. First, we split 6 months’ worth of traffic into segments of two weeks. In cluster-based prediction, we extract $\kappa$ cluster centroids from each 2-week segment, and use them to predict traffic in the next segment. We compare cluster-based prediction against two single-traffic-based predictions, namely ${ave}$ and ${max}$, which pick the historical average and component-wise max values, respectively.

We use \emph{cosine similarity} (defined in~\cite{CosineSimilarity}) to evaluate how similar the predicted TM is to the actual TM. Given two TMs’ vector representations, $\vec{T}_1, \vec{T}_2$, their cosine similarity $sim(\vec{T}_1, \vec{T}_2)$ measures how parallel (or similar) these two vectors are. If $sim(\vec{T}_1, \vec{T}_2)$ is close to $1$, it follows that a $\vec{T}_1$-optimized topology would also be close-to-optimal for $\vec{T}_2$. For multiple representative TM cases, we pick the one that is most similar to the evaluated traffic snapshot.

Fig. \ref{fig:traffic_similarity} shows that cluster-based prediction yields higher accuracy than both ${ave}$ and ${max}$, as they can more effectively capture long-term traffic behavior. The long tail of the $\kappa = 1$ curve indicates that fewer clusters may hurt worst-case accuracy, showing an inability to cover outlier TMs. Choosing $\kappa = 5$ over $\kappa = 3$ shows a diminishing improvement in prediction accuracy as $\kappa$ increases. ${max}$ has the lowest accuracy, as it captures the maximum element-wise demand that may not be representative in general. 

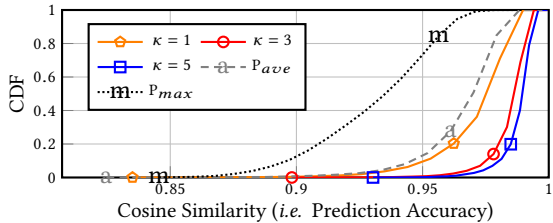
\begin{figure}[!tp]
\pgfplotstableread{plot_data/cluster_ptoe_predictability_cdf_ju2_dls03.txt}
	\datatable
\begin{tikzpicture}
\begin{axis}[xlabel = {\footnotesize Cosine Similarity (\emph{i.e. } Prediction Accuracy)}, xlabel near ticks, width=0.95\columnwidth, height=1.5in, outer sep=-1pt, ymin=0,ymax=1,xmax=1, ylabel= \footnotesize CDF, ylabel near ticks, ylabel shift=-2pt, xlabel shift=-3pt, legend style={at={(0.06, 0.90)}, legend columns=2, anchor= north west}, xtick pos=left, ytick pos=left, grid, thick]

\addplot+[mark=pentagon, mark options={solid}, mark repeat={15}, color=orange] table[y=k1_bin,x=k1_cdf] from \datatable ;
\addplot+[mark=o, mark repeat={15}, color=red] table[y=k3_bin,x=k3_cdf] from \datatable ;
\addplot+[mark=square, mark repeat={15}, color=blue] table[y=k5_bin,x=k5_cdf] from \datatable ;
\addplot+[mark=text, text mark=a, mark options={solid}, mark repeat={15}, densely dashed, color=gray,] table[x=ptoe_ave_bin,y=ptoe_ave_cdf] from \datatable ;
\addplot+[mark=text, text mark=m, mark options={solid}, mark repeat={15}, densely dotted, color=black,] table[x=ptoe_max_bin,y=ptoe_max_cdf] from \datatable ;
\legend{\tiny $\kappa = 1$, \tiny $\kappa = 3$, \tiny $\kappa = 5$, \tiny $\text{P}_{ave}$, \tiny $\text{P}_{max}$}
\end{axis}
\end{tikzpicture}
\vspace{-11pt}
 \caption{\small Accuracy of cluster-based \emph{v.s.} single-traffic predictors in estimating future traffic.} 
\vspace{-14pt}  
\label{fig:traffic_similarity}
\end{figure}



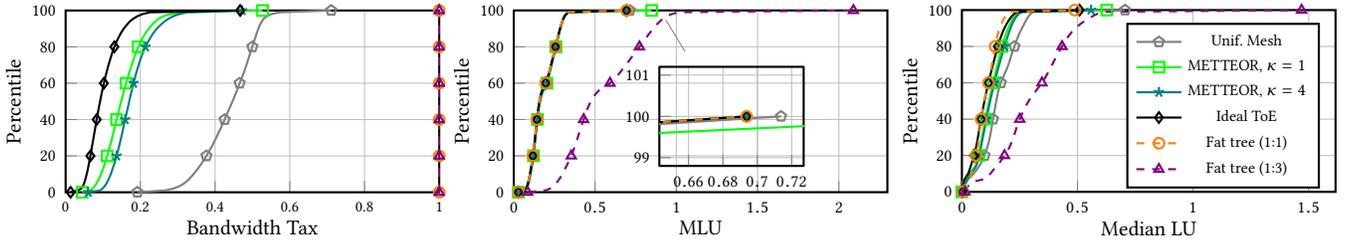
\begin{figure*}[!ht]
\pgfplotstableread{plot_data/invcdf_ju1_dls08_perfect_perfect} \perfect
\pgfplotstableread{plot_data/invcdf_ju1_dls08_static_perfect} \statictoe
\pgfplotstableread{plot_data/invcdf_ju1_dls08_robusttoe_r4032t2016c4_perfect} \robusttoe
\pgfplotstableread{plot_data/invcdf_ju1_dls08_robusttoe_r4032t2016c1_perfect} \robusttoeb
\pgfplotstableread{plot_data/invcdf_ju1_dls08_avetoe_r4032t2016_perfect} \avetoe
\pgfplotstableread{plot_data/invcdf_ju1_dls08_fattree_ecmp_taper1p000} \fattreeecmptaperone
\pgfplotstableread{plot_data/invcdf_ju1_dls08_fattree_ecmp_taper0p333} \fattreeecmptaperthree
\begin{subfigure}{0.32\textwidth}
        \centering
\begin{tikzpicture}
\begin{axis}[ylabel = \footnotesize Percentile, xlabel = \footnotesize Bandwidth Tax, ymin=0, ymax=100, ylabel near ticks, width=1.15\linewidth, height=4.0cm, ylabel shift = -3pt, xlabel shift = -4pt, xmax=2, xmin=1, xtick pos=left, ytick pos=left, grid, thick, xtick={1, 1.2, 1.4, 1.6, 1.8, 2}, xticklabels={0, 0.2, 0.4, 0.6, 0.8, 1}]
\addplot+[color=gray,mark=pentagon,mark repeat={20}] table[y=percentile, x=ave_hop_count] from \statictoe ;
\addplot+[color=green, mark=square,mark repeat={20}] table[y=percentile,x=ave_hop_count] from \robusttoeb ;
\addplot+[color=teal, mark=star, mark repeat={20}] table[y=percentile,x=ave_hop_count] from \robusttoe ;
\addplot+[color=black, mark=diamond, mark repeat={20}] table[y=percentile,x=ave_hop_count] from \perfect ;
\addplot+[color=orange, mark=o, mark repeat={20}, dashed, mark options={solid}] table[y=percentile, x=ave_hop_count] from \fattreeecmptaperone ;
\addplot+[color=violet, mark=triangle, mark repeat={20}, dashed, mark options={solid}] table[y=percentile, x=ave_hop_count] from \fattreeecmptaperthree ;
\end{axis}
\end{tikzpicture}
    \end{subfigure}
    ~ 
    \begin{subfigure}{0.32\textwidth}
        \centering
\begin{tikzpicture}
\begin{axis}[ylabel = \footnotesize Percentile, xlabel = \footnotesize MLU, ymin=0, ymax=100, xmin=0, ylabel near ticks,width=1.15\linewidth, height=4.0cm, ylabel shift = -3pt, xlabel shift = -4pt, xtick pos=left, ytick pos=left, grid, thick]
\addplot+[color=gray,mark=pentagon,mark repeat={20}] table[y=percentile,x=mlu] from \statictoe ;
\addplot+[color=green, mark=square,mark repeat={20}] table[y=percentile,x=mlu] from \robusttoeb ;
\addplot+[color=teal, mark=star, mark repeat={20}] table[y=percentile,x=mlu] from \robusttoe ;
\addplot+[color=black, mark=diamond, mark repeat={20}] table[y=percentile,x=mlu] from \perfect ;
\addplot+[color=orange, mark=o, mark repeat={20}, dashed, mark options={solid}] table[y=percentile, x=mlu] from \fattreeecmptaperone ;
\addplot+[color=violet, mark=triangle, mark repeat={20}, dashed, mark options={solid}] table[y=percentile, x=mlu] from \fattreeecmptaperthree ;
\coordinate (pt) at (axis cs:0.9, 100);
\end{axis}
\node[pin={[pin distance=0.95cm, anchor=north]330:{%
    \begin{tikzpicture}[trim axis left,trim axis right]
    \begin{axis}[
      line join=round,
      cycle list={{color = red}, {color=green}, {color=blue}, {color=black}}, enlargelimits,width=3.5cm, 
      height=2.9cm, ymin=99, ymax=101, xmin=0.65, xmax=0.72, xtick pos=left, ytick pos=left, grid, thick]
\addplot+[color=gray,mark=pentagon,mark repeat={2}] table[y=percentile,x=mlu] from \statictoe ;
\addplot+[color=green, mark=square,mark repeat={2}] table[y=percentile,x=mlu] from \robusttoeb ;
\addplot+[color=teal, mark=star, mark repeat={2}] table[y=percentile,x=mlu] from \robusttoe ;
\addplot+[color=black,mark=diamond, mark repeat={2}] table[y=percentile,x=mlu] from \perfect ;
\addplot+[color=orange, mark=o, mark repeat={4}, dashed, mark options={solid}] table[y=percentile, x=mlu] from \fattreeecmptaperone ;
    \end{axis}
    \end{tikzpicture}%
}}] at (pt) {};
\end{tikzpicture}
\end{subfigure}
    ~
\begin{subfigure}{0.32\textwidth}
\centering
\begin{tikzpicture}
\begin{axis}[ylabel = \footnotesize Percentile, xlabel = \footnotesize Median LU, ymin=0, ymax=100, ylabel near ticks,width=1.15\linewidth, height=4.0cm, ylabel shift = -3pt, xlabel shift = -4pt, legend columns=1, legend style={at={(0.44, 0.93)},anchor=north west,}, xtick pos=left, ytick pos=left, grid, thick, xmin=0]
\addplot+[color=gray, mark=pentagon, mark repeat={20}] table[y=percentile,x=lu50] from \statictoe ;
\addlegendentry{\tiny Unif. Mesh}
\addplot+[color=green, mark=square,mark repeat={20}] table[y=percentile,x=lu50] from \robusttoeb ;
\addlegendentry{\tiny {\metteor}, $\kappa=1$}
\addplot+[color=teal,mark=star, mark repeat={20}] table[y=percentile,x=lu50] from \robusttoe ;
\addlegendentry{\tiny {\metteor}, $\kappa=4$}
\addplot+[color=black, mark=diamond, mark repeat={20}] table[y=percentile,x=lu50] from \perfect ;
\addlegendentry{\tiny Ideal ToE}
\addplot+[color=orange, mark=o, mark repeat={20}, dashed, mark options={solid}] table[y=percentile, x=lu50] from \fattreeecmptaperone ;
\addlegendentry{\tiny Fat tree (1:1)}
\addplot+[color=violet,mark=triangle, mark repeat={20}, dashed, mark options={solid}] table[y=percentile, x=lu50] from \fattreeecmptaperthree ;
\addlegendentry{\tiny Fat tree (1:3)}
\end{axis}
\end{tikzpicture}
\end{subfigure}
\vspace{-13pt}
\caption{\small Percentile plot of network performance using 6-months’ worth of TM snapshots from a production DCN. {\metteor } reconfigures topology every 2 weeks. {\metteor}/Mesh/Ideal topologies use MCF routing, and fat tree uses ECMP.} 
\vspace{-4pt}
\label{fig:ste_nte_ite_results}
\end{figure*}

\section{Performance Evaluation}\label{section_performance_evaluation}
We now evaluate {\metteor}’s performance over an extended timescale. The criteria we evaluate are: 1) performance comparison among different network topologies (\S\ref{subsection:ideal_routing}), 2) performance robustness under different reconfiguration frequencies (\S\ref{subsection:robustness_to_traffic_change}), and 3) performance under routing uncertainties (\S\ref{subsection:non_ideal_routing}). We assume a fluid traffic model to help us evaluate performance over extended periods, while still capturing the essential macroscopic properties. 

\noindent\textbf{Dataset: }
Our evaluations are driven by production DCN TM snapshots. Each snapshot captures inter-pod traffic over 5 minutes. The number of snapshots from each of the 12 simulated data centers totals up to 6 months’ worth of data (i.e. slightly over 50k snapshots per data center). We present a subset of our findings here; complete results are in \ref{appendix_additional_simulation_results}.

\noindent \textbf{Metrics: } 
The main metrics we look at are: 
\begin{itemize}[leftmargin=3pt]
\item \emph{Link utilization} (LU) is a good indicator of link congestion, so a lower LU is preferred. However, MLU only reflects congestion at the busiest link, so we also look at the median LU to gauge the average case link congestion. Although LU cannot exceed 1 in practice because packets can be dropped, we allow LU to be greater than 1 in our evaluation, as it could reflect how severe the packet drop is.


\item \emph{Bandwidth tax} is the additional capacity on average needed to route traffic~\cite{mellette2019expanding}. For instance, if 60\% of traffic traverses indirect 2-hop paths and 40\% of traffic traverses direct paths, then the bandwidth tax is $0.6\times(2-1)+0.4\times(1-1)=0.6$. Since we allow a maximum of 2 inter-pod hops for each packet in this paper, bandwidth tax is equal to the fraction of 2-hop traffic. Clearly, a lower bandwidth tax is preferred due to the following reasons. First, indirect paths increases packet latency. Second, lowering bandwidth tax directly lowers the number of concurrent flows going through each switch. As DCN switches typically have shallow buffers, lowering the number of concurrent flows through a switch helps reduce the probability of incast~\cite{chen2009incast}. 
\end{itemize}




\subsection{Topological Comparison}\label{subsection:ideal_routing}
We first compare {\metteor } with other DCN topologies. 

\noindent \textbf{Topology: }
Our main contender is {\metteor } with the following settings. $\kappa = 4$ representative TMs are extracted from 2 weeks’ worth of historical traffic preceding each reconfiguration epoch, and the logical topology is reconfigured based on these representative TMs every two weeks. 

\noindent \textbf{Routing: }
We use traffic engineering (TE) for routing. As mentioned in \S\ref{section_traffic_engineering_related_works}, TE algorithms typically consists of a path-selection step, and a load balancing step. For path selection, we consider all paths between two pods that are within 2 hops. That is, in addition to a direct hop between the source and destination pods, traffic is allowed to transit at another intermediate pod before reaching its destination. For load balancing, we compute the optimal routing weights that minimizes MLU using an multi-commodity flow (MCF) formulation, as done in~\cite{jyothi2016measuring, hamedazimi2014firefly}. 

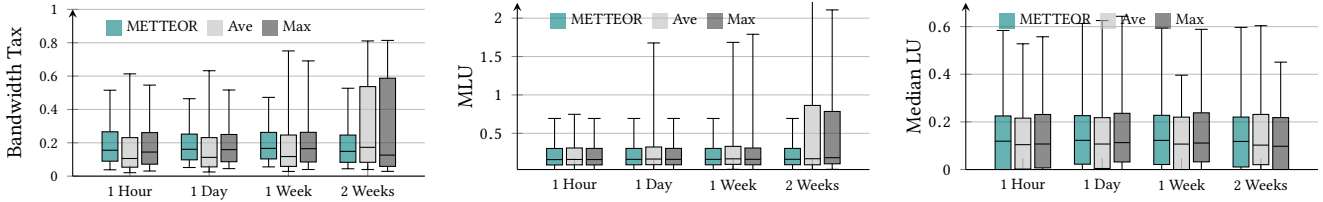
\begin{figure*}[ht!]
\begin{subfigure}[c]{0.32\linewidth} 
\centering
\begin{tikzpicture}
\begin{axis}[boxplot/draw direction=y, width=1.1\linewidth, height=3.8cm,
axis y line=left, axis x line*=bottom, ylabel=\footnotesize Bandwidth Tax, xtick={2, 6, 10, 14}, xticklabels={1 Hour, 1 Day, 1 Week, 2 Weeks}, xlabel near ticks, xlabel shift =-4pt, legend style={draw = none, fill = none, at={(0.4,1.0)}, anchor=north,legend columns=3}, legend image code/.code={ \draw[#1, draw=none] (0cm,-0.1cm) rectangle (0.2cm,0.1cm); }, cycle list={{fill=teal},{fill=lightgray},{fill=darkgray}}, ymajorgrids, ymax=2,ymin=1, ytick={1, 1.2, 1.4, 1.6, 1.8, 2}, yticklabels={0, 0.2, 0.4, 0.6, 0.8, 1}, every axis plot/.append style={fill,fill opacity=0.6}]
\addplot+[boxplot prepared={draw position=1,lower whisker=1.0380, lower quartile=1.0896, median=1.1553, upper quartile=1.2659, upper whisker=1.5153}]  coordinates {};
\addlegendentry{\tiny {\metteor}\;};
\addplot+[boxplot prepared={draw position=2, lower whisker=1.0206, lower quartile=1.0536, median=1.1057, upper quartile=1.2309, upper whisker=1.6134}]  coordinates {};
\addlegendentry{\tiny Ave\;};
\addplot+[boxplot prepared={draw position=3, lower whisker=1.0306, lower quartile=1.0716, median=1.1444, upper quartile=1.2610, upper whisker=1.5459}]  coordinates {};
\addlegendentry{\tiny Max\;};
\addplot+[boxplot prepared={draw position=5, lower whisker=1.0517, lower quartile=1.0978, median=1.1610, upper quartile=1.2521, upper whisker=1.4645}]  coordinates {};
\addplot+[boxplot prepared={draw position=6, lower whisker=1.0248, lower quartile=1.0555, median=1.1125, upper quartile=1.2310, upper whisker=1.6322}]  coordinates {};
\addplot+[boxplot prepared={draw position=7, lower whisker=1.0450, lower quartile=1.0864, median=1.1590, upper quartile=1.2498, upper whisker=1.5167}]  coordinates {};
\addplot+[boxplot prepared={draw position=9, lower whisker=1.0555, lower quartile=1.1032, median=1.1663, upper quartile=1.2624, upper whisker=1.4723}]  coordinates {};
\addplot+[boxplot prepared={draw position=10, lower whisker=1.0282, lower quartile=1.0583, median=1.1174, upper quartile=1.2465, upper whisker=1.7509}]  coordinates {};
\addplot+[boxplot prepared={draw position=11, lower whisker=1.0397, lower quartile=1.0849, median=1.1648, upper quartile=1.2628, upper whisker=1.6911}]  coordinates {};
\addplot+[boxplot prepared={draw position=13, lower whisker=1.0442, lower quartile=1.0828, median=1.1486, upper quartile=1.2463, upper whisker=1.5271}]  coordinates {};
\addplot+[boxplot prepared={draw position=14, lower whisker=1.0397, lower quartile=1.0829, median=1.1727, upper quartile=1.5370, upper whisker=1.8108}]  coordinates {};
\addplot+[boxplot prepared={draw position=15, lower whisker=1.0285, lower quartile=1.0585, median=1.1262, upper quartile=1.5873, upper whisker=1.8140}]  coordinates {};
\end{axis}
\end{tikzpicture}
\end{subfigure}
~
\begin{subfigure}[c]{0.32\linewidth} 
\centering
\begin{tikzpicture}
\begin{axis}[boxplot/draw direction=y, width=1.1\linewidth, height=3.8cm,
axis y line=left, axis x line*=bottom, ylabel=\footnotesize MLU, xtick={2, 6, 10, 14}, xticklabels={1 Hour, 1 Day, 1 Week, 2 Weeks}, xlabel near ticks, ylabel shift =-4pt, legend style={draw = none, fill = none, at={(0.4,1.0)}, anchor=north,legend columns=3}, legend image code/.code={ \draw[#1, draw=none] (0cm,-0.1cm) rectangle (0.2cm,0.1cm); }, cycle list={{fill=teal},{fill=lightgray},{fill=darkgray}}, ymajorgrids, ymax=2.2, every axis plot/.append style={fill,fill opacity=0.6}]
\addplot+[boxplot prepared={draw position=1,lower whisker=0.0284, lower quartile=0.0880, median=0.1580, upper quartile=0.3028, upper whisker=0.6938}]  coordinates {};
\addlegendentry{\tiny {\metteor}\;};
\addplot+[boxplot prepared={draw position=2, lower whisker=0.0284, lower quartile=0.0895, median=0.1594, upper quartile=0.3092, upper whisker=0.7475}]  coordinates {};
\addlegendentry{\tiny Ave\;};
\addplot+[boxplot prepared={draw position=3, lower whisker=0.0284, lower quartile=0.0885, median=0.1586, upper quartile=0.3036, upper whisker=0.6938}]  coordinates {};
\addlegendentry{\tiny Max\;};
\addplot+[boxplot prepared={draw position=5, lower whisker=0.0284, lower quartile=0.0904, median=0.1609, upper quartile=0.3036, upper whisker=0.6938}]  coordinates {};
\addplot+[boxplot prepared={draw position=6, lower whisker=0.0286, lower quartile=0.0917, median=0.1649, upper quartile=0.3220, upper whisker=1.6772}]  coordinates {};
\addplot+[boxplot prepared={draw position=7, lower whisker=0.0287, lower quartile=0.0904, median=0.1611, upper quartile=0.3040, upper whisker=0.6938}]  coordinates {};
\addplot+[boxplot prepared={draw position=9, lower whisker=0.0284, lower quartile=0.0904, median=0.1609, upper quartile=0.3036, upper whisker=0.6938}]  coordinates {};
\addplot+[boxplot prepared={draw position=10, lower whisker=0.0308, lower quartile=0.0957, median=0.1679, upper quartile=0.3301, upper whisker=1.6847}]  coordinates {};
\addplot+[boxplot prepared={draw position=11, lower whisker=0.0292, lower quartile=0.0904, median=0.1613, upper quartile=0.3087, upper whisker=1.7900}]  coordinates {};
\addplot+[boxplot prepared={draw position=13, lower whisker=0.0284, lower quartile=0.0912, median=0.1616, upper quartile=0.3037, upper whisker=0.6938}]  coordinates {};
\addplot+[boxplot prepared={draw position=14, lower whisker=0.0295, lower quartile=0.0904, median=0.1691, upper quartile=0.8631, upper whisker=2.5273}]  coordinates {};
\addplot+[boxplot prepared={draw position=15, lower whisker=0.0308, lower quartile=0.1031, median=0.1827, upper quartile=0.7847, upper whisker=2.1076}]  coordinates {};
\end{axis}
\end{tikzpicture}
\end{subfigure}
~
\begin{subfigure}[c]{0.32\linewidth} 
\centering
\begin{tikzpicture}
\begin{axis}[boxplot/draw direction=y, width=1.1\linewidth, height=3.8cm,
axis y line=left, axis x line*=bottom, ylabel=\footnotesize Median LU, xtick={2, 6, 10, 14}, xticklabels={1 Hour, 1 Day, 1 Week, 2 Weeks}, xlabel near ticks, ylabel shift =-4pt, legend style={draw = none, fill = none, at={(0.4,1)}, anchor=north,legend columns=3}, legend image code/.code={ \draw[#1, draw=none] (0cm,-0.1cm) rectangle (0.2cm,0.1cm); }, cycle list={{fill=teal},{fill=lightgray},{fill=darkgray}}, ymajorgrids, ymax=0.7, every axis plot/.append style={fill,fill opacity=0.6}]
\addplot+[boxplot prepared={draw position=1,lower whisker=0.0000, lower quartile=0.0000, median=0.1190, upper quartile=0.2250, upper whisker=0.5838}]  coordinates {};
\addlegendentry{\tiny {\metteor}\;};
\addplot+[boxplot prepared={draw position=2,lower whisker=0.0000, lower quartile=0.0020, median=0.1045, upper quartile=0.2157, upper whisker=0.5279}]  coordinates {};
\addlegendentry{\tiny Ave\;};
\addplot+[boxplot prepared={draw position=3,lower whisker=0.0000, lower quartile=0.0070, median=0.1072, upper quartile=0.2311, upper whisker=0.5574}]  coordinates {};
\addlegendentry{\tiny Max\;};
\addplot+[boxplot prepared={draw position=5,lower whisker=0.0000, lower quartile=0.0225, median=0.1226, upper quartile=0.2263, upper whisker=0.6135}]  coordinates {};
\addplot+[boxplot prepared={draw position=6,lower whisker=0.0000, lower quartile=0.0044, median=0.1071, upper quartile=0.2178, upper whisker=0.6254}]  coordinates {};
\addplot+[boxplot prepared={draw position=7,lower whisker=0.0000, lower quartile=0.0318, median=0.1130, upper quartile=0.2360, upper whisker=0.6433}]  coordinates {};
\addplot+[boxplot prepared={draw position=9,lower whisker=0.0000, lower quartile=0.0213, median=0.1224, upper quartile=0.2276, upper whisker=0.5938}]  coordinates {};
\addplot+[boxplot prepared={draw position=10,lower whisker=0.0000, lower quartile=0.0000, median=0.1069, upper quartile=0.2200, upper whisker=0.3960}]  coordinates {};
\addplot+[boxplot prepared={draw position=11,lower whisker=0.0000, lower quartile=0.0323, median=0.1113, upper quartile=0.2381, upper whisker=0.5885}]  coordinates {};
\addplot+[boxplot prepared={draw position=13,lower whisker=0.0000, lower quartile=0.0108, median=0.1181, upper quartile=0.2203, upper whisker=0.5967}]  coordinates {};
\addplot+[boxplot prepared={draw position=14,lower whisker=0.0000, lower quartile=0.0205, median=0.1025, upper quartile=0.2314, upper whisker=0.6039}]  coordinates {};
\addplot+[boxplot prepared={draw position=15,lower whisker=0.0000, lower quartile=0.0000, median=0.0977, upper quartile=0.2181, upper whisker=0.4508}]  coordinates {};
\end{axis}
\end{tikzpicture}
\end{subfigure}
\vspace{-13pt}
\caption{\small Performance sensitivity to reconfiguration frequency, comparing {\metteor } to single-traffic approaches. Boxes represent the 95-th and 5-th percentile values; the whiskers represent the max and min values.} 
\label{fig:reconfig_frequency_sensitivity}
\vspace{-11pt}
\end{figure*}


\noindent \textbf{Versus fat tree: }
We first compare against fat trees, which represent the industry standard in DCN topologies. Due to cost reasons, most network operators tend to oversubscribe to the aggregation or core layers~\cite{greenberg2009vl2, benson2010network, farrington2013facebook}. Our evaluation includes a 1:3 oversubscribed fat tree, which has comparable cost to {\metteor}, and a non-oversubscribed 1:1 fat tree. All fat tree topologies use ECMP routing.

Fat tree networks perform poorly in terms of bandwidth tax when compared against {\metteor}’s topologies. A fat tree network has an additional spine layer of packet switches. Therefore, inter-pod traffic always consumes bandwidth of 2 hops (one between the source pod and the spine, and the other between the spine and the destination pod). Under a {\metteor } topology, on average over 80\% of traffic traverses single-hop paths (OCSs are transparent to DCN traffic in between adjacent OCS reconfigurations). In terms of MLU, a 1:1 fat tree, with its full bisection bandwidth, outperforms all other topologies, barring ideal ToE. When compared against a 1:3 fat tree of comparable cost, {\metteor} reduces tail MLU by about 3$\times$.

\noindent \textbf{Versus uniform mesh: }
We also compare with a uniform mesh that directly connects pods without an OCS/spine layer. Uniform meshes are considered a class of expander networks, offering large bisection bandwidth at a lower cost than fat trees. Since {\metteor}’s logical topology is also mesh-like with non-uniform interpod connectivity, a uniform mesh represents a natural baseline for comparison. 


{\metteor} consistently lowers bandwidth tax by $\approx$0.35 on average over a uniform mesh, due to its strategic link placement between hotspots.
The lower bandwidth tax translates into an average median-LU improvement of 50\%, due to a reduction in overall traffic load. In terms of MLU, {\metteor } performs comparably to a uniform mesh. This comparison showcases the benefits of ToE, as we can reduce bandwidth tax without sacrificing MLU.


\noindent \textbf{Versus ideal ToE: }
Ideal ToE computes an offline optimal topology that minimizes the MLU and bandwidth tax for each traffic matrix, and thus it represents the performance upper bound of topology engineering. 


When paired with ideal load balancing, {\metteor}’s MLU performance is very close to that of ideal ToE. As our evaluation traffic matrices are highly skewed, with a significant amount of traffic comes from a small subset of pods, MLU becomes limited by a pod’s total egress/ingress capacity. Therefore, even ideal ToE cannot do much to improve MLU. Still, ideal ToE lowers bandwidth tax over {\metteor } by 0.08 on average, and 0.01 at the tail. 

\noindent \textbf{Versus single-traffic ToE: }
Finally, we compare {\metteor } against single-traffic ToE to showcase the benefits of using multiple TMs. As an example of single-traffic ToE, {\metteor } ($\kappa = 1$) has a longer tail than {\metteor } ($\kappa = 4$) for all the metrics shown in Fig. \ref{fig:ste_nte_ite_results}. Indeed, it is generally very difficult to predict future TMs with a single TM due to traffic uncertainties. One may also propose using an element-wise average, or maximum traffic estimator in single-traffic ToE. We postpone the detailed comparison in \S\ref{subsection:robustness_to_traffic_change}.




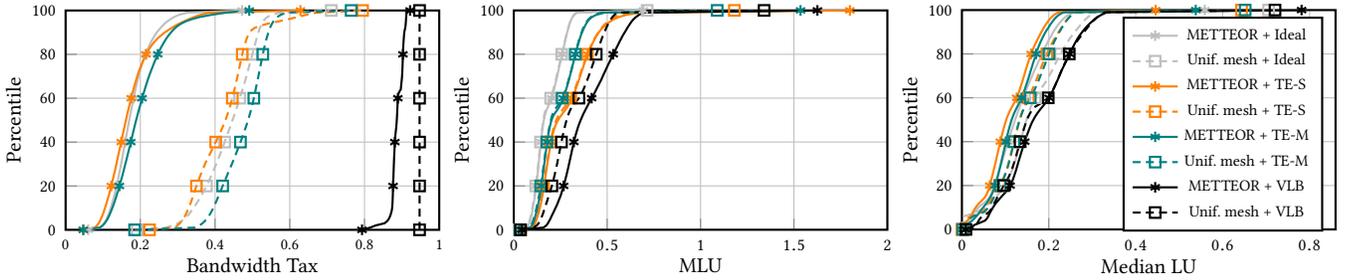
\begin{figure*}[th!]
\pgfplotstableread{plot_data/invcdf_ju1_dls08_robusttoe_r4032t2016c4_perfect} \robusttoeperfectte
\pgfplotstableread{plot_data/invcdf_ju1_dls08_static_perfect} \statictoeperfectte
\pgfplotstableread{plot_data/invcdf_ju1_dls08_robusttoe_r4032t2016c4_robustweightedte_r1t12c4} \robusttoerobustte
\pgfplotstableread{plot_data/invcdf_ju1_dls08_static_robustweightedte_r1t12c4} \statictoerobustte
\pgfplotstableread{plot_data/invcdf_ju1_dls08_robusttoe_r4032t2016c4_robustweightedte_r1t12c1} \maxtoemaxte
\pgfplotstableread{plot_data/invcdf_ju1_dls08_static_robustweightedte_r1t12c1} \staticavete
\pgfplotstableread{plot_data/invcdf_ju1_dls08_robusttoe_r4032t2016c4_wcmp} \robusttoewcmp
\pgfplotstableread{plot_data/invcdf_ju1_dls08_static_wcmp} \statictoewcmp
\begin{subfigure}{0.32\textwidth}
        \centering
\begin{tikzpicture}
\begin{axis}[ylabel = \footnotesize Percentile, xlabel = \footnotesize Bandwidth Tax, ymin=0, xmin=1, xmax=2, ymax=100, ylabel near ticks, width=1.15\linewidth, height=4.5cm, ylabel shift = -3pt, xlabel shift = -4pt, xtick pos=left, ytick pos=left, grid, thick, mark options={solid},xtick={1, 1.2, 1.4, 1.6, 1.8, 2}, xticklabels={0, 0.2, 0.4, 0.6, 0.8, 1}]
\addplot[color=lightgray, mark=asterisk,mark repeat={20}] table[y=percentile,x=ave_hop_count] from \robusttoeperfectte ;
\addplot[color=lightgray, mark=square, mark repeat={20}, densely dashed] table[y=percentile,x=ave_hop_count] from \statictoeperfectte ;
\addplot[color=orange, mark=asterisk, mark repeat={20}] table[y=percentile,x=ave_hop_count] from \maxtoemaxte ;
\addplot[color=orange, mark=square, mark repeat={20}, densely dashed] table[y=percentile,x=ave_hop_count] from \staticavete ;
\addplot[color=teal, mark=asterisk, mark repeat={20}] table[y=percentile,x=ave_hop_count] from \robusttoerobustte ;
\addplot[color=teal, mark=square, mark repeat={20}, densely dashed] table[y=percentile,x=ave_hop_count] from \statictoerobustte ;
\addplot[color=black, mark=asterisk, mark repeat={20}] table[y=percentile,x=ave_hop_count] from \robusttoewcmp ;
\addplot[color=black, mark=square, mark repeat={20}, densely dashed] table[y=percentile,x=ave_hop_count] from \statictoewcmp ;
\end{axis}
\end{tikzpicture}
    \end{subfigure}
    ~ 
    \begin{subfigure}{0.32\textwidth}
        \centering
\begin{tikzpicture}
\begin{axis}[ylabel = \footnotesize Percentile, xlabel = \footnotesize MLU, ymin=0, ymax=100, xmax=2, ylabel near ticks,width=1.15\linewidth, height=4.5cm, ylabel shift = -3pt, xlabel shift = -4pt, xtick pos=left, ytick pos=left, grid, thick, xmin=0, mark options={solid}, ]
\addplot[color=lightgray, mark=asterisk,mark repeat={20}] table[y=percentile,x=mlu] from \robusttoeperfectte ;
\addplot[color=lightgray, mark=square, mark repeat={20}, densely dashed] table[y=percentile,x=mlu] from \statictoeperfectte ;
\addplot[color=orange, mark=asterisk, mark repeat={20}] table[y=percentile,x=mlu] from \maxtoemaxte ;
\addplot[color=orange, mark=square, mark repeat={20}, densely dashed] table[y=percentile,x=mlu] from \staticavete ;
\addplot[color=teal, mark=asterisk, mark repeat={20}] table[y=percentile,x=mlu] from \robusttoerobustte ;
\addplot[color=teal, mark=square, mark repeat={20}, densely dashed] table[y=percentile,x=mlu] from \statictoerobustte ;
\addplot[color=black, mark=asterisk, mark repeat={20}] table[y=percentile,x=mlu] from \robusttoewcmp ;
\addplot[color=black, mark=square, mark repeat={20}, densely dashed] table[y=percentile,x=mlu] from \statictoewcmp ;
\end{axis}
\end{tikzpicture}
\end{subfigure}
    ~
\begin{subfigure}{0.32\textwidth}
\centering
\begin{tikzpicture}
\begin{axis}[ylabel = \footnotesize Percentile, xlabel = \footnotesize Median LU, ymin=0, ymax=100, ylabel near ticks,width=1.15\linewidth, height=4.5cm, ylabel shift = -3pt, xlabel shift = -4pt, legend columns=1, xtick pos=left, ytick pos=left, grid, thick, xmin=0, mark options={solid}, legend style={at={(0.43,0.97)}, anchor=north west}]
\addplot[color=lightgray, mark=asterisk,mark repeat={20}] table[y=percentile,x=lu50] from \robusttoeperfectte ;
\addlegendentry{\tiny METTEOR + Ideal}
\addplot[color=lightgray, mark=square, mark repeat={20}, densely dashed] table[y=percentile,x=lu50] from \statictoeperfectte ;
\addlegendentry{\tiny Unif. mesh + Ideal}
\addplot[color=orange, mark=asterisk, mark repeat={20}] table[y=percentile,x=lu50] from \maxtoemaxte ;
\addlegendentry{\tiny METTEOR + TE-S}
\addplot[color=orange, mark=square, mark repeat={20}, densely dashed] table[y=percentile,x=lu50] from \staticavete ;
\addlegendentry{\tiny Unif. mesh + TE-S}
\addplot[color=teal, mark=asterisk, mark repeat={20}] table[y=percentile,x=lu50] from \robusttoerobustte ;
\addlegendentry{\tiny METTEOR + TE-M}
\addplot[color=teal, mark=square, mark repeat={20}, densely dashed] table[y=percentile,x=lu50] from \statictoerobustte ;
\addlegendentry{\tiny Unif. mesh + TE-M}
\addplot[color=black, mark=asterisk, mark repeat={20}] table[y=percentile,x=lu50] from \robusttoewcmp ;
\addlegendentry{\tiny METTEOR + VLB}
\addplot[color=black, mark=square, mark repeat={20}, densely dashed] table[y=percentile, x=lu50] from \statictoewcmp ;
\addlegendentry{\tiny Unif. mesh + VLB}
\end{axis}
\end{tikzpicture}
\end{subfigure}
\vspace{-13pt}
\caption{\small Percentile plot of network performance with different TE solutions. Solid lines denote {\metteor } topology performance, while dashed lines represent uniform mesh.} 
\vspace{-12pt}
\label{fig:non_ideal_routing}
\end{figure*}

\subsection{Impact of Reconfiguration Frequency}\label{subsection:robustness_to_traffic_change}
Clearly, the frequency of topology-reconfiguration is a key factor in not just performance, but also the implementation and management complexity. The evaluations on {\metteor } in \S\ref{subsection:ideal_routing} are based on biweekly reconfigurations. Here, we study the interplay between topology reconfiguration frequency and performance, by comparing {\metteor}’s to other single-traffic based methods used in prior ToE works~\cite{ghobadi2016projector, farrington2011helios, hamedazimi2014firefly, cthrough_wang_2011}. As in \S\ref{subsection:compare_prediction}, we compare {\metteor } against two other single-traffic ToE approaches: $Ave$ and $Max$. $Ave$ derives its traffic estimator by taking the average of its historical traffic snapshots, while $Max$ derives its traffic estimator by taking the element-wise historical max. 

Results in Fig. \ref{fig:reconfig_frequency_sensitivity} show that in terms of tail MLU and bandwidth tax, {\metteor } generally outperforms other single traffic-based ToE approaches given the same reconfiguration frequency. As {\metteor } optimizes topology based on multiple estimated demands, it is more effective in covering future demands that resemble at least one of the predicted traffic used for topology-optimization. Furthermore, considering multiple traffic matrices when optimizing topology makes it less likely to overfit the logical topology to any single traffic demand, thereby reducing the performance penalty due to poor predictions. Note that {\metteor } exhibits very little change in tail performance even at lower reconfiguration frequencies, further highlighting the topology’s robustness to traffic changes over time. This feature allows DCN operators gain much of the benefits of topology engineering even with sporadic reconfigurations. 

\subsection{A Discussion On Suboptimal Routing}\label{subsection:non_ideal_routing}

Our evaluations so far have been based on ideal load balancing that can respond instantaneously to current traffic demands with a set of optimal routing weights that minimizes MLU. This allows us to analyze the merits of different topologies, irrespective of routing-induced suboptimality\footnote{Evaluations based on ideal MCF load balancing that minimizes MLU have been similarly done in~\cite{hamedazimi2014firefly, farrington2011helios, cthrough_wang_2011}.}. While close-to-optimal TEs have been demonstrated in the past (e.g. MicroTE~\cite{benson2011microte}), they are all adaptive algorithms that operate at very fine timescales (e.g. sub-seconds), which may inflict huge management overheads to the SDN controller~\cite{kumar2018semi}.

Since {\metteor } operates on coarse timescales, this led us to question whether it can be paired with a coarse-grained TE to work well. As routing can no longer react to all TMs optimally, this brings routing-induced suboptimality into consideration. For this evaluation, we use 3 different  load-balancing schemes that represent a large class of TE algorithms: 1) single-traffic TE, 2) Valiant load balancing (VLB), and 3) multi-traffic TE, and compare their performance when applied on uniform mesh and {\metteor} ($\kappa = 4$) topologies.



\noindent \textbf{Single-Traffic TE (TE-S): }
TE-S computes routing weights using an MCF that minimizes MLU for a single predicted TM. The TE-S results in Fig. \ref{fig:non_ideal_routing} updates routing weights every 5 minutes, based on the average traffic matrix over the past hour. We can see that TE-S performs well on average, but clearly suffers at the tail for both MLU and bandwidth tax. 

\noindent \textbf{Valiant load balancing (VLB): }
As a traffic-agnostic load-balancing algorithm, the canonical VLB derives its robustness by splitting traffic among many indirect paths at random. Our version of VLB splits traffic among direct and indirect paths, weighted by path capacity. However, because VLB sends a large portion of traffic via indirect paths, it exhibits very poor bandwidth tax performance in Fig. \ref{fig:non_ideal_routing}. Clearly, VLB’s indiscriminate traffic-splitting policy prevents it from favoring the tax-free direct paths, which makes VLB a poor choice of load-balancing for {\metteor}.



\noindent \textbf{Multi-Traffic TE (TE-M): }
TE-M is essentially a traffic engineering analog of {\metteor}. On a high-level, TE-M picks multiple representative traffic matrices based on historical measurements, and computes a set of routing weights that minimizes MLU for all of the predicted traffic matrices. The full formulation is in \ref{appendix:multi_traffic_robust_te}. Fig. \ref{fig:non_ideal_routing} shows the performance of TE-M which updates routing weights every 5 minutes, based on 4 representative TMs chosen from the past hour. Clearly, TE-M retains an impressive bandwidth tax when used with {\metteor}, while achieving better tail MLU than TE-S and VLB. This indicates that using multiple TMs can improve routing robustness under uncertainty.

\noindent \textbf{Uniform Mesh vs. {\metteor}: }
Recall from \S\ref{subsection:ideal_routing} that {\metteor } improves bandwidth tax over static uniform mesh, without sacrificing tail MLU. Unfortunately, this is no longer true when TE is suboptimal. Based on Fig. \ref{fig:non_ideal_routing}, we can see that while {\metteor } still outperforms a uniform mesh in bandwidth tax, it is more prone to a long MLU tail~\footnote{Note that the MLUs up to 99.9 percentile values are roughly the same as those of the uniform mesh.}. In fact, there is a tradeoff between average bandwidth tax and tail MLU. A ``topology + routing’’ solution with better average bandwidth tax, tends to have a longer MLU tail.

After analyzing the TM snapshots that caused long MLU tails, we found that the leading cause to be the sudden traffic bursts between pairs of pods thought to be ``cold’’, rather than an increase in traffic at the hotspots. To improve {\metteor}’s tail MLU, we need to over-provision some capacity to the ``cold’’ pod pairs. One interesting future work is to investigate the possibility of improving the above tradeoff with proper capacity-overprovisioning.


\noindent \textbf{Impact of Routing Update Frequency: }
We found that, without ideal routing, {\metteor}’s MLU exhibit long-tailed behavior, even if we use TE-M that updates routing weights every 5 minutes. Readers may wonder if the tail MLU can be improved with more frequent routing updates. However, we do not have data finer than 5 minutes. Instead, we evaluate TE-M under 4 different frequencies, ranging from once every 5 minutes to once every 2 days, and study the trend.
 
From Fig. \ref{fig:routing_uncertainty}, we can see that bandwidth tax is virtually unaffected by the routing frequency, with {\metteor} still consistently outperforming a uniform mesh. Tail MLU does improve as routing-update frequency increases. We also plot Fig. \ref{fig:overflow_probability} showing the percentage of TMs that can be supported by the underlying topology. Clearly, {\metteor} works better with more frequent routing updates. 



\noindent \textbf{Summary: }
Under suboptimal TE, {\metteor } still outperforms a uniform mesh in terms of bandwidth tax, though it is more susceptible to long-tailed MLUs. Updating routing weights more frequently helps the network respond better to traffic bursts, and improve tail MLU. Adaptive TE may be necessary to realize the full potential of {\metteor}.

\noindent \textbf{Note: }
Due to restrictions on our access to the actual TM traces, the evaluations in \S\ref{subsection:non_ideal_routing} were done using approximate reconstructions of the TM traces used in \S\ref{subsection:ideal_routing} and \S\ref{subsection:robustness_to_traffic_change}. We reconstructed each TM from the first 5 principle-component projections from our PCA study, dropping higher-order terms.  This by itself yields only a normalized approximation to the actual traffic matrices, so we inferred the correct scale factor based on MLU values from the ideal ToE results. This is a lossy reconstruction.



\begin{figure}[t!]
\hspace{-15pt}
\begin{subfigure}{0.45\columnwidth}
\begin{tikzpicture}
\begin{axis}[boxplot/draw direction=y, width=1.38\linewidth, height=3.7cm,
axis y line=left, axis x line*=bottom, xtick={1.5, 4.5, 7.5, 10.5, 13.5}, xticklabels={Ideal, 5 mins, 1 Hour, 1 Day, 2 Days}, xlabel near ticks, xlabel shift =-4pt, legend style={draw = none, fill = none, at={(0.4,1.05)}, anchor=north,legend columns=2}, legend image code/.code={ \draw[#1, draw=none] (0cm,-0.1cm) rectangle (0.2cm,0.05cm); }, cycle list={{fill=teal},{fill=gray}}, ymajorgrids, ymax=1.8,ymin=1, ytick={1, 1.2, 1.4, 1.6, 1.8}, yticklabels={0, 0.2, 0.4, 0.6, 0.8},  every axis plot/.append style={fill,fill opacity=0.6}]
\addplot+[boxplot prepared={draw position=1, lower whisker=1.0604, lower quartile=1.1098, median=1.1702, upper quartile=1.2721, upper whisker=1.4723}] coordinates {};
\addplot+[boxplot prepared={draw position=2,lower whisker=1.1920, lower quartile=1.3177, median=1.4466, upper quartile=1.5286, upper whisker=1.6809}] coordinates {};
\addplot+[boxplot prepared={draw position=4, lower whisker=1.0472, lower quartile=1.1123, median=1.1888, upper quartile=1.3210, upper whisker=1.48153357516}] coordinates {};
\addplot+[boxplot prepared={draw position=5,lower whisker=1.1846, lower quartile=1.3776, median=1.4884, upper quartile=1.5561, upper whisker=1.69417}] coordinates {};
\addplot+[boxplot prepared={draw position=7,lower whisker=1.06194760468, lower quartile= 1.10860539116, median=1.18596335939, upper quartile= 1.31527091751, upper whisker= 1.51577577715}]  coordinates {};
\addlegendentry{\tiny {\metteor}\;};
\addplot+[boxplot prepared={draw position=8,lower whisker= 1.19279,, lower quartile= 1.37261948215, median= 1.48459906174, upper quartile=1.55468097757, upper whisker= 1.70021000209, }]  coordinates {};
\addlegendentry{\tiny Unif. Mesh\;};
\addplot+[boxplot prepared={draw position=10,lower whisker= 1.05317568698, lower quartile=1.12391110086, median=1.19002844447, upper quartile=1.32184819741, upper whisker=1.4804621032,  }] coordinates {};
\addplot+[boxplot prepared={draw position=11,lower whisker=1.24224650011, lower quartile= 1.3951920791, median=1.49864918086, upper quartile=1.5594040132, upper whisker= 1.6825020371, }] coordinates {};
\addplot+[boxplot prepared={draw position=13,lower whisker= 1.04792650576, lower quartile=1.11992807497, median=1.18949223203, upper quartile= 1.31354315713, upper whisker= 1.50164582805,  }] coordinates {};
\addplot+[boxplot prepared={draw position=14,lower whisker=1.18752882673, lower quartile= 1.39001818482, median=1.49724846507, upper quartile= 1.5575952788, upper whisker= 1.71602601755, }] coordinates {};
\end{axis}
\end{tikzpicture}
\vspace{-16pt}
\caption{\small Bandwidth Tax}
\end{subfigure}
~
\begin{subfigure}{0.45\columnwidth}
\centering
\begin{tikzpicture}
\begin{axis}[boxplot/draw direction=y, width=1.38\linewidth, height=3.7cm,
axis y line=left, axis x line*=bottom, xtick={1.5, 4.5, 7.5, 10.5, 13.5}, xticklabels={Ideal, 5 mins, 1 Hour, 1 Day, 2 Days}, xlabel near ticks, ylabel shift =-4pt, legend style={draw = none, fill = none, at={(0.4,1.05)}, anchor=north,legend columns=2}, legend image code/.code={ \draw[#1, draw=none] (0cm,-0.1cm) rectangle (0.2cm,0.05cm); }, cycle list={{fill=teal},{fill=gray}}, ymajorgrids, ymax=2.2, every axis plot/.append style={fill,fill opacity=0.6}]
\addplot+[boxplot prepared={draw position=1,lower whisker=0.0284, lower quartile=0.0907, median=0.1612, upper quartile=0.3036, upper whisker=0.6938}]  coordinates {};
\addplot+[boxplot prepared={draw position=2,lower whisker=0.0284, lower quartile=0.0872, median=0.1581, upper quartile=0.3028, upper whisker=0.7057}]  coordinates {};
\addplot+[boxplot prepared={draw position=4,lower whisker=0.0332, lower quartile=0.1113, median=0.2023, upper quartile=0.3956, upper whisker=1.5351}]  coordinates {};
\addplot+[boxplot prepared={draw position=5,lower whisker=0.0333, lower quartile=0.1065, median=0.1984, upper quartile=0.3914, upper whisker=1.0391}]  coordinates {};
\addplot+[boxplot prepared={draw position=7,lower whisker= 0.0340561497139, lower quartile= 0.117976866914, median=0.218581776332, upper quartile= 0.50651474172, upper whisker= 1.52936609804}]  coordinates {};
\addlegendentry{\tiny {\metteor}\;};
\addplot+[boxplot prepared={draw position=8,lower whisker= 0.0351963201276, lower quartile= 0.113696181, median= 0.215880565711, upper quartile= 0.505023734053, upper whisker= 1.4429124817, }]  coordinates {};
\addlegendentry{\tiny Unif. Mesh\;};
\addplot+[boxplot prepared={draw position=10,lower whisker= 0.0673869446814, lower quartile=0.12563914134, median=0.238629545473, upper quartile= 0.562122644577, upper whisker= 1.57891878931,  }] coordinates {};
\addplot+[boxplot prepared={draw position=11,lower whisker=0.0610074166247, lower quartile= 0.121621411521, median=0.224344076842, upper quartile= 0.551076674226, upper whisker= 1.39220556783, }] coordinates {};
\addplot+[boxplot prepared={draw position=13,lower whisker= 0.0776470588235, lower quartile=0.129235733982, median=0.246385077466, upper quartile= 0.599687654017, upper whisker= 1.8518748875,  }] coordinates {};
\addplot+[boxplot prepared={draw position=14,lower whisker=0.0644918827078, lower quartile= 0.126022994906, median=0.234687299644, upper quartile= 0.565332166158, upper whisker= 1.4435546875, }] coordinates {};
\end{axis}
\end{tikzpicture}
\vspace{-16pt}
\caption{\small MLU}
\end{subfigure}
\vspace{-14pt}
\caption{\small TE-M performance on {\metteor } ($\kappa = 4$) and uniform mesh. Whiskers denote 100-th and 0-th percentile; box represents 95-th and 5-th percentile.} 
\label{fig:routing_uncertainty}
\end{figure}

\begin{figure}[t!]
\begin{tikzpicture}
\begin{axis}[ybar, width=1\columnwidth, height=3.7cm, enlarge x limits=0.20, bar width=6pt, legend style={at={(-0.18,1.29)}, anchor=north west,legend columns=-1}, ylabel={\footnotesize $Pr$(MLU $\leq$ 1)}, ylabel shift =-2pt,   ytick={0.985, 0.99, 0.995, 1}, symbolic x coords={5 mins,1 Hour,1 Day,2 Days}, xtick=data, nodes near coords, every node near coord/.append style={font=\tiny, rotate=90, anchor=west}, nodes near coords align={vertical},  nodes near coords style={/pgf/number format/.cd,precision=4},  ymin=0.985, ymax=1.007, yticklabel style={/pgf/number format/fixed, /pgf/number format/precision=4}, xtick pos=left, ytick pos=left, ymajorgrids, xtick align=inside]
\addplot[fill=teal] coordinates {(5 mins,0.999723778707) (1 Hour,0.998900512982) (1 Day,0.998046720859) (2 Days, 0.996073711625)};
\addlegendentry{\tiny METTEOR + TE-M}
\addplot[fill=teal, postaction={pattern=north east lines},]  coordinates {(5 mins,0.99855970325) (1 Hour,0.9956593796862) (1 Day,0.989207639492) (2 Days, 0.98826059506)};
\addlegendentry{\tiny METTEOR + TE-S}
\addplot[fill=lightgray] coordinates {(5 mins,0.999942040692) (1 Hour,0.999639386578) (1 Day,0.998683507594) (2 Days, 0.997908832783)};
\addlegendentry{\tiny Unif. Mesh + TE-M}
\addplot[fill=lightgray, postaction={pattern=north east lines},] coordinates {(5 mins,0.999886122077) (1 Hour,0.998880200425) (1 Day,0.997557811188) (2 Days, 0.996738546543)};
\addlegendentry{\tiny Unif. Mesh + TE-S}
\end{axis}
\end{tikzpicture}
\vspace{-23pt}
\caption{\small Probability of traffic demands met, as a function of TE update frequency.}
\label{fig:overflow_probability}
\vspace{-15pt}
\end{figure}

\section{Packet Level Simulations}\label{subsection:fine_grained_netbench_simulation}
The evaluations thus far have focused on macroscopic metrics like link utilization and bandwidth tax. As important as these metrics are to DCN operators~\cite{dukkipati2006flowcompletiontime}, their implications on application-level metrics such as flow completion time (FCT), are not immediately clear. To test how FCT at finer timescales relates to macroscopic metrics, we use the NetBench~\cite{netbench} packet-level simulator. The simulation emulates 2 seconds of real world time. The flow size distribution is based on a data mining workload from previous works~\cite{alizadeh2013pfabric}. Flows arrive following a Poisson process. We assume that the network links have capacity of 100Gbps, and the server-to-server latency is 600ns. We choose at random a TM to derive the communication probability between pods. 

Our first set of simulations seeks to study the effects of different MLUs on FCT. First, a production TM snapshot is chosen at random. Next, 3 different logical topologies were generated via random sampling, such that routing the same traffic matrix over each logical topology with MCF results to 3 different MLUs. We enforce the routing weights obtained from MCF such that the bandwidth tax is 0.2 for all 3 logical topologies\footnote{This can be easily done via MCF by constraining 20\% of the total traffic to traverse indirect paths.} to remove confounding variables. Table \ref{table:mlu_perf} shows that larger MLU leads to longer-tailed FCT. Intuitively, a lower MLU means less link congestion, which ultimately helps more flows to complete as more traffic may traverse the network within a given amount of time.

We similarly study how differences in bandwidth tax may affect flow level performance, given the same MLU of 0.45. Fig. \ref{fig:fct_diff_ahc} shows the FCT as a function of bandwidth tax. 
Note that while there is little difference in FCT for larger flows, the small flows have shorter FCT when the bandwidth tax is low. Small flows are more latency-sensitive, and hence their FCTs are more likely to be affected by a high bandwidth tax. Fig. \ref{fig:tcp_resend} shows that packet drop could happen in shallow-buffered data centers even before MLU reaches 1, and higher bandwidth tax leads to higher occurrences of TCP resends, which is detrimental to the throughput of small flows while waiting for packet timeout. Hence, bandwidth tax is equally important for DCN performance as MLU.


\begin{figure}[t!]
\pgfplotstableread{plot_data/fct_diffahc_ahc1p1.txt} \datafctone
\pgfplotstableread{plot_data/fct_diffahc_ahc1p3.txt} \datafctthree
\pgfplotstableread{plot_data/fct_diffahc_ahc1p7.txt} \datafctseven
\pgfplotstableread{plot_data/fct_diffahc_ahc1p8.txt} \datafctnine
\pgfplotstableread{plot_data/diff_ahc_num_tcp_resends.txt} \datatcpresend
\hspace{-8pt}
\begin{subfigure}{0.52\columnwidth}
\centering
\vspace{-6pt}
\begin{tikzpicture}
\begin{axis}[xlabel = \footnotesize FCT (ns), xtick={4, 5, 6, 7, 8}, xticklabels={$10^4$, $10^5$, $10^6$, $10^7$, $10^8$}, ylabel = \footnotesize CDF, xmin=3.5, xmax=8, ymax=1., ylabel near ticks, width=1.2\linewidth, height=4.5cm, legend style={at={(0.43,0.58)},anchor=north west,}, cycle list={{color = red}, {color=green}, {color=blue}, {color=black}}, ylabel shift = -3pt, xlabel shift = -4pt, xtick pos=left, ytick pos=left, grid, thick]
\addplot+[mark=+,mark repeat={3}] table[x=x, y=y] from \datafctone ;
\addlegendentry{\scriptsize BTax : 0.1}
\addplot+[mark=square,mark repeat={3}] table[x=x,y=y] from \datafctthree ;
\addlegendentry{\scriptsize BTax : 0.4}
\addplot+[mark=diamond, mark repeat={3}] table[x=x,y=y] from \datafctseven;
\addlegendentry{\scriptsize BTax : 0.7}
\addplot+[mark=asterisk, mark repeat={3}] table[x=x,y=y] from \datafctnine;
\addlegendentry{\scriptsize BTax : 0.9}
\end{axis}
\end{tikzpicture}
\vspace{-18pt}
\caption{\small Flow completion time (FCT)}
\label{fig:fct_diff_ahc}
\end{subfigure}
~
\begin{subfigure}{0.42\columnwidth}
\centering
\begin{tikzpicture}
\begin{axis}[xlabel = \footnotesize Bandwidth Tax, ylabel = \footnotesize TCP Resends Per Flow, xmin=1, xmax=2, ylabel near ticks, width=1.25\linewidth, height=4.5cm, cycle list={{color = black}}, ylabel shift = -4pt, xlabel shift = -4pt, xtick pos=left, ytick pos=left, grid, thick, xtick={1, 1.2, 1.4, 1.6, 1.8, 2}, xticklabels={0, 0.2, 0.4, 0.6, 0.8, 1}]
\addplot+[mark=diamond, mark options={solid}, densely dashed] table[x=x, y=y] from \datatcpresend ;
\end{axis}
\end{tikzpicture}
\vspace{-15pt}
\caption{\small TCP resend per flow}
\label{fig:tcp_resend}
\end{subfigure}
\vspace{-12pt}
\caption{\small Flow-level performance for data-mining workload under different bandwidth taxes.} 
\vspace{-9pt}
\label{fig:netbench_diff_ahc}
\end{figure}

\begin{table}[!t]
\footnotesize
\begin{tabular}{p{0.6cm}p{1.9cm}p{2.1cm}p{2.3cm}}
\hline
\textbf{MLU} & \textbf{99th \%tile FCT} & \textbf{99.9th \% tile FCT} & \textbf{99.99-th \% tile FCT} \\
\hline
\hline
0.5 & 154ms & 331ms & 928ms \\
\hline
1.0 & 103ms & 379ms & Incomplete \\
\hline
1.5 & 118ms & 414ms & Incomplete \\
\hline
\end{tabular}
\caption{\small Tail FCT performance of routing the same traffic matrix with different MLU, but fixed bandwidth tax.} 
\vspace{-17pt}
\label{table:mlu_perf}
\end{table}

\section{Conclusion}\label{section_conclusion}

We present {\metteor}, a robust topology engineering (ToE) approach that works for off-the-shelf OCSs. Unlike previous ToE solutions that react to every traffic change, {\metteor} designs logical topologies based on multiple representative TMs extracted from the slow-varying traffic clusters. As a result, {\metteor} can obtain most of the benefits of an ideal ToE, even with infrequent reconfiguration on the order of weeks. Reconfiguring topology at such low frequencies will lead to a lower technological barrier to ToE deployment, paving a path toward the incremental adoption of reconfigurable networks in commercial data centers.

\bibliographystyle{ACM-Reference-Format}
\clearpage
\bibliography{reference}
\clearpage
\begin{appendices}
\section{Additional Simulation Results}\label{appendix_additional_simulation_results}
Here, we show the complete set of simulation results in Fig. \ref{fig:all_results_boxplot} , comparing {\metteor } against a uniform mesh expanders and ideal ToE. Similar to the settings in \S\ref{subsection:ideal_routing}, {\metteor } reconfigures logical topology on a biweekly-basis, using 4 traffic clusters ($\kappa$ = 4) computed from 2-weeks’ worth historical traffic matrix snapshots.

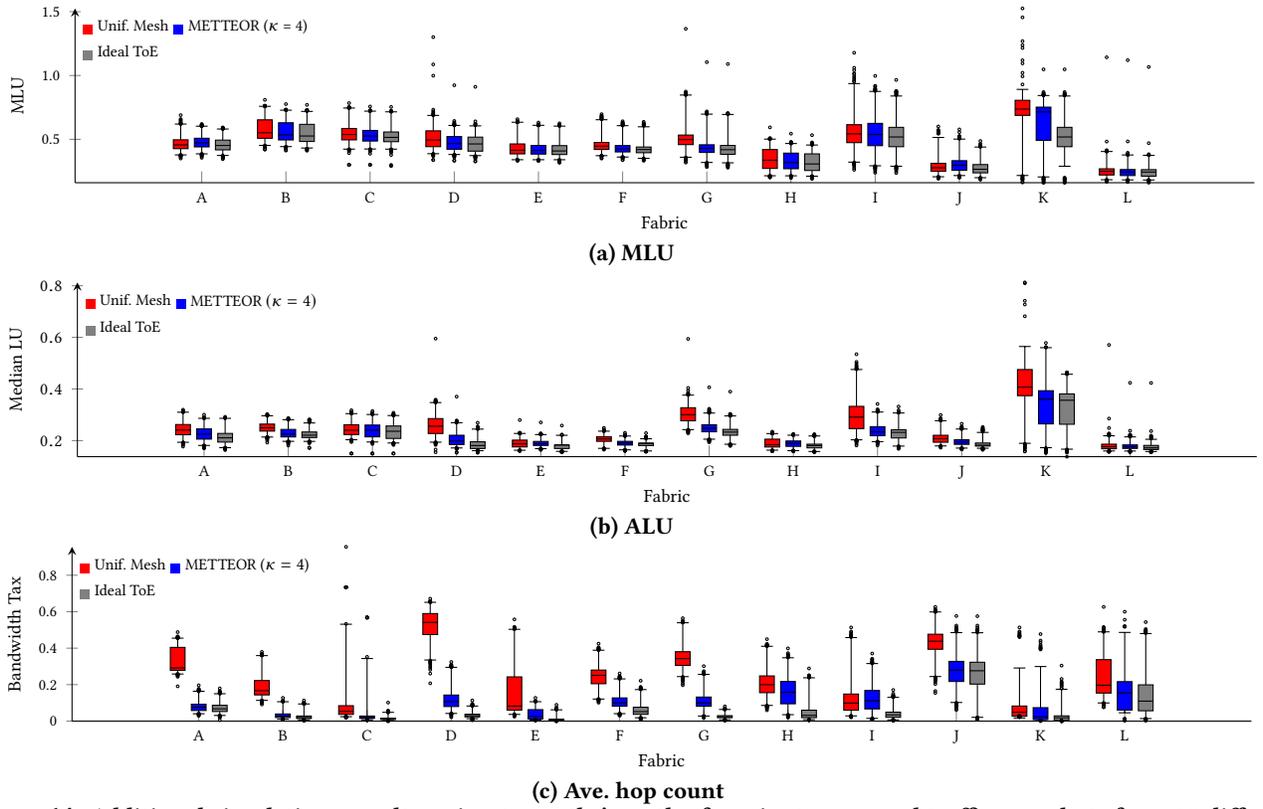
\begin{figure*}[!th]
\begin{subfigure}[c]{0.97\textwidth}
\centering
\begin{tikzpicture}
\begin{axis}[
boxplot/draw direction=y,
width=1.0\linewidth,
height=3.9cm,
axis y line=left, 
axis x line*=bottom,
xtick={2, 6, 10, 14, 18, 22, 26, 30, 34, 38, 42, 46},
xticklabels={A, B, C, D, E, F, G, H, I, J, K, L},
ytick={0, 0.5, 1.0, 1.5, 2.0},
yticklabels={0, 0.5, 1.0, 1.5, 2.0},
ylabel=\scriptsize MLU,
xlabel= \scriptsize Fabric,
xlabel shift = -3pt,
legend style={draw=none, fill=none, at={(0.0,1.0)}, anchor=north west,legend columns=2},
legend cell align=left, 
    legend image code/.code={
    \draw[#1, draw=none] (0cm,-0.1cm) rectangle (0.12cm,0.03cm);
},
cycle list={{fill=red},{fill=blue}, {fill=gray}}, 
]
\addplot+ [mark=o, mark size = 0.5, mark color=grey, 
boxplot prepared={
	draw  position = 1,
box extend=0.7,
lower whisker= 0.377881,
lower quartile= 0.426805,
median= 0.456493,
upper quartile= 0.496285,
	upper whisker= 0.619444,
}] 
table[row sep=\\,y index=0] {0.620679\\ 0.620988\\ 0.623457\\ 0.625772\\ 0.628858\\ 0.636111\\ 0.645679\\ 0.651389\\ 0.658642\\ 0.689506\\ 0.376959\\ 0.375933\\ 0.374923\\ 0.373782\\ 0.372443\\ 0.371063\\ 0.369034\\ 0.365299\\ 0.358595\\ 0.349725\\ };

\addlegendentry{\tiny Unif. Mesh};

\addplot+ [mark=o, mark size = 0.5, mark color=grey, 
boxplot prepared={
	draw  position = 2,
box extend=0.7,
lower whisker= 0.38625,
lower quartile= 0.44069,
median= 0.472845,
upper quartile= 0.508362,
	upper whisker= 0.602802,
}] 
table[row sep=\\,y index=0] {0.603276\\ 0.603578\\ 0.604397\\ 0.60569\\ 0.606207\\ 0.608448\\ 0.609224\\ 0.613534\\ 0.614698\\ 0.620216\\ 0.385302\\ 0.383836\\ 0.382284\\ 0.380739\\ 0.37931\\ 0.377716\\ 0.375302\\ 0.372371\\ 0.36509\\ 0.355216\\ };
\addlegendentry{\tiny {\metteor} ($\kappa$ = 4)};

\addplot+ [mark=o, mark size = 0.5, mark color=grey, 
boxplot prepared={
	draw  position = 3,
box extend=0.7,
lower whisker= 0.373659402811,
lower quartile= 0.417375025304,
median= 0.451805631454,
upper quartile= 0.493861,
	upper whisker= 0.580621,
}]
table[row sep=\\,y index=0] {0.581594\\ 0.582371\\ 0.583573\\ 0.584105\\ 0.585013\\ 0.585883\\ 0.58705\\ 0.587333\\ 0.589349\\ 0.593266\\ 0.373358615908\\ 0.367499191745\\ 0.364812958572\\ 0.362570606572\\ 0.35899879371\\ 0.361830162666\\ 0.359184455427\\ 0.348728768671\\ 0.345244776456\\ };
\addlegendentry{\tiny Ideal ToE};

\addplot+ [mark=o, mark size = 0.5, mark color=grey, 
boxplot prepared={
	draw  position = 5,
box extend=0.7,
lower whisker= 0.452407,
lower quartile= 0.508301,
median= 0.551416,
	upper quartile= 0.652734,
	upper whisker= 0.758496,
}] 
table[row sep=\\,y index=0] {0.758936\\ 0.759229\\ 0.759424\\ 0.760059\\ 0.760645\\ 0.761035\\ 0.761426\\ 0.762305\\ 0.763623\\ 0.809521\\ 0.451203\\ 0.449834\\ 0.448672\\ 0.446971\\ 0.444772\\ 0.442822\\ 0.439697\\ 0.436432\\ 0.423071\\ 0.418838\\ };

\addplot+ [mark=o, mark size = 0.5, mark color=grey, 
boxplot prepared={
	draw  position = 6,
box extend=0.7,
lower whisker= 0.442969,
lower quartile= 0.495119,
median= 0.534945,
upper quartile= 0.629562,
	upper whisker= 0.728376,
}] 
table[row sep=\\,y index=0] {0.728786\\ 0.72906\\ 0.729243\\ 0.729836\\ 0.730383\\ 0.730748\\ 0.731113\\ 0.731934\\ 0.733166\\ 0.776049\\ 0.441875\\ 0.440709\\ 0.439414\\ 0.437546\\ 0.436113\\ 0.434023\\ 0.431283\\ 0.427617\\ 0.415859\\ 0.411875\\ };
\addplot+ [mark=o, mark size = 0.5, mark color=grey, 
boxplot prepared={
	draw  position = 7,
box extend=0.7,
lower whisker= 0.431958828136,
lower quartile= 0.483170773518,
median= 0.526008418626,
upper quartile= 0.617334089631,
	upper whisker= 0.718687366142,
}] 
table[row sep=\\,y index=0] {0.72715171284\\ 0.72257953443\\ 0.724066603291\\ 0.725670107228\\ 0.730202632371\\ 0.7223681252\\ 0.728884436268\\ 0.769766075371\\ 0.426760557213\\ 0.425415781206\\ 0.425869067165\\ 0.418993033407\\ 0.413617810687\\ 0.411875\\ };

\addplot+ [mark=o, mark size = 0.5, mark color=grey, 
boxplot prepared={
	draw  position = 9,
box extend=0.7,
lower whisker= 0.423242,
	lower quartile= 0.495703,
	median= 0.53749,
	upper quartile= 0.584961,
	upper whisker= 0.745264,
}] 
table[row sep=\\,y index=0] {0.745898\\ 0.746289\\ 0.747021\\ 0.747705\\ 0.7479\\ 0.749316\\ 0.749805\\ 0.750928\\ 0.7521\\ 0.783789\\ 0.421826\\ 0.420605\\ 0.418555\\ 0.415771\\ 0.412451\\ 0.407422\\ 0.392041\\ 0.300488\\ 0.300195\\ 0.300098\\ };

\addplot+ [mark=o, mark size = 0.5, mark color=grey, 
boxplot prepared={
	draw  position = 10,
box extend=0.7,
lower whisker= 0.42417,
	lower quartile= 0.486531,
	median= 0.525616,
	upper quartile= 0.569603,
	upper whisker= 0.720618,
}] 
table[row sep=\\,y index=0] {0.721218\\ 0.721587\\ 0.722279\\ 0.722924\\ 0.723109\\ 0.724446\\ 0.724908\\ 0.725969\\ 0.727076\\ 0.757011\\ 0.422648\\ 0.420664\\ 0.418173\\ 0.414806\\ 0.411485\\ 0.405166\\ 0.386946\\ 0.300461\\ 0.300185\\ 0.300092\\ };
\addplot+ [mark=o, mark size = 0.5, mark color=grey, 
boxplot prepared={
lower whisker= 0.423047,
lower quartile= 0.481565314201,
median= 0.515132938163,
	upper quartile= 0.557433317834,
	upper whisker= 0.715441612686,
	draw  position = 11,
box extend=0.7,}] 
table[row sep=\\,y index=0] {0.717352656497\\ 0.719285812275\\ 0.717269419344\\ 0.721785982186\\ 0.724765936051\\ 0.719619759089\\ 0.754072822055\\ 0.421533\\ 0.419887\\ 0.417871\\ 0.404660819862\\ 0.397573920774\\ 0.400474149046\\ 0.377275491439\\ 0.292094768828\\ 0.291625295918\\ 0.299258978683\\ };
\addplot+ [mark=o, mark size = 0.5, mark color=grey, 
boxplot prepared={
	draw  position = 13,
box extend=0.7,
lower whisker= 0.387437,
	lower quartile= 0.4426,
	median= 0.494657,
	upper quartile= 0.566435,
	upper whisker= 0.687326,
}] 
table[row sep=\\,y index=0] {0.688594\\ 0.690311\\ 0.695462\\ 0.698569\\ 0.707931\\ 0.719051\\ 0.731889\\ 0.998532\\ 1.086705\\ 1.300695\\ 0.386263\\ 0.385119\\ 0.382706\\ 0.381283\\ 0.379477\\ 0.37761\\ 0.375163\\ 0.367784\\ 0.339427\\ 0.337813\\ };
\addplot+ [mark=o, mark size = 0.5, mark color=grey, 
boxplot prepared={
	draw  position = 14,
box extend=0.7,
lower whisker= 0.381214,
	lower quartile= 0.422171,
	median= 0.46573,
	upper quartile= 0.521388,
	upper whisker= 0.608801,
}] 
table[row sep=\\,y index=0] {0.609746\\ 0.611017\\ 0.611533\\ 0.614244\\ 0.617666\\ 0.621121\\ 0.625228\\ 0.631877\\ 0.643807\\ 0.923507\\ 0.380076\\ 0.378285\\ 0.376388\\ 0.374182\\ 0.371666\\ 0.369369\\ 0.366349\\ 0.358072\\ 0.333555\\ 0.332181\\ };

\addplot+ [mark=o, mark size = 0.5, mark color=grey, 
boxplot prepared={
	draw  position = 15,
box extend=0.7,
lower whisker= 0.372542080231,
lower quartile= 0.40739696869,
	median= 0.463809633871,
	upper quartile= 0.517172120364,
	upper whisker= 0.605462515427,
}]
table[row sep=\\,y index=0] {0.610058291407\\ 0.618451519459\\ 0.617886013589\\ 0.622453789295\\ 0.64086342859\\ 0.911278346753\\ 0.366901232802\\ 0.364714966068\\ 0.365747107237\\ 0.368680971106\\ 0.357057103308\\ 0.358277612054\\ 0.35689754053\\ 0.344216400731\\ 0.326749538472\\ 0.32824099946\\ };

\addplot+ [mark=o, mark size = 0.5, mark color=grey, 
boxplot prepared={
	draw  position = 17,
box extend=0.7,
lower whisker= 0.340234,
lower quartile= 0.385938,
median= 0.414697,
upper quartile= 0.463867,
	upper whisker= 0.633887,
}]
table[row sep=\\,y index=0] {0.633984\\ 0.634326\\ 0.634717\\ 0.634814\\ 0.63584\\ 0.636328\\ 0.637402\\ 0.638965\\ 0.642285\\ 0.657324\\ 0.339927\\ 0.339502\\ 0.339111\\ 0.338598\\ 0.338037\\ 0.337305\\ 0.336624\\ 0.335697\\ 0.333838\\ 0.332252\\ };
\addplot+ [mark=o, mark size = 0.5, mark color=grey, 
boxplot prepared={
	draw  position = 18,
box extend=0.7,
lower whisker= 0.340217,
lower quartile= 0.384375,
	median= 0.412772,
upper quartile= 0.452083,
	upper whisker= 0.609692,
}] 
table[row sep=\\,y index=0] {0.609783\\ 0.6101\\ 0.610462\\ 0.610553\\ 0.611504\\ 0.611957\\ 0.612953\\ 0.614402\\ 0.617482\\ 0.631431\\ 0.339764\\ 0.33914\\ 0.33871\\ 0.338098\\ 0.337319\\ 0.336232\\ 0.335145\\ 0.333992\\ 0.331386\\ 0.330254\\ };
\addplot+ [mark=o, mark size = 0.5, mark color=grey, 
boxplot prepared={
	draw  position = 19,
box extend=0.7,
lower whisker= 0.340126,
	lower quartile= 0.380901634164,
	median= 0.408801291529,
	upper quartile= 0.45052201784,
	upper whisker= 0.603448298224,
}] 
table[row sep=\\,y index=0] {0.603743743267\\ 0.609543525514\\ 0.606292071879\\ 0.609938338777\\ 0.610090205157\\ 0.613916167773\\ 0.611077478359\\ 0.625031882968\\ 0.339697\\ 0.333615755683\\ 0.336701068672\\ 0.334869193802\\ 0.328024762897\\ 0.327000667718\\ 0.327828855064\\ 0.333659481734\\ 0.320455835865\\ 0.316619513914\\ };

\addplot+ [mark=o, mark size = 0.5, mark color=grey, 
boxplot prepared={
	draw  position = 21,
box extend=0.7,
lower whisker= 0.372014,
	lower quartile= 0.420606,
	median= 0.446416,
	upper quartile= 0.475512,
	upper whisker= 0.653541,
}]
table[row sep=\\,y index=0] {0.657807\\ 0.660068\\ 0.661732\\ 0.663652\\ 0.668131\\ 0.674403\\ 0.677432\\ 0.680802\\ 0.685111\\ 0.696928\\ 0.371587\\ 0.371032\\ 0.370435\\ 0.369838\\ 0.369241\\ 0.368515\\ 0.367449\\ 0.365102\\ 0.352507\\ 0.349744\\ };
\addplot+ [mark=o, mark size = 0.5, mark color=grey, 
boxplot prepared={
	draw  position = 22,
box extend=0.7,
lower whisker= 0.361173,
	lower quartile= 0.401281,
	median= 0.425155,
	upper quartile= 0.450964,
	upper whisker= 0.607007,
}] 
table[row sep=\\,y index=0] {0.608144\\ 0.610345\\ 0.611115\\ 0.612693\\ 0.616618\\ 0.621937\\ 0.624541\\ 0.627439\\ 0.631145\\ 0.641306\\ 0.360583\\ 0.360028\\ 0.359368\\ 0.358833\\ 0.35826\\ 0.357569\\ 0.355899\\ 0.353743\\ 0.339905\\ 0.338533\\ };
\addplot+ [mark=o, mark size = 0.5, mark color=grey, 
boxplot prepared={
	draw  position = 23,
box extend=0.7,
lower whisker= 0.351161398741,
	lower quartile= 0.395107714207,
	median= 0.420126953673,
	upper quartile= 0.439643825463,
	upper whisker= 0.59821325761,
}]
table[row sep=\\,y index=0] {0.607045375814\\ 0.609382332168\\ 0.616576981987\\ 0.617749823646\\ 0.617712398974\\ 0.617813019036\\ 0.625997815426\\ 0.636832029742\\ 0.347694718466\\ 0.343368549185\\ 0.342049050783\\ 0.350718739569\\ 0.335660766262\\ 0.33285188852\\ };

\addplot+ [mark=o, mark size = 0.5, mark color=grey,
	boxplot prepared={
draw position=25,
box extend=0.7,
lower whisker= 0.359619,
	lower quartile= 0.458594,
	median= 0.498047,
	upper quartile= 0.533398,
	upper whisker= 0.846875,
}] table[row sep=\\,y index=0] {0.848145\\ 0.84873\\ 0.850781\\ 0.852832\\ 0.855469\\ 0.857617\\ 0.861328\\ 0.863477\\ 0.874121\\ 1.365723\\ 0.357764\\ 0.355566\\ 0.35332\\ 0.351074\\ 0.348828\\ 0.345996\\ 0.341788\\ 0.337598\\ 0.326367\\ 0.317871\\ };

\addplot+ [mark=o, mark size = 0.5, mark color=grey,
	boxplot prepared={
draw position=26,
box extend=0.7,
lower whisker= 0.318712,
	lower quartile= 0.397316,
	median= 0.429509,
	upper quartile= 0.457209,
	upper whisker= 0.697239,
}] table[row sep=\\,y index=0] {0.698236\\ 0.698696\\ 0.700307\\ 0.701917\\ 0.703988\\ 0.705675\\ 0.708589\\ 0.710276\\ 0.718635\\ 1.104678\\ 0.317025\\ 0.315107\\ 0.313497\\ 0.311727\\ 0.309622\\ 0.307439\\ 0.304678\\ 0.300077\\ 0.292397\\ 0.282423\\ };

\addplot+ [mark=o, mark size = 0.5, mark color=grey,
	boxplot prepared={
draw position=27,
box extend=0.7,
lower whisker= 0.316031938856,
	lower quartile= 0.382472313676,
	median= 0.418750729408,
	upper quartile= 0.452220839644,
	upper whisker= 0.694860669514,
}]
table[row sep=\\,y index=0] {0.698469697935\\ 0.69744520251\\ 0.701644418461\\ 0.701708431258\\ 0.700748973532\\ 0.705453528574\\ 1.089980921\\ 0.309767484532\\ 0.300156197819\\ 0.305263624809\\ 0.29979440768\\ 0.295628422845\\ 0.294084291205\\ 0.294331901082\\ 0.294595959381\\ 0.282591771949\\ 0.278037450845\\ };

\addplot+ [mark=o, mark size = 0.5, mark color=grey,
	boxplot prepared={
draw position=29,
box extend=0.7,
lower whisker= 0.214602,
	lower quartile= 0.272119,
	median= 0.336914,
	upper quartile= 0.419922,
	upper whisker= 0.501611,
}] table[row sep=\\,y index=0] {0.501758\\ 0.501904\\ 0.502441\\ 0.502588\\ 0.503418\\ 0.504443\\ 0.504688\\ 0.506934\\ 0.509473\\ 0.593604\\ 0.213956\\ 0.213076\\ 0.212008\\ 0.211106\\ 0.210101\\ 0.208906\\ 0.207146\\ 0.205721\\ 0.203311\\ 0.201025\\ };

\addplot+ [mark=o, mark size = 0.5, mark color=grey,
	boxplot prepared={
draw position=30,
box extend=0.7,
lower whisker= 0.214195,
	lower quartile= 0.269612,
	median= 0.318099,
	upper quartile= 0.391927,
	upper whisker= 0.467134,
}] table[row sep=\\,y index=0] {0.467274\\ 0.468588\\ 0.469531\\ 0.470205\\ 0.471929\\ 0.473491\\ 0.474838\\ 0.477802\\ 0.479256\\ 0.544314\\ 0.213823\\ 0.213408\\ 0.212739\\ 0.212105\\ 0.211502\\ 0.21072\\ 0.210069\\ 0.208941\\ 0.204905\\ 0.197232\\ };

\addplot+ [mark=o, mark size = 0.5, mark color=grey,
	boxplot prepared={
draw position=31,
box extend=0.7,
lower whisker= 0.211205,
	lower quartile= 0.255859704605,
	median= 0.306736883313,
	upper quartile= 0.385985397339,
	upper whisker= 0.455370796264,
}] table[row sep=\\,y index=0] {0.464128828758\\ 0.467064932436\\ 0.455593168027\\ 0.464842940233\\ 0.460559035661\\ 0.465318437161\\ 0.462043920524\\ 0.475742157408\\ 0.469965861724\\ 0.532613411696\\ 0.210673\\ 0.209945\\ 0.209344\\ 0.208086\\ 0.20707\\ 0.205586\\ 0.203849\\ 0.202852\\ 0.200417\\ 0.189282631963\\ };

\addplot+ [mark=o, mark size = 0.5, mark color=grey,
	boxplot prepared={
draw position=33,
box extend=0.7,
lower whisker= 0.321239,
	lower quartile= 0.474103,
	median= 0.542732,
	upper quartile= 0.614589,
	upper whisker= 0.936587,
}] table[row sep=\\,y index=0] {0.950952\\ 0.954762\\ 0.966587\\ 0.98\\ 0.993095\\ 1.010873\\ 1.016667\\ 1.040873\\ 1.060238\\ 1.177936\\ 0.318331\\ 0.315282\\ 0.312511\\ 0.309372\\ 0.306445\\ 0.303082\\ 0.297946\\ 0.290201\\ 0.275032\\ 0.261094\\ };

\addplot+ [mark=o, mark size = 0.5, mark color=grey,
	boxplot prepared={
draw position=34,
box extend=0.7,
lower whisker= 0.292871,
	lower quartile= 0.450452,
	median= 0.537626,
	upper quartile= 0.624824,
	upper whisker= 0.875351,
}] table[row sep=\\,y index=0] {0.879066\\ 0.881049\\ 0.884789\\ 0.887073\\ 0.88878\\ 0.894202\\ 0.898971\\ 0.913128\\ 0.918901\\ 0.997063\\ 0.290512\\ 0.287826\\ 0.28494\\ 0.281802\\ 0.278087\\ 0.274197\\ 0.269754\\ 0.263174\\ 0.251481\\ 0.238265\\ };

\addplot+ [mark=o, mark size = 0.5, mark color=grey,
	boxplot prepared={
draw position=35,
box extend=0.7,
lower whisker= 0.289087,
	lower quartile= 0.442236,
	median= 0.518413797702,
	upper quartile= 0.593863428116,
	upper whisker= 0.839376530677,
}] table[row sep=\\,y index=0] {0.850073936731\\ 0.847753166129\\ 0.851256629753\\ 0.858230519732\\ 0.864728609791\\ 0.867601220745\\ 0.878491\\ 0.872102246819\\ 0.875598942626\\ 0.9653930495\\ 0.286694\\ 0.283796\\ 0.281348\\ 0.278223\\ 0.274683\\ 0.27082\\ 0.266596\\ 0.260325\\ 0.249268\\ 0.236322\\ };

\addplot+ [mark=o, mark size = 0.5, mark color=grey,
	boxplot prepared={
draw position=37,
box extend=0.7,
lower whisker= 0.206787,
	lower quartile= 0.249365,
	median= 0.276919,
	upper quartile= 0.311914,
	upper whisker= 0.51468,
}] table[row sep=\\,y index=0] {0.545117\\ 0.546266\\ 0.547458\\ 0.548682\\ 0.557715\\ 0.560303\\ 0.568076\\ 0.573496\\ 0.600519\\ 0.206201\\ 0.205272\\ 0.204541\\ 0.203662\\ 0.202776\\ 0.201875\\ 0.201089\\ 0.199854\\ 0.194783\\ 0.191196\\ };

\addplot+ [mark=o, mark size = 0.5, mark color=grey,
	boxplot prepared={
draw position=38,
box extend=0.7,
lower whisker= 0.216322,
	lower quartile= 0.257229,
	median= 0.29703,
	upper quartile= 0.332666,
	upper whisker= 0.50164
}] table[row sep=\\,y index=0] {0.505684\\ 0.506324\\ 0.507951\\ 0.508839\\ 0.51176\\ 0.51838\\ 0.520747\\ 0.526664\\ 0.555288\\ 0.577626\\ 0.215445\\ 0.214706\\ 0.213511\\ 0.212492\\ 0.211225\\ 0.209604\\ 0.207768\\ 0.205965\\ 0.201797\\ 0.199799\\ };

\addplot+ [mark=o, mark size = 0.5, mark color=grey,
	boxplot prepared={
draw position=39,
box extend=0.7,
lower whisker= 0.199359,
	lower quartile= 0.2375,
	median= 0.265339,
	upper quartile= 0.302295,
	upper whisker= 0.439583296563,
}] table[row sep=\\,y index=0] {0.445117\\ 0.440913394017\\ 0.443748128846\\ 0.443842718388\\ 0.44949144924\\ 0.465186\\ 0.46062805809\\ 0.487012423551\\ 0.198758\\ 0.197949\\ 0.196928\\ 0.195801\\ 0.194395\\ 0.192871\\ 0.19166\\ 0.189609\\ 0.186426\\ 0.18521\\ };

\addplot+ [mark=o, mark size = 0.5, mark color=grey,
	boxplot prepared={
draw position=41,
box extend=0.7,
lower whisker= 0.216751,
	lower quartile= 0.687691,
	median= 0.73687,
	upper quartile= 0.805668,
	upper whisker= 0.89125,
}] table[row sep=\\,y index=0] {0.93\\ 0.993125\\ 1.02125\\ 1.094375\\ 1.119063\\ 1.215625\\ 1.2375\\ 1.26875\\ 1.455312\\ 1.525313\\ 0.21296\\ 0.209652\\ 0.203249\\ 0.197118\\ 0.192863\\ 0.18846\\ 0.185444\\ 0.180348\\ 0.163827\\ 0.16332\\ };

\addplot+ [mark=o, mark size = 0.5, mark color=grey,
	boxplot prepared={
draw position=42,
box extend=0.7,
lower whisker= 0.205273,
	lower quartile= 0.492148,
	median= 0.714629,
	upper quartile= 0.752168,
	upper whisker= 0.843145,
}] table[row sep=\\,y index=0] {0.844805\\ 0.845898\\ 0.847051\\ 0.848555\\ 0.850469\\ 0.8525\\ 0.856328\\ 0.86\\ 0.870137\\ 1.047891\\ 0.199844\\ 0.194922\\ 0.190898\\ 0.18875\\ 0.185859\\ 0.182754\\ 0.176016\\ 0.17041\\ 0.161543\\ 0.160176\\ };

\addplot+ [mark=o, mark size = 0.5, mark color=grey,
	boxplot prepared={
draw position=43,
box extend=0.7,
lower whisker= 0.289087,
	lower quartile= 0.442236,
	median= 0.518413797702,
	upper quartile= 0.593863428116,
	upper whisker= 0.839376530677,
}] table[row sep=\\,y index=0] {0.844805\\ 0.845898\\ 0.847051\\ 0.848555\\ 0.850469\\ 0.8525\\ 0.856328\\ 0.86\\ 0.870137\\ 1.047891\\ 0.199844\\ 0.194922\\ 0.190898\\ 0.18875\\ 0.185859\\ 0.182754\\ 0.176016\\ 0.17041\\ 0.161543\\ 0.160176\\ };

\addplot+ [mark=o, mark size = 0.5, mark color=grey,
	boxplot prepared={
draw position=45,
box extend=0.7,
lower whisker= 0.183421,
	lower quartile= 0.220659,
	median= 0.249843,
	upper quartile= 0.269142,
	upper whisker= 0.40455,
}] table[row sep=\\,y index=0] {0.405021\\ 0.40523\\ 0.407165\\ 0.407636\\ 0.408054\\ 0.408525\\ 0.410565\\ 0.414801\\ 0.4825\\ 1.142586\\ 0.182688\\ 0.182008\\ 0.181119\\ 0.180259\\ 0.179446\\ 0.178086\\ 0.176778\\ 0.175052\\ 0.171234\\ 0.167364\\ };

\addplot+ [mark=o, mark size = 0.5, mark color=grey,
	boxplot prepared={
draw position=46,
box extend=0.7,
lower whisker= 0.181996,
	lower quartile= 0.216291,
	median= 0.243788,
	upper quartile= 0.263246,
	upper whisker= 0.377472,
}] table[row sep=\\,y index=0] {0.377612\\ 0.379338\\ 0.379757\\ 0.380131\\ 0.38055\\ 0.382229\\ 0.384229\\ 0.394714\\ 0.483362\\ 1.120095\\ 0.18153\\ 0.180727\\ 0.17988\\ 0.179151\\ 0.178479\\ 0.177239\\ 0.176114\\ 0.174504\\ 0.171035\\ 0.165485\\ };

\addplot+ [mark=o, mark size = 0.5, mark color=grey,
	boxplot prepared={
draw position=47,
box extend=0.7,
lower whisker= 0.181295,
	lower quartile= 0.210975959388,
	median= 0.243561,
	upper quartile= 0.262373,
	upper whisker= 0.374392304719,
}] table[row sep=\\,y index=0] {0.376158039822\\ 0.37776472733\\ 0.383214030211\\ 0.470101317161\\ 1.066952\\ 0.18073\\ 0.179961\\ 0.179199\\ 0.178277\\ 0.177548\\ 0.176351\\ 0.175077\\ 0.173449\\ 0.169824\\ 0.160340932243\\ };
\end{axis}
\end{tikzpicture}
\vspace{-5pt}
\caption{\small MLU}
\end{subfigure} \\

\begin{subfigure}[c]{0.97\textwidth}
\centering
\begin{tikzpicture}
\begin{axis}[
boxplot/draw direction=y,
width=1.0\linewidth,
height=3.9cm,
axis y line=left, 
axis x line*=bottom,
xtick={2, 6, 10, 14, 18, 22, 26, 30, 34, 38, 42, 46},
xticklabels={A, B, C, D, E, F, G, H, I, J, K, L},
ylabel=\scriptsize Median LU,
xlabel= \scriptsize Fabric,
xlabel shift = -3pt,
legend style={draw=none, fill=none, at={(0.0,1.0)}, anchor=north west,legend columns=2},
legend cell align=left, 
    legend image code/.code={
    \draw[#1, draw=none] (0cm,-0.1cm) rectangle (0.12cm,0.03cm);
},
cycle list={{fill=red},{fill=blue}, {fill=gray}}, 
]
\addplot+ [mark=o, mark size = 0.5, mark color=grey, 
boxplot prepared={
	draw  position = 1,
box extend=0.7,
lower whisker= 0.193526,
	lower quartile= 0.222703,
	median= 0.240682,
	upper quartile= 0.262199,
	upper whisker= 0.309991,
}] table[row sep=\\,y index=0] {0.310534\\ 0.310652\\ 0.31092\\ 0.311288\\ 0.311601\\ 0.311785\\ 0.312163\\ 0.3139\\ 0.314749\\ 0.319426\\ 0.192864\\ 0.192173\\ 0.191418\\ 0.190749\\ 0.18995\\ 0.188716\\ 0.187635\\ 0.186031\\ 0.182916\\ 0.176772\\ };
\addlegendentry{\tiny Unif. Mesh};

\addplot+ [mark=o, mark size = 0.5, mark color=grey, 
boxplot prepared={
	draw  position = 2,
box extend=0.7,
lower whisker= 0.181235,
	lower quartile= 0.205648,
	median= 0.225145,
	upper quartile= 0.245262,
	upper whisker= 0.286238,
}] table[row sep=\\,y index=0] {0.286408\\ 0.286815\\ 0.287059\\ 0.287368\\ 0.287518\\ 0.287624\\ 0.28838\\ 0.289042\\ 0.291341\\ 0.300181\\ 0.180669\\ 0.180087\\ 0.179463\\ 0.178881\\ 0.178053\\ 0.176775\\ 0.175647\\ 0.174528\\ 0.17116\\ 0.170044\\ };
\addlegendentry{\tiny {\metteor} $(\kappa = 4)$};

\addplot+ [mark=o, mark size = 0.5, mark color=grey, 
boxplot prepared={
	draw  position = 3,
box extend=0.7,
lower whisker= 0.172557,
	lower quartile= 0.194271,
	median= 0.210869,
	upper quartile= 0.227758,
	upper whisker= 0.286235,
}] table[row sep=\\,y index=0] {0.286758\\ 0.286821186899\\ 0.288831541745\\ 0.292103807469\\ 0.172254\\ 0.17187\\ 0.171439\\ 0.171129\\ 0.170642\\ 0.170136\\ 0.169613\\ 0.168594\\ 0.166539\\ 0.163369\\ };
\addlegendentry{\tiny Ideal ToE};

\addplot+ [mark=o, mark size = 0.5, mark color=grey, 
boxplot prepared={
	draw  position = 5,
box extend=0.7,
lower whisker= 0.213829,
	lower quartile= 0.236744,
	median= 0.249431,
	upper quartile= 0.263925,
	upper whisker= 0.296286,
}]  table[row sep=\\,y index=0] {0.296518\\ 0.29687\\ 0.29715\\ 0.297565\\ 0.297874\\ 0.298208\\ 0.298436\\ 0.299188\\ 0.300227\\ 0.302622\\ 0.213439\\ 0.212933\\ 0.212278\\ 0.211371\\ 0.210281\\ 0.209293\\ 0.20763\\ 0.205695\\ 0.199972\\ 0.192397\\ };

\addplot+ [mark=o, mark size = 0.5, mark color=grey, 
boxplot prepared={
	draw  position = 6,
box extend=0.7,
lower whisker= 0.197248,
	lower quartile= 0.214217,
	median= 0.225133,
	upper quartile= 0.243238,
	upper whisker= 0.281001,
}]  table[row sep=\\,y index=0] {0.281113\\ 0.281405\\ 0.281549\\ 0.281748\\ 0.282154\\ 0.282546\\ 0.282992\\ 0.283447\\ 0.284335\\ 0.287309\\ 0.1969\\ 0.19647\\ 0.195956\\ 0.195412\\ 0.194723\\ 0.193663\\ 0.19241\\ 0.191047\\ 0.186967\\ 0.180894\\ };

\addplot+ [mark=o, mark size = 0.5, mark color=grey, 
boxplot prepared={
	draw  position = 7,
box extend=0.7,
lower whisker= 0.197033,
	lower quartile= 0.211764,
	median= 0.221006,
	upper quartile= 0.233345757404,
	upper whisker= 0.268422043784,
}]  table[row sep=\\,y index=0] {0.279840240077\\ 0.279664242116\\ 0.276812164126\\ 0.280472821712\\ 0.278043262069\\ 0.27205559132\\ 0.282076616057\\ 0.272743024738\\ 0.271235258663\\ 0.278932182247\\ 0.196782\\ 0.196359\\ 0.195862\\ 0.195143\\ 0.194544\\ 0.193557\\ 0.191971\\ 0.190066\\ 0.173304346369\\ 0.170874297058\\ };

\addplot+ [mark=o, mark size = 0.5, mark color=grey, 
boxplot prepared={
	draw  position = 9,
box extend=0.7,
lower whisker= 0.203348,
	lower quartile= 0.223347,
	median= 0.240163,
	upper quartile= 0.261316,
	upper whisker= 0.30449,
}] table[row sep=\\,y index=0] {0.304834\\ 0.305276\\ 0.305486\\ 0.306522\\ 0.307195\\ 0.308657\\ 0.309902\\ 0.310369\\ 0.310898\\ 0.317671\\ 0.202952\\ 0.202524\\ 0.201962\\ 0.201227\\ 0.200472\\ 0.198728\\ 0.194486\\ 0.150034\\ 0.150015\\ 0.150008\\ };

\addplot+ [mark=o, mark size = 0.5, mark color=grey, 
boxplot prepared={
	draw  position = 10,
box extend=0.7,
lower whisker= 0.19458,
	lower quartile= 0.21453,
	median= 0.240319,
	upper quartile= 0.260176,
	upper whisker= 0.301137,
}] table[row sep=\\,y index=0] {0.3015\\ 0.302183\\ 0.302374\\ 0.303289\\ 0.304862\\ 0.305525\\ 0.306142\\ 0.307109\\ 0.307776\\ 0.314591\\ 0.194231\\ 0.193795\\ 0.193342\\ 0.192893\\ 0.192105\\ 0.190904\\ 0.187035\\ 0.150031\\ 0.150013\\ 0.150007\\ };

\addplot+ [mark=o, mark size = 0.5, mark color=grey, 
boxplot prepared={
	draw  position = 11,
box extend=0.7,
lower whisker= 0.190811,
	lower quartile= 0.208432,
	median= 0.236222,
	upper quartile= 0.256682,
	upper whisker= 0.297122,
}] table[row sep=\\,y index=0] {0.297646\\ 0.29794\\ 0.298477\\ 0.299287\\ 0.300063\\ 0.300354\\ 0.30098\\ 0.301387\\ 0.302456\\ 0.307866\\ 0.190378\\ 0.189889\\ 0.189356\\ 0.188788\\ 0.187928\\ 0.186646\\ 0.183046\\ 0.150016\\ 0.150003\\ 0.150001\\ };

\addplot+ [mark=o, mark size = 0.5, mark color=grey, 
boxplot prepared={
	draw  position = 13,
box extend=0.7,
lower whisker= 0.193876,
	lower quartile= 0.227474,
	median= 0.255713,
	upper quartile= 0.285252,
	upper whisker= 0.347326,
}] table[row sep=\\,y index=0] {0.347696\\ 0.348874\\ 0.349649\\ 0.350348\\ 0.350774\\ 0.351986\\ 0.353706\\ 0.356027\\ 0.358516\\ 0.595149\\ 0.19329\\ 0.192628\\ 0.191851\\ 0.191041\\ 0.19021\\ 0.188885\\ 0.186948\\ 0.183089\\ 0.165202\\ 0.154817\\ };

\addplot+ [mark=o, mark size = 0.5, mark color=grey, 
boxplot prepared={
	draw  position = 14,
box extend=0.7,
lower whisker= 0.171554,
	lower quartile= 0.184726,
	median= 0.198742,
	upper quartile= 0.22072,
	upper whisker= 0.269254,
}] table[row sep=\\,y index=0] {0.269692\\ 0.270511\\ 0.271221\\ 0.271549\\ 0.272106\\ 0.272905\\ 0.273698\\ 0.275601\\ 0.277939\\ 0.370417\\ 0.171345\\ 0.171104\\ 0.170797\\ 0.170519\\ 0.170104\\ 0.169609\\ 0.168724\\ 0.167059\\ 0.155324\\ 0.153034\\ };

\addplot+ [mark=o, mark size = 0.5, mark color=grey, 
boxplot prepared={
	draw  position = 15,
box extend=0.7,
lower whisker= 0.160943,
	lower quartile= 0.169262,
	median= 0.1809,
	upper quartile= 0.195441,
	upper whisker= 0.244488,
}] table[row sep=\\,y index=0] {0.244876\\ 0.24536\\ 0.247056\\ 0.248041\\ 0.248925\\ 0.249881\\ 0.250458\\ 0.251942\\ 0.254853\\ 0.269169\\ 0.160795\\ 0.160627\\ 0.160446\\ 0.160227\\ 0.159941\\ 0.159658\\ 0.159184\\ 0.157813\\ 0.154722\\ 0.152752\\ };

\addplot+ [mark=o, mark size = 0.5, mark color=grey, 
boxplot prepared={
	draw  position = 17,
box extend=0.7,
lower whisker= 0.162727,
	lower quartile= 0.175463,
	median= 0.187475,
	upper quartile= 0.201929,
	upper whisker= 0.227714,
}] table[row sep=\\,y index=0] {0.227792\\ 0.227908\\ 0.228044\\ 0.228113\\ 0.228402\\ 0.228569\\ 0.228921\\ 0.229691\\ 0.230648\\ 0.27973\\ 0.162531\\ 0.162302\\ 0.162103\\ 0.161905\\ 0.161684\\ 0.161405\\ 0.161075\\ 0.160695\\ 0.160082\\ 0.159742\\ };

\addplot+ [mark=o, mark size = 0.5, mark color=grey, 
boxplot prepared={
	draw  position = 18,
box extend=0.7,
lower whisker= 0.168926,
	lower quartile= 0.181231,
	median= 0.187779,
	upper quartile= 0.196747,
	upper whisker= 0.223293,
}] table[row sep=\\,y index=0] {0.223571\\ 0.223762\\ 0.224072\\ 0.224292\\ 0.224812\\ 0.225205\\ 0.225679\\ 0.226429\\ 0.22799\\ 0.270578\\ 0.168764\\ 0.16856\\ 0.168354\\ 0.168092\\ 0.167763\\ 0.167309\\ 0.166829\\ 0.166102\\ 0.165023\\ 0.164255\\ };

\addplot+ [mark=o, mark size = 0.5, mark color=grey, 
boxplot prepared={
	draw  position = 19,
box extend=0.7,
lower whisker= 0.158344,
	lower quartile= 0.169337,
	median= 0.183479,
	upper quartile= 0.182691054511,
	upper whisker= 0.220547357128,
}] table[row sep=\\,y index=0] {0.221110867236\\ 0.223203692487\\ 0.222525840885\\ 0.258274\\ 0.158207\\ 0.15803\\ 0.157861\\ 0.157722\\ 0.157551\\ 0.157366\\ 0.157118\\ 0.15676\\ 0.15582\\ 0.15547\\ };

\addplot+ [mark=o, mark size = 0.5, mark color=grey, 
boxplot prepared={
	draw  position = 21,
box extend=0.7,
lower whisker= 0.169946,
	lower quartile= 0.197145,
	median= 0.206301,
	upper quartile= 0.214769,
	upper whisker= 0.236497,
}] table[row sep=\\,y index=0] {0.236567\\ 0.236678\\ 0.237069\\ 0.237564\\ 0.237994\\ 0.238518\\ 0.23898\\ 0.240045\\ 0.242325\\ 0.249191\\ 0.169749\\ 0.169513\\ 0.169285\\ 0.169079\\ 0.168822\\ 0.168478\\ 0.168112\\ 0.167553\\ 0.166362\\ 0.164797\\ };

\addplot+ [mark=o, mark size = 0.5, mark color=grey, 
boxplot prepared={
	draw  position = 22,
box extend=0.7,
lower whisker= 0.164147,
	lower quartile= 0.182661,
	median= 0.189978,
	upper quartile= 0.197066,
	upper whisker= 0.217669,
}] table[row sep=\\,y index=0] {0.217853\\ 0.21817\\ 0.218378\\ 0.21867\\ 0.218983\\ 0.21909\\ 0.219435\\ 0.221848\\ 0.223161\\ 0.230017\\ 0.163924\\ 0.163563\\ 0.163261\\ 0.16295\\ 0.16258\\ 0.162221\\ 0.161688\\ 0.161191\\ 0.160386\\ 0.160027\\ };

\addplot+ [mark=o, mark size = 0.5, mark color=grey, 
boxplot prepared={
	draw  position = 23,
box extend=0.7,
lower whisker= 0.159858,
	lower quartile= 0.180835,
	median= 0.189485,
	upper quartile= 0.190226637537,
	upper whisker= 0.207078962334,
}] table[row sep=\\,y index=0] {0.21137862817\\ 0.214000178859\\ 0.207745901344\\ 0.213595568001\\ 0.215808262735\\ 0.208880130997\\ 0.214462923563\\ 0.21008990975\\ 0.214364305069\\ 0.22895964884\\ 0.159725\\ 0.159591\\ 0.159459\\ 0.159339\\ 0.15923\\ 0.159045\\ 0.158718\\ 0.158349\\ 0.157337\\ 0.156675\\ };

\addplot+ [mark=o, mark size = 0.5, mark color=grey,
	boxplot prepared={
draw position=25,
box extend=0.7,
lower whisker= 0.242594,
	lower quartile= 0.276773,
	median= 0.300675,
	upper quartile= 0.327318,
	upper whisker= 0.3762,
}] table[row sep=\\,y index=0] {0.377212\\ 0.377569\\ 0.378142\\ 0.379633\\ 0.380295\\ 0.38264\\ 0.384641\\ 0.393666\\ 0.403835\\ 0.593986\\ 0.241824\\ 0.240944\\ 0.240086\\ 0.239174\\ 0.238067\\ 0.236912\\ 0.235507\\ 0.233464\\ 0.23027\\ 0.228497\\ };

\addplot+ [mark=o, mark size = 0.5, mark color=grey,
	boxplot prepared={
draw position=26,
box extend=0.7,
lower whisker= 0.205511,
	lower quartile= 0.234176,
	median= 0.247715,
	upper quartile= 0.261453,
	upper whisker= 0.306686,
}] table[row sep=\\,y index=0] {0.307444\\ 0.307911\\ 0.309634\\ 0.309941\\ 0.310717\\ 0.312981\\ 0.316435\\ 0.320068\\ 0.323197\\ 0.406424\\ 0.204927\\ 0.20419\\ 0.203456\\ 0.202788\\ 0.201974\\ 0.201144\\ 0.199981\\ 0.198709\\ 0.19591\\ 0.193626\\ };

\addplot+ [mark=o, mark size = 0.5, mark color=grey,
	boxplot prepared={
draw position=27,
box extend=0.7,
lower whisker= 0.185917,
	lower quartile= 0.22075,
	median= 0.23293,
	upper quartile= 0.243419,
	upper whisker= 0.295612,
}] table[row sep=\\,y index=0] {0.296199\\ 0.29647\\ 0.2971\\ 0.297398\\ 0.29852\\ 0.300067\\ 0.300955\\ 0.302429\\ 0.303539\\ 0.389652\\ 0.185554\\ 0.185134\\ 0.184831\\ 0.18438\\ 0.184015\\ 0.183596\\ 0.182908\\ 0.182118\\ 0.180648\\ 0.179451\\ };

\addplot+ [mark=o, mark size = 0.5, mark color=grey,
	boxplot prepared={
draw position=29,
box extend=0.7,
lower whisker= 0.163008,
	lower quartile= 0.175199,
	median= 0.184073,
	upper quartile= 0.205542,
	upper whisker= 0.227251,
}] table[row sep=\\,y index=0] {0.227435\\ 0.227737\\ 0.22799\\ 0.228127\\ 0.228456\\ 0.228731\\ 0.228886\\ 0.229671\\ 0.231668\\ 0.234597\\ 0.16276\\ 0.162262\\ 0.161937\\ 0.16159\\ 0.161086\\ 0.160521\\ 0.159809\\ 0.159457\\ 0.158976\\ 0.158744\\ };

\addplot+ [mark=o, mark size = 0.5, mark color=grey,
	boxplot prepared={
draw position=30,
box extend=0.7,
lower whisker= 0.161018,
	lower quartile= 0.177768,
	median= 0.189321,
	upper quartile= 0.198472,
	upper whisker= 0.22093,
}] table[row sep=\\,y index=0] {0.22112\\ 0.22133\\ 0.221756\\ 0.221991\\ 0.222474\\ 0.222716\\ 0.222947\\ 0.223311\\ 0.223704\\ 0.22593\\ 0.160829\\ 0.160605\\ 0.16042\\ 0.160204\\ 0.159988\\ 0.159852\\ 0.159623\\ 0.159417\\ 0.159055\\ 0.158655\\ };

\addplot+ [mark=o, mark size = 0.5, mark color=grey,
	boxplot prepared={
draw position=31,
box extend=0.7,
lower whisker= 0.15783,
	lower quartile= 0.171818,
	median= 0.180708,
	upper quartile= 0.186161176893,
	upper whisker= 0.217922,
}] table[row sep=\\,y index=0] {0.217971\\ 0.218015\\ 0.218049\\ 0.218392\\ 0.218547\\ 0.219529\\ 0.220257\\ 0.222102\\ 0.222479\\ 0.224882\\ 0.157601\\ 0.157438\\ 0.157295\\ 0.157165\\ 0.157064\\ 0.156973\\ 0.156846\\ 0.156673\\ 0.15634\\ 0.155877\\ };

\addplot+ [mark=o, mark size = 0.5, mark color=grey,
	boxplot prepared={
draw position=33,
box extend=0.7,
lower whisker= 0.20304,
	lower quartile= 0.246492,
	median= 0.291093,
	upper quartile= 0.332617,
	upper whisker= 0.475424,
}] table[row sep=\\,y index=0] {0.476725\\ 0.477233\\ 0.479406\\ 0.480895\\ 0.484106\\ 0.488567\\ 0.490918\\ 0.498481\\ 0.50492\\ 0.534577\\ 0.20232\\ 0.201522\\ 0.200591\\ 0.199832\\ 0.198734\\ 0.197451\\ 0.195871\\ 0.193372\\ 0.186937\\ 0.181109\\ };

\addplot+ [mark=o, mark size = 0.5, mark color=grey,
	boxplot prepared={
draw position=34,
box extend=0.7,
lower whisker= 0.196147,
	lower quartile= 0.218882,
	median= 0.23379,
	upper quartile= 0.253991,
	upper whisker= 0.310664,
}] table[row sep=\\,y index=0] {0.311082\\ 0.312331\\ 0.312461\\ 0.312862\\ 0.313672\\ 0.315734\\ 0.31678\\ 0.318286\\ 0.321105\\ 0.342088\\ 0.195711\\ 0.195209\\ 0.194672\\ 0.193979\\ 0.19335\\ 0.192477\\ 0.191211\\ 0.189769\\ 0.186464\\ 0.181278\\ };

\addplot+ [mark=o, mark size = 0.5, mark color=grey,
	boxplot prepared={
draw position=35,
box extend=0.7,
lower whisker= 0.178773,
	lower quartile= 0.210241,
	median= 0.231834,
	upper quartile= 0.242052516728,
	upper whisker= 0.306984684159,
}] table[row sep=\\,y index=0] {0.308163861768\\ 0.308254260884\\ 0.31286748026\\ 0.317284689632\\ 0.331832620709\\ 0.178194\\ 0.177434\\ 0.1769\\ 0.17653\\ 0.175857\\ 0.175256\\ 0.174465\\ 0.173251\\ 0.170642\\ 0.168151\\ };

\addplot+ [mark=o, mark size = 0.5, mark color=grey,
	boxplot prepared={
draw position=37,
box extend=0.7,
lower whisker= 0.179323,
	lower quartile= 0.19363,
	median= 0.207802,
	upper quartile= 0.220565,
	upper whisker= 0.27708,
}] table[row sep=\\,y index=0] {0.277561\\ 0.278189\\ 0.278396\\ 0.279593\\ 0.280838\\ 0.281252\\ 0.282495\\ 0.284856\\ 0.286363\\ 0.299255\\ 0.179055\\ 0.178694\\ 0.17839\\ 0.178042\\ 0.177697\\ 0.177431\\ 0.176958\\ 0.176511\\ 0.174711\\ 0.172686\\ };

\addplot+ [mark=o, mark size = 0.5, mark color=grey,
	boxplot prepared={
draw position=38,
box extend=0.7,
lower whisker= 0.170977,
	lower quartile= 0.185433,
	median= 0.194282,
	upper quartile= 0.203281,
	upper whisker= 0.24485,
}] table[row sep=\\,y index=0] {0.245086\\ 0.245399\\ 0.245984\\ 0.246139\\ 0.248193\\ 0.249997\\ 0.250647\\ 0.252016\\ 0.253935\\ 0.265391\\ 0.17079\\ 0.170569\\ 0.170313\\ 0.170068\\ 0.1697\\ 0.16935\\ 0.168883\\ 0.168306\\ 0.167274\\ 0.165919\\ };

\addplot+ [mark=o, mark size = 0.5, mark color=grey,
	boxplot prepared={
draw position=39,
box extend=0.7,
lower whisker= 0.170359,
	lower quartile= 0.178765155007,
	median= 0.181982861504,
	upper quartile= 0.190683480507,
	upper whisker= 0.231667441379,
}] table[row sep=\\,y index=0] {0.236498082828\\ 0.244688293141\\ 0.241680939654\\ 0.242606539113\\ 0.237047407483\\ 0.236060000516\\ 0.239758173549\\ 0.2446758723\\ 0.24901004741\\ 0.252497640109\\ 0.170034\\ 0.169777\\ 0.169528\\ 0.169194\\ 0.168863\\ 0.168327\\ 0.167906\\ 0.167309\\ 0.165503\\ 0.164371\\ };

\addplot+ [mark=o, mark size = 0.5, mark color=grey,
	boxplot prepared={
draw position=41,
box extend=0.7,
lower whisker= 0.189471,
	lower quartile= 0.374003,
	median= 0.406849,
	upper quartile= 0.4753,
	upper whisker= 0.564591,
}] table[row sep=\\,y index=0] {0.68176\\ 0.727455\\ 0.741831\\ 0.809656\\ 0.812846\\ 0.187587\\ 0.185075\\ 0.181389\\ 0.178666\\ 0.176276\\ 0.173199\\ 0.170382\\ 0.16603\\ 0.158847\\ 0.157749\\ };

\addplot+ [mark=o, mark size = 0.5, mark color=grey,
	boxplot prepared={
draw position=42,
box extend=0.7,
lower whisker= 0.172507,
	lower quartile= 0.264742,
	median= 0.360171,
	upper quartile= 0.393304,
	upper whisker= 0.560021,
}] table[row sep=\\,y index=0] {0.561245\\ 0.56163\\ 0.562188\\ 0.578183\\ 0.17096\\ 0.169281\\ 0.167665\\ 0.165913\\ 0.164081\\ 0.162362\\ 0.160788\\ 0.159042\\ 0.154314\\ 0.151773\\ };

\addplot+ [mark=o, mark size = 0.5, mark color=grey,
	boxplot prepared={
draw position=43,
box extend=0.7,
lower whisker= 0.166726992546,
	lower quartile= 0.26329406078,
	median= 0.356271467714,
	upper quartile= 0.380938145534,
	upper whisker= 0.457076800796,
}] table[row sep=\\,y index=0] {0.45816671914\\ 0.457626992521\\ 0.464146882207\\ 0.463634512949\\ 0.461801210179\\ 0.464100732021\\ 0.159992347054\\ 0.16567\\ 0.163996\\ 0.162116\\ 0.157922589152\\ 0.154016398634\\ 0.154253\\ 0.137711857753\\ };

\addplot+ [mark=o, mark size = 0.5, mark color=grey,
	boxplot prepared={
draw position=45,
box extend=0.7,
lower whisker= 0.159775,
	lower quartile= 0.168303,
	median= 0.176289,
	upper quartile= 0.186131,
	upper whisker= 0.218669,
}] table[row sep=\\,y index=0] {0.219339\\ 0.219657\\ 0.220393\\ 0.222778\\ 0.224928\\ 0.225387\\ 0.227946\\ 0.248633\\ 0.285504\\ 0.570703\\ 0.159683\\ 0.159544\\ 0.159418\\ 0.159215\\ 0.15899\\ 0.158754\\ 0.158415\\ 0.157948\\ 0.156661\\ 0.156223\\ };

\addplot+ [mark=o, mark size = 0.5, mark color=grey,
	boxplot prepared={
draw position=46,
box extend=0.7,
lower whisker= 0.159689,
	lower quartile= 0.169729,
	median= 0.175762,
	upper quartile= 0.183393,
	upper whisker= 0.214373,
}] table[row sep=\\,y index=0] {0.215046\\ 0.216233\\ 0.217509\\ 0.21828\\ 0.218928\\ 0.219922\\ 0.221366\\ 0.223394\\ 0.238284\\ 0.423861\\ 0.159565\\ 0.159378\\ 0.159242\\ 0.158944\\ 0.158678\\ 0.15845\\ 0.158089\\ 0.157609\\ 0.155989\\ 0.155625\\ };

\addplot+ [mark=o, mark size = 0.5, mark color=grey,
	boxplot prepared={
draw position=47,
box extend=0.7,
lower whisker= 0.157768,
	lower quartile= 0.16537,
	median= 0.171922,
	upper quartile= 0.181354,
	upper whisker= 0.205019857848,
}] table[row sep=\\,y index=0] {0.213791967566\\ 0.20869764991\\ 0.213082829691\\ 0.208288057043\\ 0.219062078039\\ 0.216414157259\\ 0.236886\\ 0.42330054034\\ 0.157648\\ 0.157521\\ 0.157376\\ 0.157254\\ 0.157107\\ 0.15679\\ 0.156495\\ 0.156048\\ 0.155169\\ 0.154877\\ };
\end{axis}
\end{tikzpicture}
\vspace{-5pt}
\caption{\small ALU}
\end{subfigure} \\

\begin{subfigure}[c]{0.97\textwidth}
\centering
\begin{tikzpicture}
\begin{axis}[
boxplot/draw direction=y,
width=1.0\linewidth,
height=3.9cm,
axis y line=left, 
axis x line*=bottom,
xtick={2, 6, 10, 14, 18, 22, 26, 30, 34, 38, 42, 46},
xticklabels={A, B, C, D, E, F, G, H, I, J, K, L},
ylabel=\scriptsize Bandwidth Tax,
xlabel=\scriptsize Fabric,
xlabel shift = -3pt,
ytick={5.0, 5.6, 6.2, 6.8, 7.4, 8.0}, 
yticklabels={0,0.2,0.4,0.6,0.8,1.0},
legend style={draw=none, fill=none, at={(0.0,1.00)}, anchor=north west,legend columns=2},
legend cell align=left,
legend image code/.code={
    \draw[#1, draw=none] (0cm,-0.1cm) rectangle (0.12cm,0.03cm);
},
cycle list={{fill=red},{fill=blue}, {fill=gray}}, 
]
\addplot+ [mark=o, mark size = 0.5, mark color=grey, 
boxplot prepared={
	draw position=1,
box extend=0.7,
lower whisker= 5.774768,
	lower quartile= 5.835998,
	median= 5.875289,
upper quartile= 6.212846,
	upper whisker= 6.364703,
}] 
table[row sep=\\,y index=0] {6.365702\\ 6.367223\\ 6.371201\\ 6.37217\\ 6.376769\\ 6.378068\\ 6.382772\\ 6.385007\\ 6.394262\\ 6.464057\\ 5.773463\\ 5.771618\\ 5.769452\\ 5.76719\\ 5.76353\\ 5.759681\\ 5.754305\\ 5.746838\\ 5.729042\\ 5.572163\\ };
\addlegendentry{\tiny Unif. Mesh}
\addplot+ [mark=o, mark size = 0.5, mark color=grey, 
boxplot prepared={
	draw position=2,
box extend=0.7,
lower whisker= 5.120759,
	lower quartile= 5.174219,
	median= 5.227667,
	upper quartile= 5.277962,
	upper whisker= 5.49095,
}]
table[row sep=\\,y index=0] {5.491946\\ 5.492573\\ 5.495723\\ 5.500742\\ 5.505614\\ 5.506862\\ 5.5145\\ 5.521853\\ 5.531474\\ 5.588471\\ 5.119769\\ 5.118431\\ 5.116937\\ 5.115383\\ 5.113706\\ 5.111909\\ 5.108378\\ 5.104982\\ 5.093906\\ 5.087639\\ };
\addlegendentry{\tiny {\metteor} ($\kappa = 4$)}
\addplot+ [mark=o, mark size = 0.5, mark color=grey, 
boxplot prepared={
	draw position=3,
box extend=0.7,
lower whisker= 5.093027,
	lower quartile= 5.15774,
	median= 5.201504,
	upper quartile= 5.260265,
	upper whisker= 5.450048,
}] 
table[row sep=\\,y index=0] {5.45282\\ 5.456147\\ 5.459297\\ 5.460836\\ 5.465999\\ 5.469923\\ 5.476826\\ 5.487233\\ 5.493791\\ 5.53773941686\\ 5.091155\\ 5.089862\\ 5.088227\\ 5.086292\\ 5.083337\\ 5.079716\\ 5.075777\\ 5.069603\\ 5.054111\\ 5.015501\\ };
\addlegendentry{\tiny Ideal ToE}

\addplot+ [mark=o, mark size = 0.5, mark color=grey, 
boxplot prepared={
	draw position=5,
box extend=0.7,
lower whisker= 5.341391,
	lower quartile= 5.430947,
	median= 5.501081,
	upper quartile= 5.667299,
	upper whisker= 6.077669,
}]
table[row sep=\\,y index=0] {6.079118\\ 6.083495\\ 6.084449\\ 6.085754\\ 6.089381\\ 6.095348\\ 6.098492\\ 6.099599\\ 6.1007\\ 6.138575\\ 5.339009\\ 5.33639\\ 5.334554\\ 5.332454\\ 5.329028\\ 5.325518\\ 5.320442\\ 5.313608\\ 5.295701\\ 5.273651\\ };

\addplot+ [mark=o, mark size = 0.5, mark color=grey, 
boxplot prepared={
	draw position=6,
box extend=0.7,
lower whisker= 5.031491,
	lower quartile= 5.059592,
	median= 5.079512,
	upper quartile= 5.117405,
	upper whisker= 5.318552,
}] 
table[row sep=\\,y index=0] {5.319173\\ 5.321234\\ 5.32238\\ 5.324465\\ 5.327486\\ 5.33024\\ 5.333114\\ 5.335139\\ 5.339696\\ 5.382041\\ 5.031146\\ 5.030567\\ 5.030048\\ 5.029472\\ 5.028755\\ 5.028038\\ 5.027012\\ 5.025656\\ 5.022722\\ 5.019209\\ };

\addplot+ [mark=o, mark size = 0.5, mark color=grey, 
boxplot prepared={
	draw position=7,
box extend=0.7,
lower whisker= 5.018333,
	lower quartile= 5.042135,
	median= 5.058578,
	upper quartile= 5.083427,
	upper whisker= 5.282144,
}] 
table[row sep=\\,y index=0] {5.284172\\ 5.287268\\ 5.28915402217\\ 5.293832\\ 5.28215488708\\ 5.29429347485\\ 5.28544153961\\ 5.30371241204\\ 5.30465063164\\ 5.338418\\ 5.01794\\ 5.017517\\ 5.017016\\ 5.016506\\ 5.015786\\ 5.01503\\ 5.013953\\ 5.012654\\ 5.009576\\ 5.006747\\ };

\addplot+ [mark=o, mark size = 0.5, mark color=grey, 
boxplot prepared={
	draw position=9,
box extend=0.7,
lower whisker= 5.072504,
	lower quartile= 5.110472,
	median= 5.158952,
	upper quartile= 5.253401,
	upper whisker= 6.59375,
}] 
table[row sep=\\,y index=0] {6.599609\\ 7.203125\\ 7.203125\\ 7.203125\\ 7.203125\\ 7.203125\\ 7.203125\\ 7.203125\\ 7.203125\\ 7.867187\\ 5.071904\\ 5.071286\\ 5.070617\\ 5.069813\\ 5.068646\\ 5.067641\\ 5.066441\\ 5.06501\\ 5.059994\\ 5.051828\\ };

\addplot+ [mark=o, mark size = 0.5, mark color=grey, 
boxplot prepared={
	draw position=10,
box extend=0.7,
lower whisker= 5.021375,
	lower quartile= 5.043239,
	median= 5.05865,
	upper quartile= 5.083181,
	upper whisker= 6.029519,
}] 
table[row sep=\\,y index=0] {6.057197\\ 6.704798\\ 6.704798\\ 6.704798\\ 6.704798\\ 6.704798\\ 6.704798\\ 6.704798\\ 6.704798\\ 5.020934\\ 5.02055\\ 5.020052\\ 5.019557\\ 5.01893\\ 5.018264\\ 5.017538\\ 5.016443\\ 5.013719\\ 5.011445\\ };

\addplot+ [mark=o, mark size = 0.5, mark color=grey, 
boxplot prepared={
	draw position=11,
box extend=0.7,
lower whisker= 5.008838,
	lower quartile= 5.025698,
	median= 5.035271,
	upper quartile= 5.048618,
	upper whisker= 5.144084,
}] 
table[row sep=\\,y index=0] {5.144852\\ 5.145743\\ 5.146808\\ 5.148788\\ 5.15075\\ 5.158286\\ 5.164007\\ 5.16797\\ 5.171645\\ 5.304296\\ 5.008433\\ 5.008028\\ 5.007674\\ 5.007302\\ 5.006801\\ 5.006174\\ 5.005388\\ 5.004293\\ 5.0\\ 5.0\\ };

\addplot+ [mark=o, mark size = 0.5, mark color=grey, 
boxplot prepared={
	draw position=13,
box extend=0.7,
lower whisker= 6.005072,
	lower quartile= 6.424025,
	median= 6.624554,
	upper quartile= 6.769289,
	upper whisker= 6.954431,
}] 
table[row sep=\\,y index=0] {6.955781\\ 6.9581\\ 6.95924\\ 6.960512\\ 6.963752\\ 6.967202\\ 6.967694\\ 6.970841\\ 6.975143\\ 7.013831\\ 5.995574\\ 5.98244\\ 5.967776\\ 5.954852\\ 5.936132\\ 5.904827\\ 5.877272\\ 5.844779\\ 5.781134\\ 5.621381\\ };

\addplot+ [mark=o, mark size = 0.5, mark color=grey, 
boxplot prepared={
	draw position=14,
box extend=0.7,
lower whisker= 5.12606,
	lower quartile= 5.239571,
	median= 5.319668,
	upper quartile= 5.428841,
	upper whisker= 5.884442,
}] 
table[row sep=\\,y index=0] {5.888159\\ 5.89016\\ 5.893511\\ 5.89703\\ 5.899448\\ 5.901602\\ 5.909837\\ 5.919083\\ 5.929157\\ 5.973284\\ 5.123546\\ 5.121944\\ 5.119736\\ 5.11709\\ 5.114564\\ 5.111081\\ 5.105315\\ 5.097509\\ 5.076911\\ 5.064023\\ };

\addplot+ [mark=o, mark size = 0.5, mark color=grey, 
boxplot prepared={
	draw position=15,
box extend=0.7,
lower whisker= 5.03912,
	lower quartile= 5.069378,
	median= 5.089022,
	upper quartile= 5.112974,
	upper whisker= 5.248103,
}] 
table[row sep=\\,y index=0] {5.249747\\ 5.250323\\ 5.253239\\ 5.255795\\ 5.258051\\ 5.264012\\ 5.266964\\ 5.278082\\ 5.284385\\ 5.33993\\ 5.038625\\ 5.037731\\ 5.037104\\ 5.036111\\ 5.035136\\ 5.034077\\ 5.032526\\ 5.029448\\ 5.023211\\ 5.012936\\ };

\addplot+ [mark=o, mark size = 0.5, mark color=grey, 
boxplot prepared={
	draw position=17,
box extend=0.7,
lower whisker= 5.110364,
	lower quartile= 5.180738,
	median= 5.24861,
	upper quartile= 5.727593,
	upper whisker= 6.510389,
}] 
table[row sep=\\,y index=0] {6.511955\\ 6.51356\\ 6.516308\\ 6.520895\\ 6.525695\\ 6.529508\\ 6.536078\\ 6.542078\\ 6.5498\\ 6.673817\\ 5.108795\\ 5.107268\\ 5.10563\\ 5.10386\\ 5.102129\\ 5.100194\\ 5.097806\\ 5.094467\\ 5.086115\\ 5.073428\\ };

\addplot+ [mark=o, mark size = 0.5, mark color=grey, 
boxplot prepared={
	draw position=18,
box extend=0.7,
lower whisker= 5.01599,
	lower quartile= 5.034251,
	median= 5.076578,
	upper quartile= 5.188382,
	upper whisker= 5.318996,
}] 
table[row sep=\\,y index=0] {5.319428\\ 5.321423\\ 5.323304\\ 5.323928\\ 5.326361\\ 5.327528\\ 5.329571\\ 5.331185\\ 5.333177\\ 5.381867\\ 5.01575\\ 5.015498\\ 5.015249\\ 5.014895\\ 5.014589\\ 5.0141\\ 5.013512\\ 5.012618\\ 5.010944\\ 5.010221\\ };

\addplot+ [mark=o, mark size = 0.5, mark color=grey, 
boxplot prepared={
	draw position=19,
box extend=0.7,
lower whisker= 5.004854,
	lower quartile= 5.015252,
	median= 5.022131,
	upper quartile= 5.032442,
	upper whisker= 5.185814,
}] 
table[row sep=\\,y index=0] {5.187698\\ 5.189423\\ 5.193263\\ 5.19587\\ 5.197856\\ 5.203931\\ 5.209532\\ 5.217401\\ 5.225942\\ 5.268071\\ 5.004644\\ 5.004452\\ 5.004236\\ 5.00402\\ 5.003729\\ 5.003408\\ 5.003033\\ 5.002517\\ 5.001692\\ 5.000318\\ };

\addplot+ [mark=o, mark size = 0.5, mark color=grey, 
boxplot prepared={
	draw position=21,
box extend=0.7,
lower whisker= 5.363126,
	lower quartile= 5.614868,
	median= 5.753021,
	upper quartile= 5.839583,
	upper whisker= 6.166193,
}] 
table[row sep=\\,y index=0] {6.168455\\ 6.171917\\ 6.177767\\ 6.180326\\ 6.185645\\ 6.1919\\ 6.195902\\ 6.209342\\ 6.211457\\ 6.274328\\ 5.361605\\ 5.359844\\ 5.357507\\ 5.355617\\ 5.352839\\ 5.349506\\ 5.344391\\ 5.335319\\ 5.320319\\ 5.303198\\ };
\addplot+ [mark=o, mark size = 0.5, mark color=grey, 
boxplot prepared={
	draw position=22,
box extend=0.7,
lower whisker= 5.122448,
	lower quartile= 5.246135,
	median= 5.301947,
	upper quartile= 5.381738,
	upper whisker= 5.694371,
}] 
table[row sep=\\,y index=0] {5.696639\\ 5.702294\\ 5.704916\\ 5.707955\\ 5.711111\\ 5.721044\\ 5.723981\\ 5.732561\\ 5.739107\\ 5.783615\\ 5.120798\\ 5.119115\\ 5.117186\\ 5.115059\\ 5.112512\\ 5.110391\\ 5.107892\\ 5.104343\\ 5.094734\\ 5.090777\\ };
\addplot+ [mark=o, mark size = 0.5, mark color=grey, 
boxplot prepared={
	draw position=23,
box extend=0.7,
lower whisker= 5.055161,
	lower quartile= 5.111204,
	median= 5.155706,
	upper quartile= 5.228246,
	upper whisker= 5.518481,
}] 
table[row sep=\\,y index=0] {5.521904\\ 5.526506\\ 5.528783\\ 5.531591\\ 5.535161\\ 5.538254\\ 5.547074\\ 5.549933\\ 5.562809\\ 5.663798\\ 5.053916\\ 5.052686\\ 5.051333\\ 5.049455\\ 5.047199\\ 5.045231\\ 5.042438\\ 5.03681\\ 5.027054\\ 5.020694\\ };

\addplot+ [mark=o, mark size = 0.5, mark color=grey, 
boxplot prepared={
	draw position=25,
box extend=0.7,
lower whisker= 5.731838,
	lower quartile= 5.916068,
	median= 6.027872,
	upper quartile= 6.14171,
	upper whisker= 6.618254,
}] table[row sep=\\,y index=0] {6.620711\\ 6.625262\\ 6.628808\\ 6.630548\\ 6.632489\\ 6.637271\\ 6.639779\\ 6.644225\\ 6.649952\\ 6.692762\\ 5.726918\\ 5.721593\\ 5.715029\\ 5.70812\\ 5.696963\\ 5.683979\\ 5.670308\\ 5.651204\\ 5.617283\\ 5.591222\\ };

\addplot+ [mark=o, mark size = 0.5, mark color=grey, 
boxplot prepared={
	draw position=26,
box extend=0.7,
lower whisker= 5.083532,
	lower quartile= 5.244638,
	median= 5.301302,
	upper quartile= 5.395793,
	upper whisker= 5.76899,
}] table[row sep=\\,y index=0] {5.771564\\ 5.772716\\ 5.775554\\ 5.78006\\ 5.782454\\ 5.786474\\ 5.792321\\ 5.803433\\ 5.821139\\ 5.903963\\ 5.081879\\ 5.07992\\ 5.077769\\ 5.075777\\ 5.073683\\ 5.071094\\ 5.068508\\ 5.06585\\ 5.058383\\ 5.052593\\ };

\addplot+ [mark=o, mark size = 0.5, mark color=grey, 
boxplot prepared={
	draw position=27,
box extend=0.7,
lower whisker= 5.021888,
	lower quartile= 5.056547,
	median= 5.070947,
	upper quartile= 5.089466,
	upper whisker= 5.191394,
}] table[row sep=\\,y index=0] {5.192744\\ 5.193614\\ 5.195\\ 5.196281\\ 5.19794\\ 5.200811\\ 5.201576\\ 5.205134\\ 5.211335\\ 5.241647\\ 5.021237\\ 5.020619\\ 5.019959\\ 5.019056\\ 5.018258\\ 5.017082\\ 5.016059\\ 5.014385\\ 5.010218\\ 5.008007\\ };

\addplot+ [mark=o, mark size = 0.5, mark color=grey, 
boxplot prepared={
	draw position=29,
box extend=0.7,
lower whisker= 5.257913,
	lower quartile= 5.467682,
	median= 5.599412,
	upper quartile= 5.743367,
	upper whisker= 6.23102,
}] table[row sep=\\,y index=0] {6.235232\\ 6.237722\\ 6.23897\\ 6.241457\\ 6.246122\\ 6.249614\\ 6.254816\\ 6.261134\\ 6.269792\\ 6.348263\\ 5.254385\\ 5.251229\\ 5.247329\\ 5.243585\\ 5.239151\\ 5.233109\\ 5.227241\\ 5.21873\\ 5.201456\\ 5.178692\\ };

\addplot+ [mark=o, mark size = 0.5, mark color=grey, 
boxplot prepared={
	draw position=30,
box extend=0.7,
lower whisker= 5.105672,
	lower quartile= 5.285075,
	median= 5.47502,
	upper quartile= 5.651462,
	upper whisker= 6.041858,
}] table[row sep=\\,y index=0] {6.047243\\ 6.050405\\ 6.057224\\ 6.060434\\ 6.062486\\ 6.076766\\ 6.082805\\ 6.098231\\ 6.117842\\ 6.196802\\ 5.103125\\ 5.099654\\ 5.096171\\ 5.093135\\ 5.090402\\ 5.085236\\ 5.081561\\ 5.076932\\ 5.066735\\ 5.057498\\ };

\addplot+ [mark=o, mark size = 0.5, mark color=grey, 
boxplot prepared={
	draw position=31,
box extend=0.7,
lower whisker= 5.023904,
	lower quartile= 5.057651,
	median= 5.095007,
	upper quartile= 5.179625,
	upper whisker= 5.712593,
}] table[row sep=\\,y index=0] {5.716184\\ 5.718905\\ 5.723717\\ 5.725829\\ 5.727665\\ 5.731004\\ 5.739164\\ 5.74562\\ 5.755037\\ 5.866241\\ 5.023382\\ 5.022857\\ 5.02241\\ 5.021678\\ 5.021003\\ 5.020112\\ 5.019218\\ 5.017661\\ 5.014415\\ 5.011826\\ };

\addplot+ [mark=o, mark size = 0.5, mark color=grey, 
boxplot prepared={
	draw position=33,
box extend=0.7,
lower whisker= 5.085803,
	lower quartile= 5.179634,
	median= 5.296457,
	upper quartile= 5.445017,
	upper whisker= 6.374105,
}] table[row sep=\\,y index=0] {6.381284\\ 6.393083\\ 6.398555\\ 6.404261\\ 6.410894\\ 6.424499\\ 6.435227\\ 6.458465\\ 6.473537\\ 6.54074\\ 5.084609\\ 5.083283\\ 5.082269\\ 5.081012\\ 5.079083\\ 5.077754\\ 5.075549\\ 5.072069\\ 5.066501\\ 5.062634\\ };

\addplot+ [mark=o, mark size = 0.5, mark color=grey, 
boxplot prepared={
	draw position=34,
box extend=0.7,
lower whisker= 5.044355,
	lower quartile= 5.200247,
	median= 5.334371,
	upper quartile= 5.503454,
	upper whisker= 5.946992,
}] table[row sep=\\,y index=0] {5.948705\\ 5.953742\\ 5.962691\\ 5.96525\\ 5.975057\\ 5.979629\\ 5.997431\\ 6.015926\\ 6.025262\\ 6.1118\\ 5.043056\\ 5.042045\\ 5.040464\\ 5.039111\\ 5.037767\\ 5.036036\\ 5.033804\\ 5.030483\\ 5.025764\\ 5.02076\\ };

\addplot+ [mark=o, mark size = 0.5, mark color=grey, 
boxplot prepared={
	draw position=35,
box extend=0.7,
lower whisker= 5.019365,
	lower quartile= 5.062658,
	median= 5.10248,
	upper quartile= 5.14736,
	upper whisker= 5.390702,
}] table[row sep=\\,y index=0] {5.392877\\ 5.39612\\ 5.39864\\ 5.40428\\ 5.407316\\ 5.415452\\ 5.419379\\ 5.426744\\ 5.45153\\ 5.511647\\ 5.019017\\ 5.018516\\ 5.017913\\ 5.017205\\ 5.016569\\ 5.015507\\ 5.014874\\ 5.013479\\ 5.010311\\ 5.007773\\ };

\addplot+ [mark=o, mark size = 0.5, mark color=grey, 
boxplot prepared={
	draw position=37,
box extend=0.7,
lower whisker= 5.733221,
	lower quartile= 6.181958,
	median= 6.314276,
	upper quartile= 6.429371,
	upper whisker= 6.795872,
}] table[row sep=\\,y index=0] {6.797198\\ 6.799484\\ 6.807746\\ 6.809804\\ 6.813617\\ 6.815057\\ 6.821108\\ 6.827057\\ 6.843128\\ 6.878081\\ 5.722217\\ 5.71298\\ 5.70179\\ 5.687636\\ 5.676389\\ 5.662013\\ 5.639654\\ 5.610353\\ 5.49449\\ 5.462975\\ };

\addplot+ [mark=o, mark size = 0.5, mark color=grey, 
boxplot prepared={
	draw position=38,
box extend=0.7,
lower whisker= 5.306348,
	lower quartile= 5.652023,
	median= 5.839853,
	upper quartile= 5.982593,
	upper whisker= 6.456089,
}] table[row sep=\\,y index=0] {6.470714\\ 6.486605\\ 6.490832\\ 6.499133\\ 6.5102\\ 6.516194\\ 6.518405\\ 6.543953\\ 6.571292\\ 6.729188\\ 5.301485\\ 5.297093\\ 5.290076\\ 5.278475\\ 5.268926\\ 5.256608\\ 5.245454\\ 5.231336\\ 5.204792\\ 5.182496\\ };

\addplot+ [mark=o, mark size = 0.5, mark color=grey, 
boxplot prepared={
	draw position=39,
box extend=0.7,
lower whisker= 5.062235,
	lower quartile= 5.608631,
	median= 5.82801794855,
	upper quartile= 5.96834912888,
	upper whisker= 6.45386387111,
}] table[row sep=\\,y index=0] {6.45848321815\\ 6.47767436966\\ 6.4864607775\\ 6.48711477852\\ 6.49611189468\\ 6.50786881453\\ 6.50463046959\\ 6.53907852993\\ 6.57123076964\\ 6.72613191559\\ 5.060513\\ 5.058644\\ 5.056679\\ 5.054549\\ 5.051981\\ 5.04896\\ 5.044784\\ 5.038805\\ 5.026376\\ 5.023781\\ };

\addplot+ [mark=o, mark size = 0.5, mark color=grey, 
boxplot prepared={
	draw position=41,
box extend=0.7,
lower whisker= 5.05172,
	lower quartile= 5.080718,
	median= 5.145884,
	upper quartile= 5.248238,
	upper whisker= 5.873519,
}] table[row sep=\\,y index=0] {6.381284\\ 6.393083\\ 6.398555\\ 6.404261\\ 6.410894\\ 6.424499\\ 6.435227\\ 6.458465\\ 6.473537\\ 6.54074\\ 5.084609\\ 5.083283\\ 5.082269\\ 5.081012\\ 5.079083\\ 5.077754\\ 5.075549\\ 5.072069\\ 5.066501\\ 5.062634\\ };

\addplot+ [mark=o, mark size = 0.5, mark color=grey, 
boxplot prepared={
	draw position=42,
box extend=0.7,
lower whisker= 5.007146,
	lower quartile= 5.025416,
	median= 5.071907,
	upper quartile= 5.220884,
	upper whisker= 5.897849,
}] table[row sep=\\,y index=0] {6.17879\\ 6.189821\\ 6.197756\\ 6.203057\\ 6.21596\\ 6.233285\\ 6.253259\\ 6.294986\\ 6.324626\\ 6.431471\\ 5.006963\\ 5.006732\\ 5.006531\\ 5.00633\\ 5.006033\\ 5.005568\\ 5.005025\\ 5.004389\\ 5.002835\\ 5.002313\\ };

\addplot+ [mark=o, mark size = 0.5, mark color=grey, 
boxplot prepared={
	draw position=43,
box extend=0.7,
lower whisker= 5.002931,
	lower quartile= 5.010557,
	median= 5.054393,
	upper quartile= 5.08601,
	upper whisker= 5.523557,
}] table[row sep=\\,y index=0] {5.535383\\ 5.54276\\ 5.575073\\ 5.582642\\ 5.608157\\ 5.619788\\ 5.633207\\ 5.65343\\ 5.693801\\ 5.91299\\ 5.002841\\ 5.002736\\ 5.002637\\ 5.00255\\ 5.00246\\ 5.002328\\ 5.00207\\ 5.001797\\ 5.001152\\ 5.001041\\ };

\addplot+ [mark=o, mark size = 0.5, mark color=grey, 
boxplot prepared={
	draw position=45,
box extend=0.7,
lower whisker= 5.298668,
	lower quartile= 5.461076,
	median= 5.5901,
	upper quartile= 6.011153,
	upper whisker= 6.468503,
}] table[row sep=\\,y index=0] {6.469649\\ 6.47324\\ 6.475982\\ 6.479108\\ 6.490034\\ 6.505025\\ 6.523178\\ 6.532661\\ 6.54824\\ 6.879443\\ 5.295806\\ 5.292743\\ 5.290448\\ 5.286248\\ 5.281325\\ 5.276636\\ 5.270132\\ 5.259938\\ 5.234642\\ 5.226644\\ };

\addplot+ [mark=o, mark size = 0.5, mark color=grey, 
boxplot prepared={
	draw position=46,
box extend=0.7,
lower whisker= 5.136825,
	lower quartile= 5.180761,
	median= 5.464279,
	upper quartile= 5.649067,
	upper whisker= 6.459156,
}] table[row sep=\\,y index=0] {6.553755\\ 6.555492\\ 6.558858\\ 6.561939\\ 6.667942\\ 6.79963\\ 5.039108\\ 5.036789\\ 5.034323\\ 5.031695\\ 5.029205\\ 5.025482\\ 5.022128\\ 5.018048\\ 5.008736\\ 5.006825\\ };

\addplot+ [mark=o, mark size = 0.5, mark color=grey, 
boxplot prepared={
	draw position=47,
box extend=0.7,
lower whisker= 5.041343,
	lower quartile= 5.169902,
	median= 5.328569,
	upper quartile= 5.596502,
	upper whisker= 6.440687,
}] table[row sep=\\,y index=0] {6.448583\\ 6.450002\\ 6.461957\\ 6.472112\\ 6.477188\\ 6.483056\\ 6.48809\\ 6.498644\\ 6.51407\\ 6.628745\\ 5.039108\\ 5.036789\\ 5.034323\\ 5.031695\\ 5.029205\\ 5.025482\\ 5.022128\\ 5.018048\\ 5.008736\\ 5.006825\\ };
\end{axis}
\end{tikzpicture}
\vspace{-5pt}
\caption{\small Ave. hop count}
\end{subfigure}
\vspace{-14pt}
\caption{\small Additional simulations results, using 6 months’ worth of 5-minute-averaged traffic snapshots from 12 different DCN fabrics. } 
\vspace{-7pt}
\label{fig:all_results_boxplot}
\end{figure*}

\section{Algorithms for BPM and LDM}\label{appendix:barrier_penalty_detailed_walkthrough}
Here, we fully flesh out the Barrier Penalty Method (BPM) and Lagrangian Dual Method (LDM), and provide the numerical algorithms needed for each method to work. Both methods are iterative, and will save the solutions that yields the lowest ratio of soft constraint violations encountered up till the current iteration. First, we introduce a goodness function for a feasible OCS switch configuration state, $\mathbf{x}$ used to keep track of the best solution thus far:
\begin{equation}\label{eqn:goodness_function}
\begin{aligned}
\Psi(\mathbf{x}) = \sum\limits_{i, j \in \mathcal{S}}\psi_{ij}
\end{aligned}
\end{equation}
Where $\psi_{ij}$ is an indicator variable that equals 1 when the $(i, j)$ pod pair’s soft constraints is satisfied, and 0 otherwise.

\subsection{Detailed Walkthrough for BPM}\label{appendix:barrier_penalty_detailed_walkthrough}
Even though (\ref{eqn:barrier_objective}) has relaxed the soft constraints, solving it to optimality directly is still challenging due to its quadratic objective function. We want a low-complexity algorithm with good objective value, rather than the optimal solution. To achieve this, we use first-order approximation on the objective function $U(\mathbf{x})$:
\begin{eqnarray}\label{eqn:barrier_linear_approx}
&&\min\limits_{\mathbf{x}} U(\mathbf{x}) \\
&&\approx \min\limits_{\mathbf{x}}\bigg\{ U(\mathbf{\hat{x}}) + \sum_{k=1}^K\sum_{i=1}^n\sum_{j=1}^n\Big(\frac{\partial U}{\partial {x_{ij}^k}}\bigg\rvert_{\mathbf{x}=\mathbf{\hat{x}}}\Big)\times \Big(x_{ij}^k-\hat{x}_{ij}^k\Big)\bigg\}\nonumber\\
&&= C+\sum_{k=1}^K \bigg\{\min\limits_{\mathbf{x}^k} \sum_{i=1}^n\sum_{j=1}^n\bigg[\sum_{k^{\prime}=1}^K 2\hat{x}_{ij}^{k^{\prime}} - (\ceil{d_{ij}^*} + \floor{d_{ij}^*})\bigg]x_{ij}^k\bigg\}\nonumber
\end{eqnarray}
where $\mathbf{\hat{x}}$ is an initial value of $\mathbf{x}$, and $C$ is a constant.

As the constraints in (\ref{constraint:ocslevel}), and the approximation form of $U(\mathbf{x})$ in (\ref{eqn:barrier_linear_approx}) are separable in $k$, we can solve for $\mathbf{x}$ iteratively, one OCS at a time, as follows:
\begin{eqnarray}\label{eqn:max_weighted_matching}
&\min\limits_{\mathbf{x}^k}& \sum_{i=1}^n\sum_{j=1}^n\bigg[\sum_{k^{\prime}=1}^K 2\hat{x}_{ij}^{k^{\prime}} - (\ceil{d_{ij}^*} + \floor{d_{ij}^*})\bigg]x_{ij}^k\\
& \text{s.t: }& \sum_{i=1}^n x_{ij}^k \leq h_{\text{ig}}^k(j), \sum_{j=1}^n x_{ij}^k \leq h_{\text{eg}}^k(i),\nonumber\\
&&\max\{\hat{x}_{ij}^k - 1, 0\}\leq x_{ij}^k\leq \hat{x}_{ij}^k + 1 \; \forall \; i,j=1,...,n\nonumber
\end{eqnarray}

We add a range for every $x_{ij}^k$ because the approximation in (\ref{eqn:barrier_linear_approx}) only works in the neighborhood of $\mathbf{\hat{x}}$.
(\ref{eqn:max_weighted_matching}) is easily solvable using min-cost circulation algorithms (see Appendix \ref{appendix:mincost_flow}). Further, since all the bounds (i.e., $h_{\text{ig}}^k(i)$, $h_{\text{eg}}^k(i)$) are integers, an integer solution of $x_{ij}^k$ is guaranteed. Due to space limits, Appendix \ref{appendix:barrier_penalty_detailed_walkthrough} provides the BPM pseudocode.

We have gone through the intuition of BPM in \S \ref{section_barrier_penalty}. Here, we provide the detailed pseudocode in Algorithm \ref{overall_barrier}.

\begin{algorithm}
 \KwData{
    \begin{itemize}
        \item $D^* = [d_{ij}^*] \in \mathbb{R}^{n \times n}$ - fractional topology
        \item $\tau_{max}$ - number of iterations
    \end{itemize}
 }
 \KwResult{
    $\mathbf{x}^{*} = [x_{ij}^k{}^*] \in \mathbb{Z}^{n^2K}$ - OCS switch states.
 }
 Initialize: $\hat{\mathbf{x}} := 0, \mathbf{x}^{*} := 0$\;
 \For{$\tau \in \{1, 2, ..., \tau_{max}\}$}{
  \For{$k \in \{1,2,...,K\}$}{
      Solve (\ref{eqn:max_weighted_matching}) based on Appendix \ref{appendix:mincost_flow} and let $x^k$ be the integer solution\;
      Update $\hat{\mathbf{x}}$ in the $k$-th OCS by setting $\hat{x}^k=x^k$\;
      \If{$\Psi(\mathbf{x}^*) < \Psi(\hat{\mathbf{x}})$} {
        $\mathbf{x}^{*} := \hat{\mathbf{x}}$\;
      }
   }
 }
 \caption{Barrier penalty method}\label{overall_barrier}
\end{algorithm}

Algorithm \ref{overall_barrier} is an iterative algorithm. Although only one OCS gets updated in each step, we obtain a new solution after combining other OCSs’ old states. The goodness function $\Psi(\mathbf{x})$ is used to track the best solution obtained so far. In our implementation, we use use (\ref{eqn:goodness_function}) as our goodness function. Many alternative goodness functions exists, though their relative merits are subject for future work.

\subsection{Detailed Walkthrough for LDM}\label{appendix:lagrangian_dual_detailed_walkthrough}
Lagrangian Dual method was motivated by the dual ascent method in~\cite{Boyd2011ADMM}. By introducing dual variables for soft constraints, LDM not only achieves graceful relaxation of soft constraints, but also relaxes the original NP-hard problem to a polynomial-time solvable problem. Nevertheless, LDM differs from the dual ascent method due to integer requirement. In this section, we detail the steps required for LDM to work.

\subsubsection{Primal Problem}\hfill\\
Our goal is to find an integer solution of $\mathbf{x} = [x_{ij}^k]$ satisfying the soft constraint (\ref{con_constraint}) and the hard constraints in (\ref{constraint:ocslevel}). In theory, there is no need for an objective function of $\mathbf{x}$ in our problem, since the problem itself is more concerned with satisfiability of the soft-constraints. However, this will lead to an algorithm with extremely poor convergence property. To speed up convergence, we introduce a strictly convex objective function for our primal problem, which is written as:
\begin{equation}
\begin{aligned}
    \mathbf{P: } \max_{\mathbf{x}} U(\mathbf{x}) &= \sum_{k=1}^K\sum_{i=1}^n\sum_{j=1}^n U_{ij}^k(x^k_{ij})\\
     \text{s.t : } & \quad (\ref{constraint:ocslevel}), \; (\ref{con_constraint}), \; \text{ are satisfied }\\
\end{aligned}
\label{primal_objective}
\end{equation}

At first, we chose $U_{ij}^k(x^k_{ij})=0$, which is not strictly convex. As expected, the solution does not converge even after running large number of iterations. We then chose $U_{ij}^k(x^k_{ij})=-(x^k_{ij})^2$, which introduces a sharper objective function landscape that facilitated superior convergence. However, this objective function will result in a solution of $\mathbf{x}$ that connects as fewer links as possible in each OCS, which not only wastes physical resources but also resulted in an overall decrease in network capacity. Finally, we went with:
\begin{equation}\label{primalU}
U_{ij}^k(x_{ij}^k) = - \Big(x_{ij}^k \; - \; h_{ij}^k \Big)^2
\end{equation}
Where $h_{ij}^k = \min\big(h_{eg}^k(s_i), h_{in}^k(s_j)\big)$, taking advantage of the fact that $h_{ij}^k \geq x_{ij}^k$ to ensure that the optimal solution maximizes the formation of logical links.

\subsubsection{Dual Problem}\hfill\\
To relax the soft constraint (\ref{con_constraint}), we introduce dual variables $\mathbf{p}^+=[p_{ij}^+] \geq 0, \mathbf{p}^-=[p_{ij}^-] \geq 0$, and the following Lagrangian of the primal problem (\ref{primal_objective}):
\begin{eqnarray}
&&L(\mathbf{x}, \mathbf{p}^+, \mathbf{p}^-)\nonumber\\
&=&\sum_{i=1}^n\sum_{j=1}^n \left[ \sum_{k=1}^KU_{ij}^k(x_{ij}^k) - p_{ij}^+\Bigg( \sum_{k=1}^Kx_{ij}^k - \ceil{d_{ij}^*}\Bigg)\right.\nonumber\\
&&\left.\hspace{28mm} + p_{ij}^-\Bigg(\sum_{k=1}^K x_{ij}^k - \floor{d_{ij}^*}\Bigg) \right].\nonumber
\end{eqnarray}

Note that for every $\mathbf{x}$ satisfying constraints (\ref{constraint:ocslevel}),  (\ref{con_constraint}), and every
$\mathbf{p}^+ \geq 0$ and $\mathbf{p}^-\geq 0$, the following inequality holds: $$L(\mathbf{x}, \mathbf{p}^+, \mathbf{p}^-)\geq \sum_{i=1}^n\sum_{j=1}^n \sum_{k=1}^KU_{ij}^k(x_{ij}^k).$$ Let
$$g(\mathbf{p}^+, \mathbf{p}^-):=\max_{\mathbf{x}} L(\mathbf{x}, \mathbf{p}^+, \mathbf{p}^-)\quad\text{s.t. }(\ref{constraint:ocslevel}) \text{ is satisfied}$$
We then have
\begin{eqnarray}\label{eqn:duality_gap}
g(\mathbf{p}^+, \mathbf{p}^-)&\geq&\max_{\mathbf{x}} L(\mathbf{x}, \mathbf{p}^+, \mathbf{p}^-)\;\text{satisfying }(\ref{con_constraint}), (\ref{constraint:ocslevel})\\
&\geq&\text{Optimal value of the primal problem (\ref{primal_objective})}\nonumber
\end{eqnarray}

Next, we introduce the dual problem:
\begin{equation}\label{dual_problem}
\mathbf{D: } \min_{\mathbf{p}^+, \mathbf{p}^-} g(\mathbf{p}^+, \mathbf{p}^-)\quad\text{s.t  } \mathbf{p}^+ \geq 0, \; \mathbf{p}^- \geq 0.
\end{equation}
Since the inequality (\ref{eqn:duality_gap}) holds for all $\mathbf{p}^+ \geq 0$ and $\mathbf{p}^-\geq 0$, we must have
\begin{eqnarray}
&&\text{The minimum value of the dual problem }(\ref{dual_problem})\nonumber\\
&\geq&\text{The maximum value of the primal problem }(\ref{primal_objective}).\nonumber
\end{eqnarray}
\emph{Duality gap} is then defined as the difference between the minimum value of the dual problem (\ref{dual_problem}) and the maximum value of the primal problem (\ref{primal_objective}).

If the primal decision variable $\mathbf{x}$ were fractional numbers instead of integers, under mild constraints\footnote{For Slater's Condition: see \S 5.2.3 in~\cite{ConvexOptimization}.}, the duality gap would be $0$. In that case, the optimal primal solution can be obtained by solving the dual problem instead. As we will see shortly, the dual problem (\ref{dual_problem}) is much easier to solve. However, (\ref{primal_objective}) is an integer problem with non-zero duality gap, hence solving the dual problem (\ref{dual_problem}) cannot give us the optimal solution of the primal problem (\ref{primal_objective}). Nevertheless, by optimizing the dual problem, we can still obtain a good sub-optimal solution to (\ref{primal_objective}) that satisfies all the hard constraints and a vast majority of the soft constraints.

\subsubsection{Subgradient Method}\hfill\\
The key aspect of LDM the optimization of the dual problem (\ref{dual_problem}). Since the dual objective function is not differentiable, the typical gradient descent algorithm cannot be applied here. Hence, we use the subgradient method~\cite{SubgradientMethod} instead, whose general form is given as follows:

\begin{definition}\label{def:subgradient}
(Subgradient method~\cite{SubgradientForm}): Let $f:\mathbb{R}^n\rightarrow \mathbb{R}$ be a convex function with domain $\mathbb{R}^n$, a classical subgradient method iterates
$$y^{(m+1)}=y^{(m)} - \alpha_m \gamma^{(m)},$$
where $\gamma^{(m)}$ denotes a subgradient of $f$ at $y^{(m)}$, where $y^{(m)}$ is the $m$-th iterate of $y$. If $f$ is differentiable, then the only subgradient is the gradient vector of $f$. It may happen that $\gamma^{(m)}$ is not a descent direction for $f$ at $y^{(m)}$. We therefore keep a list of $f_{\text{best}}$ to keep track of the lowest objective function value found so far, i.e.,
$$f_{\text{best}}=\min\{f_{\text{best}}, f(y^{(m)})\}.$$
\end{definition}

Computing subgradient is the key step of the above subgradient method. The following lemma tells us how to compute a subgradient for the dual objective function $g(\mathbf{p}^+, \mathbf{p}^-)$.

\begin{lemma}\label{lem:subgradient}
For a given $(\hat{\mathbf{p}}^+, \hat{\mathbf{p}}^-)$, let $\hat{\mathbf{x}}$ be an integer solution that maximizes the lagrangian $L(\mathbf{x}, \hat{\mathbf{p}}^+, \hat{\mathbf{p}}^-)$, i.e., $$g(\hat{\mathbf{p}}^+, \hat{\mathbf{p}}^-) = \max_{\mathbf{x}}L(\mathbf{x}, \hat{\mathbf{p}}^+, \hat{\mathbf{p}}^-) = L(\hat{\mathbf{x}}, \hat{\mathbf{p}}^+, \hat{\mathbf{p}}^-).$$
Then, $[\ceil{d_{ij}^*} - \sum_{k=1}^K \hat{x}_{ij}^k, \sum_{k=1}^K \hat{x}_{ij}^k - \floor{d_{ij}^*}, i,j = 1,...,n]$ is a subgradient of $g(\mathbf{p}^+, \mathbf{p}^-)$ at $(\hat{\mathbf{p}}^+, \hat{\mathbf{p}}^-)$, i.e.,
\begin{eqnarray}
&&g(\mathbf{p}^+, \mathbf{p}^-)-g(\hat{\mathbf{p}}^+, \hat{\mathbf{p}}^-)\nonumber\\
&\geq& \sum_{i=1}^n\sum_{j=1}^n \bigg(\ceil{d_{ij}^*} - \sum_{k=1}^K \hat{x}_{ij}^k\bigg)(p_{ij}^+ - \hat{p}_{ij}^+)\nonumber\\
&&\hspace{-3mm}+ \sum_{i=1}^n\sum_{j=1}^n \bigg(\sum_{k=1}^K \hat{x}_{ij}^k - \floor{d_{ij}^*}\bigg)(p_{ij}^- - \hat{p}_{ij}^-)\nonumber
\end{eqnarray}
for any $(\mathbf{p}^+, \mathbf{p}^-)$ in a neighbourhood of $(\hat{\mathbf{p}}^+, \hat{\mathbf{p}}^-)$.
\end{lemma}

\begin{proof}
Consider an arbitrary $(\mathbf{p}^+, \mathbf{p}^-)$. According to the definition of $g(\mathbf{p}^+, \mathbf{p}^-)$, we must have
$$g(\mathbf{p}^+, \mathbf{p}^-) = \max_{\mathbf{x}} L(\mathbf{x}, \mathbf{p}^+, \mathbf{p}^-) \geq L(\hat{\mathbf{x}}, \mathbf{p}^+, \mathbf{p}^-).$$
Then,
\begin{eqnarray}
&&g(\mathbf{p}^+, \mathbf{p}^-)-g(\hat{\mathbf{p}}^+, \hat{\mathbf{p}}^-)\nonumber\\
&\geq& L(\hat{\mathbf{x}}, \mathbf{p}^+, \mathbf{p}^-) - L(\hat{\mathbf{x}}, \hat{\mathbf{p}}^+, \hat{\mathbf{p}}^-)\nonumber\\
&=&\sum_{i=1}^n\sum_{j=1}^n \bigg(\ceil{d_{ij}^*} - \sum_{k=1}^K \hat{x}_{ij}^k\bigg)(p_{ij}^+ - \hat{p}_{ij}^+)\nonumber\\
&&\hspace{-3mm}+ \sum_{i=1}^n\sum_{j=1}^n \bigg(\sum_{k=1}^K \hat{x}_{ij}^k - \floor{d_{ij}^*}\bigg)(p_{ij}^- - \hat{p}_{ij}^-),\nonumber
\end{eqnarray}
which completes the proof.
\end{proof}

\begin{remark}
Note that for each $(\hat{\mathbf{p}}^+, \hat{\mathbf{p}}^-)$, $\hat{\mathbf{x}}$ may not be the only solution that maximizes the Lagrangian $L(\mathbf{x}, \hat{\mathbf{p}}^+, \hat{\mathbf{p}}^-)$, because $L(\mathbf{x}, \hat{\mathbf{p}}^+, \hat{\mathbf{p}}^-)$ has integer variables $\mathbf{x}$. It is thus possible to have multiple subgradients for $g(\mathbf{p}^+, \mathbf{p}^-)$ at $(\hat{\mathbf{p}}^+, \hat{\mathbf{p}}^-)$, in which case $g(\mathbf{p}^+, \mathbf{p}^-)$ is not differentiable at $(\hat{\mathbf{p}}^+, \hat{\mathbf{p}}^-)$. If $g(\mathbf{p}^+, \mathbf{p}^-)$ were differentiable at $(\hat{\mathbf{p}}^+, \hat{\mathbf{p}}^-)$, there would be only one subgradient, which is the gradient of $g(\mathbf{p}^+, \mathbf{p}^-)$.
\end{remark}

According to Lemma \ref{lem:subgradient}, the most critical part of calculating subgradient is to find a maximizer for a given Lagrangian. By rearranging the dual objective function $g(\mathbf{p}^+, \mathbf{p}^-)$, we obtain the following:
\begin{eqnarray}
&&g(\mathbf{p}^+, \mathbf{p}^-)\nonumber\\
&=& \max_{\mathbf{x}} L(\mathbf{x}, \mathbf{p}^+, \mathbf{p}^-)\quad\text{s.t. }(\ref{constraint:ocslevel}) \text{ is satisfied }\nonumber\\
&=&\sum_{k=1}^K\max_{\mathbf{x}^k}\left[\sum_{i=1}^n\sum_{j=1}^n \bigg(U_{ij}^k(x_{ij}^k)+(p_{ij}^- - p_{ij}^+)x_{ij}^k\bigg)\right]\nonumber\\
&&+\sum_{i=1}^n\sum_{j=1}^n (p_{ij}^+\ceil{d_{ij}^*} - p_{ij}^- \floor{d_{ij}^*})\quad\text{s.t. }(\ref{constraint:ocslevel}) \text{ is satisfied} \nonumber
\end{eqnarray}
From the above equation, we can see that optimizing the Lagrangian can be decomposed into $K$ subproblems:
\begin{eqnarray}\label{eqn:DecomposedProblem}
&\max_{\mathbf{x}^k}&\sum_{i=1}^n\sum_{j=1}^n \bigg(U_{ij}^k(x_{ij}^k)+(p_{ij}^- - p_{ij}^+)x_{ij}^k\bigg)\\
& \text{s.t: }& \sum_{i=1}^n x_{ij}^k \leq h_{\text{ig}}^k(j), \sum_{j=1}^n x_{ij}^k \leq h_{\text{eg}}^k(i)\quad \forall \; i,j\nonumber
\end{eqnarray}
Although these subproblems have significantly fewer decision variables, they are still integer programming problem with quadratic objective function, which can be hard to solve. To further reduce complexity, we apply the same first-order approximation (see Eqn. (\ref{eqn:max_weighted_matching})) again to the nonlinear terms in (\ref{eqn:DecomposedProblem}), and obtain
\begin{eqnarray}\label{eqn:LinearDecomposedProblem}
&\max_{x^k}&\sum_{i=1}^n\sum_{j=1}^n \bigg(\frac{d U_{ij}^k}{d x_{ij}^k}(\hat{x}_{ij}^k)+(p_{ij}^- - p_{ij}^+)x_{ij}^k\bigg)\\
& \text{s.t: }& \sum_{i=1}^n x_{ij}^k \leq h_{\text{ig}}^k(j), \sum_{j=1}^n x_{ij}^k \leq h_{\text{eg}}^k(i),\nonumber\\
&&\max\{\hat{x}_{ij}^k - 1, 0\}\leq x_{ij}^k\leq \hat{x}_{ij}^k + 1 \quad \forall \; i,j\nonumber
\end{eqnarray}
where $\mathbf{\hat{x}}$ is the previous estimate of $\mathbf{x}$. The approximated problem (\ref{eqn:LinearDecomposedProblem}) can be solved in polynomial time using the method in Appendix \ref{appendix:mincost_flow}.

\subsubsection{Detailed Algorithm}\hfill\\
The detailed algorithm is shown in Algorithm \ref{algorithm:overall_primaldual}. Note that we update dual variables right after computing a configuration for each OCS to hasten solution convergence. Another option is to update dual variables after iterating through all the OCSs for one round. The problem with this option is that OCSs with the same physical striping will be configured exactly the same way in the same iteration, causing the solution to oscillate and slows down convergence.

Notice that the harmonic step size function $\delta(\tau)$ is chosen because its sum approaches infinity as we take infinitely many step sizes. This way, we ensure that $\mathbf{p}^+, \mathbf{p}^-$'s growth is not handicapped by the step size if their optimal values are large.

\begin{algorithm}
 \KwData{
    \begin{itemize}
        \item $D^* = [d_{ij}^*] \in \mathbb{R}^{n \times n}$ - fractional topology
        \item $\tau_{max}$ - number of iterations
    \end{itemize}
 }
 \KwResult{
    $\mathbf{x}^{*}=[x_{ij}^k{}^*] \in \mathbb{Z}^{n^2K}$ - OCS switch states
 }
Initialize: $\hat{\mathbf{x}} := 0, \mathbf{x}^{*} := 0, \mathbf{p}^+ :=0, \mathbf{p}^- :=0$ \;
Build network flow graphs $G_1, ..., G_k$ based on each OCS in $\mathcal{O} = \{o_1, ..., o_k\}$\;
 \For{$\tau \in \{1, 2, ..., \tau_{max}\}$}{
  Set step size $\delta := \frac{1}{\tau}$  \;
  \For{$k \in \{1,2,...,K\}$}{
      Solve (\ref{eqn:LinearDecomposedProblem}) based on Appendix \ref{appendix:mincost_flow}, and let $x^k$ be the integer solution\;
      Update $\hat{\mathbf{x}}$ in the $k$-th OCS by setting $\hat{x}^k=x^k$\;
      \If{$\Psi(\mathbf{x}^*) < \Psi(\hat{\mathbf{x}})$} {
        $\mathbf{x}^{*} := \hat{\mathbf{x}}$\;
      }
      Update dual variables using $p_{ij}^+:=\max\{p_{ij}^+-\delta(\ceil{d_{ij}^*}-\sum_{k^{\prime}=1}^K\hat{x}_{ij}^{k^{\prime}}), 0\}$ and $p_{ij}^-:=\max\{p_{ij}^--\delta(\sum_{k^{\prime}=1}^K\hat{x}_{ij}^{k^{\prime}}-\floor{d_{ij}^*}), 0\}$.
   }
 }
 \caption{Lagrangian duality method}\label{algorithm:overall_primaldual}
\end{algorithm}
\section{Mapping (\ref{eqn:max_weighted_matching}) to a Min-Cost Circulation Problem}\label{appendix:mincost_flow}
In this section, we study a general form of (\ref{eqn:max_weighted_matching}) as follows:
\begin{eqnarray}\label{eqn:general_mincostflow}
&\min\limits_{\mathbf{a}=[a_{ij}]}& \sum_{i=1}^I\sum_{j=1}^J C_{ij}a_{ij}\\
& \text{s.t: }& \sum_{i=1}^I a_{ij} \leq P_j, \sum_{j=1}^J a_{ij} \leq Q_i,\nonumber\\
&& L_{ij}\leq a_{ij} \leq U_{ij} \quad \forall \; i=1,...,I,j=1,...,J\nonumber
\end{eqnarray}
where $\mathbf{a}=[a_{ij}]$ is an $I\times J$ integer matrix to be solved, and $\mathbf{C}=[C_{ij}], \mathbf{P}=[P_j], \mathbf{Q}=[Q_i], \mathbf{L}=[L_{ij}], \mathbf{U}=[U_{ij}]$ are predefined constants. We would like to show that (\ref{eqn:general_mincostflow}) can be easily mapped to a min-cost circulation problem, which is polynomial time-solvable with  integer solution guarantees as long as $\mathbf{P}=[P_j], \mathbf{Q}=[Q_i], \mathbf{L}=[L_{ij}], \mathbf{U}=[U_{ij}]$ are all integers.

\subsection{Min-Cost Circulation Problem}
\begin{definition}\label{mincostcirculation}
(Min-Cost Circulation Problem) Given a flow network with
\begin{itemize}
  \item $l(v,w)$, lower bound on flow from node $v$ to node $w$;
  \item $u(v,w)$, upper bound on flow from node $v$ to node $w$;
  \item $c(v,w)$, cost of a unit of flow on $(v,w)$,
\end{itemize}
the goal of the min-cost circulation problem is to find a flow assignment $f(v,w)$ that minimizes
$$\sum_{(v,w)}c(v,w)\cdot f(v,w),$$
while satisfying the following two constraints:
\begin{enumerate}
  \item Throughput constraints: $l(v,w)\leq f(v,w)\leq u(v,w)$;
  \item Flow conservation constraints: $\sum_u f(u,v) = \sum_w f(v,w)$ for any node $v$.
\end{enumerate}
\end{definition}

Note that all the constant parameters $l(v,w),u(v,w)$ are all positive and $c(v,w)$ can be either positive or negative. In addition, min-cost circulation problem has a very nice property that guarantees integer solutions:
\begin{lemma}\label{integralflowtheorem}
(Integral Flow Theorem) Given a feasible circulation problem, if $l(v,w)$'s and $u(v,w)$'s are all integers, then there exists a feasible flow assignment such that all flows are integers.
\end{lemma}

In fact, for feasible circulation problems with integer bounds, most max-flow algorithms, e.g., Edmonds-Karp algorithm \cite{Edmonds1972Theoretical} and Goldberg-Tarjan algorithm \cite{Goldberg1988A}, are guaranteed to generate integer solutions.

\subsection{Detailed Transformation Steps}
\begin{figure}[ht!]
    \centering
    \includegraphics[scale=0.44]{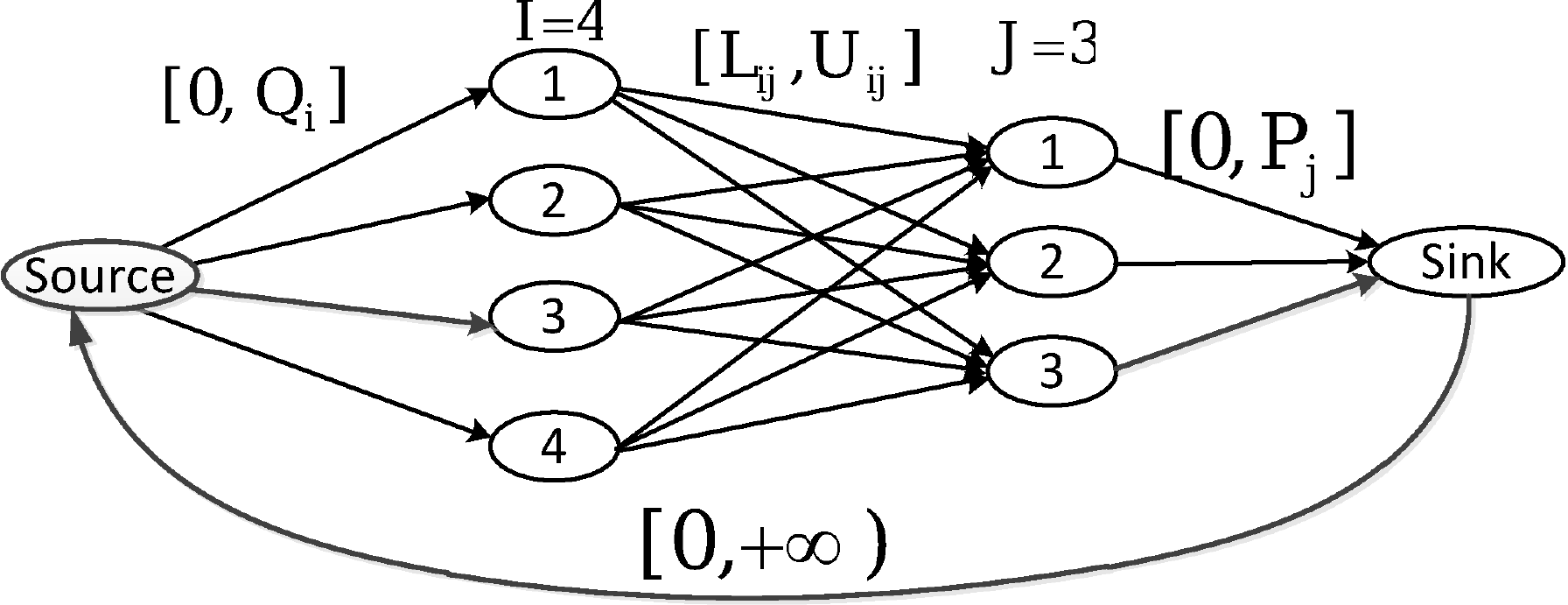}
    \caption{A flow graph example corresponding to equation (\ref{eqn:general_mincostflow}).}
    \label{fig:min_cost_flow}
\end{figure}

We first construct a flow network based on equation (\ref{eqn:general_mincostflow}) as follows
(see Fig. \ref{fig:min_cost_flow} for graphical illustration):
\begin{enumerate}
  \item Create a directed bipartite graph. Note that ${\mathbf{a}}$ is an $I\times J$ matrix. We create $I$ nodes on the left hand side of the bipartite graph, and create $J$ nodes on the right hand side of the bipartite graph. We add a directed link from $i$ to $j$, and set the bounds of this link as $[L_{ij}, U_{ij}]$ and the cost of this link as $C_{ij}$.
  \item Add a source node, and for each of the $I$ left nodes, add a link that connects to this source node. The bounds of the $i$-th link is set as $[0, Q_i]$, and the cost is set to $0$.
  \item Add a sink node and $J$ links from the $J$ right nodes to this sink node. The bounds of the $j$-th link is set as $[0, P_i]$, and the cost is set to $0$.
  \item Add a feedback link from the sink node to the source node. The bounds of this feedback link is set as $[0,\infty)$, and the cost is set as a very small negative value $-\epsilon$, e.g., $-10^{-6}$.
\end{enumerate}

We then assign flows to this flow network.
\begin{enumerate}
  \item For the link from the $i$-th left node to the $j$-th right node, assign $a_{ij}$ amount of flow.
  \item For the link from the source node to the $i$-th left node, assign $\sum_{j=1}^J a_{ij}$ amount of flow.
  \item For the link from the $j$-th right node to the sink node, assign $\sum_{i=1}^I a_{ij}$ amount of flow.
  \item For the feedback link from the sink node to the source node, assign $\sum_{i=1}^I\sum_{j=1}^J a_{ij}$ amount of flow.
\end{enumerate}

It is easy to verify that the above flow assignment satisfies the flow conservation constraints in Definition \ref{mincostcirculation}. Further, by enforcing the throughput constraints in Definition \ref{mincostcirculation}, all the constraints in (\ref{eqn:general_mincostflow}) are also satisfied. Further, the objective function of this min-cost flow problem is
\begin{equation}\label{eqn:mcf_objective}
\sum_{i=1}^I\sum_{j=1}^J C_{ij}a_{ij}+\epsilon \bigg(\sum_{i=1}^I\sum_{j=1}^J a_{ij}\bigg).
\end{equation}
Since $a_{ij}$’s are all integers, $\sum_{i=1}^I\sum_{j=1}^J C_{ij}a_{ij}$ cannot be take on a continuum of values. Then, as long as $\epsilon$ is small enough, minimizing (\ref{eqn:mcf_objective}) will also minimizes the objective function in (\ref{eqn:general_mincostflow}). The benefit of having a small negative cost $\epsilon$ is that more flows can be assigned if possible.

\begin{figure}[t!]
\pgfplotstableread{plot_data/ocs_reconf_benchmark.txt}
	\datatable
\hspace{-0.5cm}
\begin{subfigure}[c]{0.47\columnwidth} 
\centering
\begin{tikzpicture}
\begin{axis}[xlabel = \footnotesize frac. constraints violated, xlabel near ticks, cycle list={{color=black, densely dotted, semithick}, {color=red, mark=square}, {color = blue, mark=x}}, width=1.3\linewidth, height=3.3cm, outer sep=-3pt, ymin=0,ymax=1,xmin=0, ylabel= \footnotesize CDF, ylabel near ticks, ylabel shift=-2.pt, title={\footnotesize Nearest neighbor}, xlabel shift=-1.pt, legend style={at={(0.50, 0.65)}, anchor= north west}, xtick pos=left, ytick pos=left, grid, thick, xtick={0.02, 0.04, 0.06, 0.08, 0.1}, xticklabels={0.02, 0.04, 0.06, 0.08, 0.1}]
\addplot+[mark repeat={5}] table[x=violation__knn_greedy_x,y=violation__knn_greedy_y] from \datatable ;
\addplot+[mark repeat={5}] table[x=violation__knn_barrier_x, y=violation__knn_barrier_y] from \datatable ;
\addplot+[mark repeat={5}] table[x=violation__knn_primaldual_x, y=violation__knn_primaldual_y] from \datatable ;
\legend{\scriptsize greedy, \scriptsize BPM, \scriptsize LDM}
\end{axis}
\end{tikzpicture}
\end{subfigure} 
\hfill
\begin{subfigure}[c]{0.47\columnwidth}
\centering
\begin{tikzpicture}
\begin{axis}[xlabel = \footnotesize frac. constraints violated, xlabel near ticks, cycle list={{color=black, densely dotted, semithick}, {color=red, mark=square}, {color = blue, mark=x}}, width=1.3\linewidth, height=3.3cm, outer sep=-3pt, ymin=0,ymax=1,xmin=0, ylabel= \footnotesize CDF, ylabel near ticks, ylabel shift=-2.pt, xlabel shift=-1.pt, title={\footnotesize Random}, legend style={at={(0.50, 0.65)}, anchor= north west}, xtick pos=left, ytick pos=left, grid, thick]
\addplot+[mark repeat={5}] table[x=violation__random_greedy_x,y=violation__random_greedy_y] from \datatable ;
\addplot+[mark repeat={5}] table[x=violation__random_barrier_x, y=violation__random_barrier_y] from \datatable ;
\addplot+[mark repeat={5}] table[x=violation__random_primaldual_x, y=violation__random_primaldual_y] from \datatable ;
\legend{\scriptsize greedy, \scriptsize BPM, \scriptsize LDM}
\end{axis}
\end{tikzpicture}
\end{subfigure}\\

\hspace{-0.5cm}
\begin{subfigure}[c]{0.47\columnwidth}
\centering
\begin{tikzpicture}
\begin{axis}[xlabel = \footnotesize optimality gap, xlabel near ticks, cycle list={{color=black, densely dotted, semithick}, {color=red, mark=square}, {color = blue, mark=x}}, width=1.3\linewidth, height=3.3cm, outer sep=-3pt, ymin=0,ymax=1,xmin=0 , ylabel= \footnotesize CDF, ylabel near ticks, ylabel shift=-2.pt, xlabel shift=-1.pt, legend style={at={(0.50, 0.65)}, anchor= north west}, xtick pos=left, ytick pos=left, grid, thick]
\addplot+[mark repeat={5}] table[x=optimality_loss__knn_greedy_x,y=optimality_loss__knn_greedy_y] from \datatable ;
\addplot+[mark repeat={5}] table[x=optimality_loss__knn_barrier_x, y=optimality_loss__knn_barrier_y] from \datatable ;
\addplot+[ mark repeat={5}] table[x=optimality_loss__knn_primaldual_x, y=optimality_loss__knn_primaldual_y] from \datatable ;
\legend{\scriptsize greedy, \scriptsize BPM, \scriptsize LDM}
\end{axis}
\end{tikzpicture}
\end{subfigure}
\hfill
\begin{subfigure}[c]{0.47\columnwidth}
\centering
\begin{tikzpicture}
\begin{axis}[xlabel = \footnotesize optimality gap, xlabel near ticks, cycle list={{color=black, densely dotted, semithick}, {color=red, mark=square}, {color = blue, mark=x}}, width=1.3\linewidth, height=3.3cm, outer sep=-3pt, ymin=0,ymax=1,xmin=0, ylabel= \footnotesize CDF, ylabel near ticks, ylabel shift=-2.pt, xlabel shift=-1.pt, legend style={at={(0.50, 0.65)}, anchor= north west}, xtick pos=left, ytick pos=left, grid, thick]
\addplot+[mark repeat={5}] table[x=optimality_loss__random_greedy_x,y=optimality_loss__random_greedy_y] from \datatable ;
\addplot+[mark repeat={5}] table[x=optimality_loss__random_barrier_x, y=optimality_loss__random_barrier_y] from \datatable ;
\addplot+[mark repeat={5}] table[x=optimality_loss__random_primaldual_x, y=optimality_loss__random_primaldual_y] from \datatable ;
\legend{\scriptsize greedy, \scriptsize BPM, \scriptsize LDM}
\end{axis}
\end{tikzpicture}
\end{subfigure}\\
\vspace{-12pt}
\caption{\small Optimality of OCS-mapping algorithms, using nearest neighbor (left col) and random (right col) traffic.}
\label{fig:combined_ocs_benchmark}
\vspace{-16pt}
\end{figure}
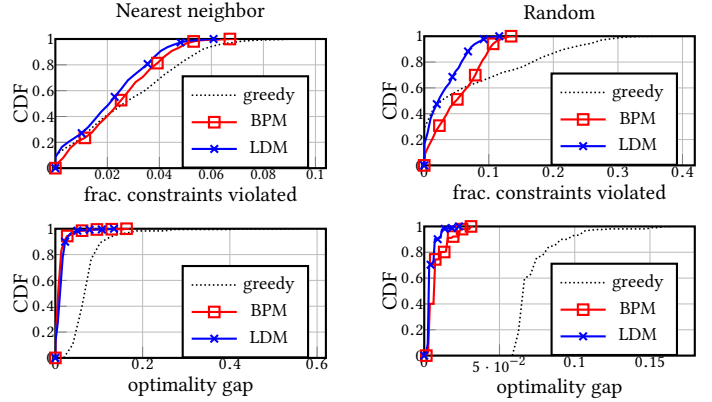

\section{OCS mapping - Optimality Analysis}\label{appendix_reconfiguration_algo_analysis}
Although LDM and BPM are motivated by convex optimization theories, our problem requires integer solutions and is thus not convex. Therefore, neither LDM nor BPM can guarantee optimality. Nevertheless, we found via simulation that LDM and BPM show superior performance.

We generated 900 DCN instances with pod-counts between 12 and 66. Each DCN instance is heterogeneous, containing pods with a mixture of 256, 512, and 1024 ports, interconnected via 128-port OCSs. The greedy algorithm described in Helios~\cite{farrington2011helios} acts as a baseline. All 900 instances are tested using: 1) nearest-neighbor, and 2) random permutation TMs. For nearest-neighbor TM, each pod sends traffic only to pods within $\rho$-units of circular index distance, where $\rho$ is $\sim\frac{1}{8}th$ the fabric size; this imitates skewed, neighbor-intensive traffic. Random TM is generated by treating each off-diagonal entry as a uniform random variable.

Next, we compute a logical topology \emph{w.r.t.} to its TM.  We use two solution-optimality metrics: 1) soft-constraint violation ratio, and 2) optimality loss. Soft-constraint violations counts the number of $(i, j)$ pairs where (\ref{con_constraint}) is violated. Optimality loss measures the throughput loss/gap as we approximate the fractional topology with an integer one. This is measured as $1 - \frac{\mu_{int}^*}{\mu_{frac}^*}$; $\mu_{int}^*$ and $\mu_{frac}^*$ denote the throughputs under the integer and fractional logical topologies. 

Fig. \ref{fig:combined_ocs_benchmark} shows LDM slightly outperforming BPM, due to its adaptability afforded by its dual variables, which help ``coerce'' the solution towards optimality. Both LDM and BPM clearly outperform the greedy method in the optimality gap and matching soft constraints.

\section{Multi-traffic Traffic Engineering (TE-M) Formulation}\label{appendix:multi_traffic_robust_te}
Given $m$ representative traffic matrices, $\{T_1, …, T_m\}$, and an integer logical topology, $X$, TE-M computes the optimal routing weights, $\omega_p \forall p \in \mathcal{P}$, that minimizes MLU for all $m$ input demands. Here, $\omega_p$ denotes the fraction of traffic sent via path $p$, such that $\sum\limits_{p \in \mathcal{P}_{ij}}  \omega_p = 1$. Rather than solving this MLU directly, however, we scale up each input traffic matrix, $T_\tau$, using $\mu_{\tau}^*$ until MLU reaches 1. This additional step ensures that the computed routing weights will account for all traffic matrices. Once the scaling factor has been computed for each input traffic matrix, we then solve for the optimal routing weights that minimizes MLU, $\eta$ for all the scaled traffic matrices with the following:

\begin{equation}
\begin{aligned}
 & \min\limits_{\Omega} \eta \\
\text{s.t. } & 1) \; \sum\limits_{p \in \mathcal{P}_{ij}} \omega_{p} = 1 \; \forall \; i, j = 1, …, n \\
& 2) \; \sum\limits_{p \in \mathcal{P}, (s_i, s_j) \text{ is a link in } p} \omega_{p} \mu_\tau^* t_{\text{src}_p, \text{dst}_p}^\tau \leq \eta x_{ij} b_{ij} \\ 
& \quad \quad \forall \; i, j = 1, …, n, \; \tau = 1, …, m,
\end{aligned}
\end{equation}
where $t_{\text{src}_p, \text{dst}_p}^{\tau}$ denotes the traffic demand (in Gbps) between the source and destination pods of path $p$.

\end{appendices}

\end{document}